\documentclass[11pt]{article}
 \pdfoutput=1
\usepackage{amsfonts,amsmath,amssymb,theorem,graphicx,color,natbib}
\usepackage{mathrsfs} 
\usepackage{multirow}
\usepackage{color}
\usepackage{caption}
\usepackage{subcaption}
\usepackage{latexsym}
\usepackage{nicefrac}
\usepackage{verbatim}

\usepackage{lineno}

\usepackage[ruled]{algorithm2e}

\newcommand\bx{\mathbf{x}}

\newcommand\MF{{\mathfrak{M}}}

\newenvironment{proof}{\paragraph{Proof:}}{\hfill$\square$\\}
  
  \renewcommand\footnotemark{}

\newtheorem{exemp}{Example}[section]

\newtheorem{theorem}{Theorem}

\begin{document} 

\title{\bf Computer code validation via mixture model estimation\thanks{Kaniav Kamary, UMR MIA-Paris, AgroParisTech, INRA, Universit\'e Paris-Saclay, 75005, Paris, France, {\sf kamary@ceremade.dauphine.fr}, Merlin Keller and C\'edric G\oe ury, EDF R\&D PRISME, Chatou, France, {\sf merlin.keller@edf.fr} and {\sf cedric.goeury@edf.fr}, Pierre Barbillon and \'Eric Parent, UMR MIA-Paris, AgroParisTech, INRA, Universit\'e Paris-Saclay, 75005, Paris, France, {\sf pierre.barbillon@agroparistech.fr} and {\sf eric.parent@agroparistech.fr}.}}

\author{{\sc Kaniav Kamary},
{\sc Merlin Keller},
{\sc Pierre Barbillon},
{\sc C\'edric G\oe{}ury},
{\sc \'Eric Parent}
}
\maketitle

\begin{abstract} 
When computer codes are used for modeling complex physical systems, their unknown parameters are tuned by calibration techniques.
A discrepancy function may be added to the computer code in order to capture its discrepancy with the real physical process.  
This discrepancy (also called code error or bias) 
is potentially caused by other inaccuracies of the computer code than the uncertainty on parameters to calibrate. 
While both parameter calibration and discrepancy are two different sources of model uncertainty, distinguishing  the  effects  of  the  two  sources can  be  challenging. 
This difficulty results in a non-identifiability problem between the discrepancy function and the code parameters. 
By considering the validation question of a computer code as a Bayesian selection model problem, \citet{GDMKPBAPEP2016} have highlighted a possible confounding effect in certain configurations between the code discrepancy and a linear computer code by using a Bayesian testing procedure based on the intrinsic Bayes factor.
In this paper,
we investigate the issue of code error identifiability by applying another Bayesian model selection technique which has been recently developed by \citet{KKKMCPRJR2014}. By embedding the competing models within an encompassing mixture model, \citet{KKKMCPRJR2014}'s method allows each observation to belong to a different mixing component, providing a more flexible inference, while remaining competitive in terms of computational cost with the intrinsic Bayesian approach. By using the technique of sharing parameters mentioned in \citet{KKKMCPRJR2014}, an improper non-informative prior can be used for some computer code parameters and we demonstrate that the resulting posterior distribution is proper.
We then check the sensitivity of our posterior estimates to the choice of the parameter prior distributions. 
We illustrate that the value of the correlation length of the discrepancy Gaussian process prior impacts the Bayesian inference of the mixture model parameters and that the model discrepancy can be identified by applying the \cite{KKKMCPRJR2014} method when the correlation length is not too small.
Eventually, the proposed method is applied on a hydraulic code in an industrial context.
\end{abstract}


\section{Introduction} 
Since field experiments are often either impractical or economically expensive, computer codes have been developed
as substitutes for many complex physical systems \citep{santner2003,fang2006}.
At the core of such codes are physical models, {\em i.e.} mathematical equations modeling a certain physical phenomenon. These are implemented into computer codes, also know as numerical codes or computer models, which solve numerically the equations defining the physical model, at a certain cost and with a certain level of approximation, depending (among other things) on the computer's internal precision. Hence, throughout the paper, we will use the term
``computer code" to denote the implementation of a physical model. We will avoid to use ``computer model'' to prevent the confusion with the statistical models in which the computer code is embedded.

The aim of numerical simulation is therefore to build such a numerical code, with the aim of simulating the physical system's behavior as realistically as possible. Two fundamental challenges need to be faced to achieve this goal:
{\em Verification} \citep{Roache1998} which is the task of ensuring that the computer code is a `sufficiently good' approximation (in a certain sense, and over a certain domain) of the physical model;
{\em Validation} \citep{AIAA1998} which is the task of ensuring that the physical model is a `sufficiently good' approximation (again, in a certain sense, and over a certain domain) of the actual physical phenomenon.
These two steps constitute the Verification \& Validation framework, \citep{AIAA1998,Roy2011}, which has become increasingly popular in engineering practice. These are often complemented by a so-called {\em calibration} step \citep{MCKAO2001}. This aims at making validation easier by reducing the uncertainty tainting some parameters of the physical model using experimental data.
Note that verification, which is a prerequisite of both calibration and validation, is beyond the scope of this work. Hence, in the following we will assume that the code is verified, and can be assimilated to the physical model it represents.
\citet{MJBJOBPRSJA2007} proposed to validate the computer code by checking the predictions are within a certain tolerance bands around the true physical process. 
They advocated for using a prediction which makes use of the computer code corrected by the discrepancy function while \citet{GDMKPBAPEP2016} considered validation as the decision favoring the computer code alone over the computer code with a discrepancy correction.   
This decision relies on Bayesian model selection methods, and it can deal with the case where the computer code depends on uncertain parameters to be calibrated. Before making this clear, we start by reviewing Bayesian model selection methods.

From a statistical point of view, a `model' is a probability distribution that reflects a set of assumptions concerning the generation of a dataset. In practice, what we expect is to obtain a model which adequately mimics the distribution from which the dataset has been produced \citep{CR2001,MJBJOBAFGGD2012}. When several potential models are available for a dataset, model selection aims at identifying the model that is most strongly supported by the data. In the context of Bayesian statistics, the most common comparison tool is the Bayes factor or its various extensions \citep{AEGAKD1994,JOBLRP1993,STBKPBNHA1997,CR2001,JMMNSPCPRJR2014}. The Bayes factor defined by the ratio of marginal likelihoods (whatever the estimated model) is obtained by a probabilist reasoning based on a loss function over competing models. 
In order to make a decision about two competing models, the Bayes factor is usually compared with the threshold value of one. In the case where more than two models are compared, the selection procedure is based on the highest posterior probability of the model given the data. 
Despite a widespread use of Bayes factor and its equivalent (posterior probability of the model), by the Bayesian community, it appears however problematic in some cases \citep[see][]{KKKMCPRJR2014}. \citet{KKKMCPRJR2014} proposed therefore a paradigm shift in Bayesian hypothesis testing which resolves many difficulties with testing via Bayes factor. Their idea is to
define an alternative to the posterior probability construction that the data originates from a specific encompassing model based on considering the models under comparison as components of a mixture
model. 
The posterior probabilities of the competing models are therefore replaced by estimation of the weights of the competing models within the mixture model. In addition to allow for an extended use of improper priors, one noticeable advantage of this approach is theoretically proof of asymptotic convergence toward true model and its natural interpretation. In this paper, we apply the model selection via mixture model estimation to study whether a pure code prediction or a discrepancy-corrected one provides the best approximation of the physical system.
Traditionally the discrepancy-corrected prediction is based on a Gaussian process modeling of the discrepancy function \citep{kennedy2001}
which results in a dependence structure between data.
However, the \cite{KKKMCPRJR2014}s method is based on the original mixture model formulation where the data are assumed to be iid.
We will show that, under some conditions, this data dependence difficulty can be bypassed which makes the \cite{KKKMCPRJR2014}'s method possible.      

The paper is organized as follows. In Section \ref{sec:models}, the two models in competitions and the mixture model are presented. Section \ref{sec:analysis} deals with the analysis of the mixture model and details the MCMC algorithm used for the inference. 
Section \ref{sec:illu} is a simulation study to illustrate the efficiency of the proposed model selection procedure and Section \ref{sec:case} deals with an application in an industrial context.

\section{Models in competition}\label{sec:models}

In the code validation context, a pure computer code or a simulator is supposed to be a parametric function $f(x, \pmb{\theta})$ aiming at substituting the physical quantity of interest denoted by $r(x)\in \mathbb{R}$ where $x\in \mathcal{X}\subset \mathbb{R}^d$ is the controllable input variable and $\pmb{\theta} \in \Theta \subset \mathbb{R}^d$ is a vector of parameters. We assume that the code
 response is a deterministic real value function of input variable $x$. 
 In order to account for systematic differences between the true value of $r(x)$ and the value predicted by $f(x, \pmb{\theta})$, \cite{MCKAO2001} introduced a model {\em bias}, or {\em discrepancy} term, defined as 
\begin{equation}\label{eq:b}
\delta(x)=r(x)-f(x, \pmb{\theta}^*).
\end{equation} 
in which $\pmb{\theta}^*$ is supposed to be the true value of $\pmb{\theta}$. $\delta(x)$ is therefore an unknown process independent from $\pmb{\theta}$ and both $\delta(x), \pmb{\theta}^*$ can be jointly estimated using the available observations and limited number of simulations \citep{MJBJOBPRSJA2007}. 
The above definition is problematic, since there are infinitely many pairs $(\pmb{\theta}, \delta(\cdot))$ verifying Equation~(\ref{eq:b}), meaning that any statistical model including both parameters is unidentifiable. This is why a prior is needed to distinguish between them, since only the sum $\delta(x) + f(x, \pmb{\theta})$ is estimable under certain conditions.

In the case where the bias is significantly different from zero, a discrepancy-corrected 
prediction is considered as the accurate estimation of $r(x)$. 
In order to perform the statistical modeling of the code validation, 
we suppose that $X$ and $Y$ are respectively the input physical design matrix of size $n \times d$ and the vector of related available 
field measurements of size $n$. We have then $X=(x_1 \quad \ldots \quad x_n)^T$ where $x_i=(x_{i1}\quad \ldots \quad x_{id}); i=1, \ldots, n$ and $Y=(y_1, \ldots, y_n)$. For each $i \in \{1, \ldots, n \}$, $y_i$ is therefore written as
$y_i=r(x_i)+\epsilon_i$ where $\epsilon_i \sim\mathcal{N}(0, \lambda^2)$. 
If $\delta$ cannot be neglected, the field measurement 
is written as $y_i=f(x_i, \pmb{\theta})+\delta(x_i)+\epsilon_i$ otherwise $y_i=f(x_i, \pmb{\theta})+\epsilon_i$ which means that the code sufficiently represents the physical system. The problem of detecting the statistical significance of $\delta$ can be therefore considered as a Bayesian model comparison problem between two following models:
\begin{align}\label{eq:2}
\MF_0: y_i&=f(x_i, \pmb{\theta}_0)+\epsilon_i^0 \nonumber\\
\MF_1: y_i&=f(x_i, \pmb{\theta}_1)+\delta(x_i)+\epsilon_i^1.
\end{align}

where $\epsilon_i^0 \sim \mathcal{N}(0, \lambda_0^2)$ and $\epsilon_i^1 \sim \mathcal{N}(0, \lambda_1^2)$. If we consider a latent variable $\zeta_i$ which takes values $0$ or $1$, then the model selection problem in \eqref{eq:2} is equivalent to test the hypotheses $\mathcal{H}_0: \zeta=0$ versus $\mathcal{H}_1: \zeta=1$ when $y_i=f(x_i, \pmb{\theta}_1)+\zeta \delta(x_i)+\epsilon_i^1$. 

Under the definition \eqref{eq:2}, if $\ell_{\MF_0}(\pmb{\theta}_0, \lambda_0;y_i, x_i)$ and $\ell_{\MF_1}(\pmb{\theta}_1, \lambda_1, \delta;y_i, x_i)$ indicate the likelihood functions of $y_i$ under the models $\MF_0$ and $\MF_1$, both models can be embedded within an encompassing mixture model \citep{KKKMCPRJR2014} as follows

\begin{equation}\label{eq:1}
\MF_\alpha: y_i\sim \alpha \left(\ell_{\MF_0}(\pmb{\theta}_0, \lambda_0;y_i, x_i)\right)+(1-\alpha)\left(\ell_{\MF_1}(\pmb{\theta}_1, \lambda_1, \delta;y_i, x_i)\right)
\end{equation}
where mixture weight, $\alpha$ belongs to $[0, 1]$. 
The model $\MF_\alpha$ has been defined under the hypothesis that the likelihood of the model $\MF_1$ in \eqref{eq:2} is conditioned on $\delta$. More precisely, $\delta$ is considered as a parameter of $\MF_1$ on which we will later place a prior distribution and that will be a posteriori estimated. This hypothesis ensures independence of $y_i$s under the model $\MF_1$ which leads to perform the mixture model \eqref{eq:1}.
The model \eqref{eq:1} interprets each observation $y_i$ as having been generated according to $\MF_\alpha$ in the special cases where $\alpha=0$ or $\alpha=1$. The Bayesian decision step proceeds by comparing the posterior of $\alpha$ to a Dirac mass at $1$. \cite{KKKMCPRJR2014} demonstrated that $\alpha$ always converges to one as the sample size increases when $\MF_0$ is the true model from which the data has been generated.
In other words, when the posterior distribution of the weight $\alpha$ is concentrated close to one, this means that model $\MF_0$ is strongly supported by the data versus $\MF_1$. 

The complexity of the posterior distribution of the mixture model parameters depends on the complexity of the code structure. In practice, the code $f()$ can be either a  simple linear function such as $f(x,\pmb{\theta})=g(x)\pmb{\theta}$ \citep{GDMKPBAPEP2016} or a complex function whose running may be time-consuming and expensive. Hence, depending on the complexity of the function $f()$, we may launch either unlimited simulation runs or only a fairly limited number of runs can be carried out. In the latter case, out of the scope of this paper, the code itself has to be considered as unknown for runs not yet performed and a common solution is to make recourse to a Gaussian process realization as a surrogate to emulate the code \citep{sacks1989}. 


\section{Bayesian analysis of the mixture model}\label{sec:analysis}
In the case where the number of the code runs can be considered as unlimited with regards to computer tune requirements, the likelihood functions under model $\MF_0$ and $\MF_1$ can be therefore written as 
\begin{align}
\ell_{\MF_0}(\pmb{\theta}_0, \lambda_0;y_i, x_i)&=\nicefrac{\exp\left(-\frac{1}{2\lambda_0^2}(y_i-f(x_i,\pmb{\theta}_0))^2 \right)}{(2\pi\lambda_0^2)^{n/2}}\nonumber\\
\ell_{\MF_1}(\pmb{\theta}_1, \lambda_1, \delta;y_i, x_i)&=\nicefrac{\exp\left(-\frac{1}{2\lambda_1^2}(y_i-f(x_i,\pmb{\theta}_1)-\delta(x_i))^2\right)}{(2\pi\lambda_1^2)^{n/2}}\,.
\end{align}
For the linear function $f(x, \pmb{\theta})=g(x)\pmb{\theta}$, the likelihoods above are defined as
 \begin{align}
\ell_{\MF_0}(\pmb{\theta}_0, \lambda_0;y_i, x_i)&=\nicefrac{\exp\left(-\frac{1}{2\lambda_0^2}(y_i-g(x_i)\pmb{\theta}_0)^2 \right)}{(2\pi\lambda_0^2)^{n/2}}\nonumber\\
\ell_{\MF_1}(\pmb{\theta}_1, \lambda_1, \delta;y_i, x_i)&=\nicefrac{\exp\left(-\frac{1}{2\lambda_1^2}(y_i-g(x_i)\pmb{\theta}_1-\delta(x_i))^2\right)}{(2\pi\lambda_1^2)^{n/2}}\,.
\end{align}
The likelihood function of $Y=(y_1, \ldots, y_n)$ under the mixture model $\MF_\alpha$ will be
\begin{equation}
\ell_{\MF_\alpha}(\pmb{\theta}_0, \lambda_0, \pmb{\theta}_1, \lambda_1, \delta;Y, X)=\prod_{i=1}^n\left(\alpha\ell_{\MF_0}(\pmb{\theta}_0, \lambda_0;y_i, x_i)+(1-\alpha )\ell_{\MF_1}(\pmb{\theta}_1, \lambda_1, \delta;y_i, x_i) \right). 
\end{equation}
While the prior distribution is a key part of Bayesian inference, posterior distribution of the unknown mixture model parameters is derived from combining the prior with the probability distribution of new data \citep{AG2002}. The following section deals with the prior choice for the parameters of the mixture model $\MF_{\alpha}$ and we focus on the case where the code is a simple linear function. 

\subsection{Prior choice}
Since the variation of the discrepancy in Equation \eqref{eq:b} depends on the variation of the inputs $x_i$s, formally, the random function, or stochastic process $\delta(X)$ has a Gaussian process prior 
\begin{equation}\label{dltp}
\delta(X)\sim \mathcal{GP}(\mu_\delta, \Sigma_{\delta}); \quad \Sigma_{\delta}=\sigma_\delta^2\text{Corr}_{\gamma_\delta}(x_i, x_{i'})
\end{equation}

If we further assume that the process is stationary, it has a known constant mean which is usually considered zero \citep{CERCKIW2006} and a covariance
function only depending on the distance between the inputs. While the computer code is determined based on the expert opinions, $\mu_\delta$ can be indeed considered as a value a priori constant and known.   
The variance covariance matrix $\Sigma_{\delta}$ is written by the product of a scale parameter $\sigma^2_\delta$ that measure the discrepancy variances and a correlation function $\text{Corr}_{\gamma_\delta}(x_i, x_{i'})$ indexed by a parameter denoted by $\gamma_\delta$. The parameter $\sigma_\delta$ expresses the overall scale of the prior and $\gamma_\delta$ may allow a different length scale on each input dimension which means that for irrelevant inputs, $\gamma_\delta$ becomes small \citep{CKIWDB1998}. Here, we suppose that the correlation function is an exponential function defined by $\text{Corr}_{\gamma_\delta}(x_i, x_{i'})=\exp\left(-\nicefrac{|x_i-x_{i'} |}{\gamma_\delta} \right)$.

\subsubsection{Informative prior}  
Since the parameters are well identified, the prior distribution choice can be done based on the available information. However, the assessment of the information that can
be included in prior distributions is a key issue in eliciting a prior \citep{AG2002}. In practice, precise determination of an exact or even a parametrized distribution for the prior is not always an easy task even in the presence of fairly precise information about the parameters or with qualified experts \citep{CR2001,IASDCGJSLCKMJR2012}. %

\subsubsection{Noninformative prior}
If we do not have information about some uncertain parameters, we resort to noninformative priors. In the case where the prior distribution of the model parameter is proper, generally the posterior distribution of the mixture model $\MF_\alpha$ is proper as well. However, the improper noninformative prior can not be used on the mixture components because of the inconsistent behavior of the resulting posterior distribution \citep{JMMCPR2007}. 
The only way to enable the noninformative priors to be used for the mixture models is to first reparameterize the mixture components towards common-meaning and shared parameters \citep{KKKMCPRJR2014}. This allows us to use the noninformative priors for the common parameters without jeopardizing the propriety of the posterior distribution. 

For instance, suppose that the parameters $\pmb{\theta}$ and $\lambda$ are made common to both models $\MF_0, \MF_1$ in \eqref{eq:2}: the mixture \eqref{eq:1} can read as 

\begin{equation}\label{eq:6}
\MF_\alpha: y_i\sim \alpha \ell_{\MF_0}(\pmb{\theta}, \lambda;y_i, x_i)+(1-\alpha)\ell_{\MF_1}(\pmb{\theta}, \lambda, \delta;y_i, x_i)
\end{equation} 

where for each $y_i$, $i=1, \ldots, n$ 
\begin{align}\label{eq:7}
\ell_{\MF_0}(\pmb{\theta}, \lambda;y_i, x_i)&=\nicefrac{\exp\left(-\frac{1}{2\lambda^2}(y_i-g(x_i)\pmb{\theta})^2 \right)}{(2\pi\lambda^2)^{n/2}}\nonumber\\
\ell_{\MF_1}(\pmb{\theta}, \lambda, \delta;y_i, x_i)&=\nicefrac{\exp\left(-\frac{1}{2\lambda^2}(y_i-g(x_i)\pmb{\theta}-\delta(x_i))^2\right)}{(2\pi\lambda^2)^{n/2}}. 
\end{align}
and the likelihood of $\MF_{\alpha}$ is therefore given by

\begin{equation}\label{eq:8}
\ell_{\MF_\alpha}=\nicefrac{1}{(2\pi\lambda^2)^{n/2}}\prod_{i=1}^n \left(\alpha \exp\left(-\frac{1}{2\lambda^2}(y_i-g(x_i)\pmb{\theta})^2 \right)+(1-\alpha)\exp\left(-\frac{1}{2\lambda^2}(y_i-g(x_i)\pmb{\theta}-\delta(x_i))^2\right)\right).
\end{equation}

In this case, the Jeffreys prior for $\pmb{\theta}, \lambda$ is defined as $\pi(\pmb{\theta}, \lambda)=\nicefrac{1}{\lambda}$ \citep{OJBVDBS2001}. In the following theorem, we establish that under some conditions,  the noninformative Jeffreys prior choice for the shared parameters of the mixture model $\MF_{\alpha}$  results in a proper posterior distribution.

\begin{theorem}\label{theorem1}
Let $g: \mathbb{R}^d \to \mathbb{R}^d; d>1$ be a finite-valued function and vector $x_1,\ldots,x_n$ such that the rank of $\{g(x_1),\ldots,g(x_n)\}$ is $d$.
The posterior distribution associated with the prior $\pi(\pmb{\theta}, \lambda)=\nicefrac{1}{\lambda}$ and with the likelihood \eqref{eq:8} is proper when
\begin{itemize}
\item for any $0<k<1$, the hyperparameter $\sigma_{\delta}^2$ of the discrepancy prior distribution is reparameterized as $\sigma_{\delta}^2=\nicefrac{\lambda^2}{k} $ and so $\Sigma_{\delta}=(\nicefrac{\lambda^2}{k})  \text{Corr}_{\gamma_\delta}$ when $\text{Corr}_{\gamma_\delta}$ is the correlation function of $\delta$. This reparameterization essentially implies that the variance of the computer code discrepancy is a priori supposed to be larger than the variance of the field measurement white noise; 
\item the mixture weight $\alpha$ has a proper beta prior $\mathcal{B}(a_0, a_0)$; 
\item $\gamma_\delta$ has a proper Beta prior $\mathcal{B}(b_1, b_2)$.
\item proper distribution is used on $k$.
\end{itemize}

\end{theorem}

\begin{proof}
See Appendix \ref{appendix}.
\end{proof}

An alternative prior choice for $\pmb{\theta}, \lambda$, can consist of a Jeffreys prior for $\pmb{\theta}$ and an improper maximum entropy prior for $\lambda$ that evidently results in a proper posterior distribution for the mixture model. 
If the entropy of the prior distribution $\pi(\lambda)$ relative to a given distribution $\pi_0(\lambda)$ is defined as 
\begin{equation*}
\mathcal{E}(\pi(\lambda))=-\int \pi(\lambda) \log(\nicefrac{\pi(\lambda)}{\pi_0(\lambda)})\mathrm{d}(\lambda),
\end{equation*}

then a maximum entropy (or minimum information \citep{JSP2010}) prior distribution $\pi^*(\lambda)=\underset{\pi}{\arg\max}\mathcal{E}(\pi(\lambda))$ is defined as 
\begin{equation}
\pi^*(\lambda)\propto \pi_0(\lambda)\exp(-\nicefrac{\beta}{\lambda}).
\end{equation} 

where $\mathbb{E}_{\lambda}(\nicefrac{1}{\lambda})=\nicefrac{1}{\beta}$ and $\beta <\infty$.This choice is basically obtained based on the maximum entropy prior proposed by \cite{JSP2010} and $\pi_0(\lambda)$ denotes a benchmark distribution against which $\pi(\lambda)$ is measured. If we consider the Jeffreys prior for $\lambda$ as $\pi_0(\lambda)=\nicefrac{1}{\lambda}$, then we obtain $\pi^*(\lambda)\propto \nicefrac{1}{\lambda}\exp(-\nicefrac{\beta}{\lambda})$. Because the posterior results obtained based on the improper prior $\pi(\pmb{\theta},\lambda)=\nicefrac{1}{\lambda}$ are similar to those obtained for maximum entropy prior, in the following section we only show the posterior inference based on $\pi(\pmb{\theta},\lambda)=\nicefrac{1}{\lambda}$.

\subsection{MCMC algorithm}\label{sec:algo}         
 For the simulation studies, \cite{KKKMCPRJR2014} recommended the implementation of a Metropolis-Hastings algorithm that generates samples from the posterior distribution over the parameter space instead of using Gibbs sampler.
 This is basically because of the convergence difficulties that the Gibbs sampling faces especially in the case of large samples. Here for the mixture model $\MF_\alpha$, the number of unknown entries is $d+n+4$. While the vector length $(\delta(x_1), \delta(x_2), \ldots, \delta(x_n))$ increases with the sample size, sampling $\delta$ based on the standard Metropolis-Hastings algorithm encounters difficulties. Essentially, Metropolis-Hastings based on independent proposal is not easy to explore such high-dimensional space. On the one hand, finding a proposal distribution scale that results in a proper proposal acceptance rate is almost impossible in the case of using a random walk Metropolis-Hastings. Our algorithm is therefore based on Metropolis-Hastings within Gibbs algorithm (noted by Metropolis-within-Gibbs) that explores the subspace at a time. The algorithm contains the following steps for simulating the model discrepancy after having randomly initialized all model parameters based on their support

 \begin{algorithm}[H]
  \SetSideCommentLeft
  \DontPrintSemicolon
  \For{t=1,\ldots,T}{  
  \begin{description}
\item[a)] 
$\delta^{(t)}$ is sampled from $\pi(\delta|\pmb{y},\bx,\pmb{\theta}^{(t-1)},\lambda^{(t-1)},k^{(t-1)},\gamma_\delta^{(t-1)},\alpha^{(t-1)})$ as follows.
\begin{description}
\item[a.1)] For $i=1, \ldots, n; j=0,1$, generate auxiliaire variable $\zeta_i^{(t)}$ from 
$$\mathbb{P}(\zeta_i=j|y_i,x_i,\delta^{(t-1)},\pmb{\theta}^{(t-1)},\lambda^{(t-1)},k^{(t-1)},\gamma_\delta^{(t-1)})\,.$$
\item[a.2)] Generate $\delta^{(t)}$ according to the conditional posterior distribution 
$$\delta^{(t)}|\pmb{y},\bx,\zeta^{(t)}=1,\pmb{\theta}^{(t-1)},\lambda^{(t-1)},k^{(t-1)},\gamma_\delta^{(t-1)},\alpha^{(t-1)}\sim\mathcal{N}_{n}(\hat{\mu}_\delta, \hat{\Sigma}_\delta)\,.$$
\end{description}
\item[b)] Generate $\pmb{\theta}^{(t)}|\pmb{y}, \bx, \pmb{\zeta}^{(t)}, \delta^{(t)}, \lambda^{(t-1)}, k^{(t-1)}, \alpha^{(t-1)} \sim\mathcal{N}_d(\hat{\mu}_{\pmb{\theta}}, \hat{\Sigma}_{\pmb{\theta}})$.
\item[c)] Generate $\lambda^{(t)}|\pmb{y}, \bx, \pmb{\zeta}^{(t)}, \delta^{(t-1)}, \pmb{\theta}^{(t)}, k^{(t-1)}, \alpha^{(t-1)} \sim \mathcal{IG}(\hat{a}_\lambda, \hat{b}_\lambda)$.
\item[d)] Generate $\alpha^{(t)}|\pmb{y}, \bx, \pmb{\zeta}^{(t)}, \delta^{(t)}, \pmb{\theta}^{(t)}, \lambda^{(t)}, k^{(t-1)}\sim \mathcal{B}eta(n-m+a_0, m+a_0)$.
\item[e)] Generate $k^{(t)}$ from a random walk Metropolis-Hastings algorithm conditionally to $(\pmb{y}, \bx, \pmb{\zeta}^{(t)}, \delta^{(t)}, \pmb{\theta}^{(t)}, \lambda^{(t)},\alpha^{(t)}, \gamma_\delta^{(t-1)})$.
\item[f)] Generate $\gamma_\delta^{(t)}$ from a random walk Metropolis-Hastings algorithm conditionally to $(\pmb{y}, \bx, \pmb{\zeta}^{(t)}, \delta^{(t)}, \pmb{\theta}^{(t)}, \lambda^{(t)}, \alpha^{(t)},k^{(t)})$.
\end{description}
  }
  \caption{Metropolis-within-Gibbs algorithm}
 \label{algo:mwg}
\end{algorithm}

Algorithm \ref{algo:mwg} sketches the Metropolis-within-Gibbs algorithm.
The details for the conditional distribution are provided below.

The auxiliary variable $\zeta_i$ associated to each observation indicates its
component. $\zeta_i$ and its conditional distribution has been defined so that to obtain more gains in efficiency compared to standard Markov chain Monte Carlo \cite{DMH1998}.
$\mathbb{P}(\zeta_i=j|.)$ is defined as 
\footnotesize \begin{equation}\label{zeta}
\mathbb{P}(\zeta_i=j|.)\propto \left(\alpha^{(t-1)}f_{\MF_0}(y_i|x_i,\pmb{\theta}^{(t-1)},\lambda^{(t-1)})\right)^{\mathbb{I}_{j=0}}\left((1-\alpha)^{(t-1)}f_{\MF_1}(y_i|x_i,\delta^{(t-1)},\pmb{\theta}^{(t-1)},\lambda^{(t-1)},k^{(t-1)},\gamma_\delta^{(t-1)})\right)^{1-\mathbb{I}_{j=0}}
\end{equation}\normalsize
where $\mathbb{I}_{j=1}$ is the indicator function, having the value 1 for $j=0$ and the value 0 otherwise. 

 In order to compute the conditional posterior distribution $\pi(\delta|.)$, we first determine the joint distribution of $(\delta, \pmb{y_m})^T$ in which 
$\delta =(\delta(x_1), \ldots, \delta(x_n))$, $\pmb{y_m}=(y_{i_1}, \ldots, y_{i_m})$ and $\pmb{x_m}=(x_{i_1}, \ldots, x_{i_m})$. $m$ is the number of observations $y_i$ for which $\zeta_i=1$ and $y_{i_1}, \ldots, y_{i_m}$ are defined by $y_{i_l}=g(x_{i_l})\pmb{\theta}+\delta(x_{i_l})$. Under $\delta$ prior distribution \eqref{dltp}, we then have

\begin{equation*}
\begin{pmatrix} 
\delta \\
\pmb{y_m}  
\end{pmatrix}\sim \mathcal{N}_{n+m}\left(\begin{pmatrix} 
\mu_{\delta} \\
g(\pmb{x_m})\pmb{\theta}+{\mu_{\delta}}_m  
\end{pmatrix}, \begin{pmatrix} 
\Sigma_{\pmb{\delta},\pmb{\delta}} & \Sigma_{\pmb{\delta},\pmb{y_m}} \\
\Sigma_{\pmb{\delta},\pmb{y_m}}^T & \Sigma_{\pmb{y_m},\pmb{y_m}} 
\end{pmatrix}\right)
\end{equation*}
in which
\begin{align*}
\pmb{x_m}&=(x_{i_1}, \ldots, x_{i_m})^T\\
{\mu_{\delta}}_m&=({\mu_{\delta}}_{i_1}, \ldots, {\mu_{\delta}}_{i_m})^T\\
[\Sigma_{\pmb{\delta},\pmb{\delta}}]_{k,k'}&=\sigma_\delta^2\text{Corr}_\delta(x_k,x_{k'})\\
[\Sigma_{\pmb{\delta},\pmb{y_m}}]_{k,i_l}&=\sigma_\delta^2\text{Corr}_\delta(x_k,x_{i_l})\\
[\Sigma_{\pmb{y_m},\pmb{y_m}}]_{i_l,i_{l'}}&=\lambda^2+\sigma_\delta^2\text{Corr}_\delta(x_{i_l},x_{i_{l'}})
\end{align*}
since
\begin{align*}
[\Sigma_{\pmb{\delta},\pmb{y_m}}]_{k,i_l}&=\text{Cov}(\delta(x_k),g(x_{i_l})\pmb{\theta}+\delta(x_{i_l})+\epsilon_i)\\
&=\text{Cov}(\delta(x_k),\delta(x_{i_l}))\\
&=\sigma_\delta^2\text{Corr}_\delta(x_k,x_{i_l}).
\end{align*}
for all $k,k'=1,\ldots,n$, $l,l'=1, \ldots, m$. The conditional distribution $\pi(\delta|\pmb{y},\bx,\zeta=1,\pmb{\theta},\lambda,k,\gamma_\delta,\alpha)$ is therefore given by

\begin{align}\label{musigdlt}
\pi(\delta &|\pmb{y},\bx,\zeta=1,\pmb{\theta},\lambda,k,\gamma_\delta,\alpha)\propto \mathcal{N}_n(\hat{\mu}_\delta, \hat{\Sigma}_\delta)\nonumber\\
\hat{\mu}_\delta&=\mu_{\delta}+\Sigma_{\pmb{\delta},\pmb{y_m}}\Sigma_{\pmb{y_m},\pmb{y_m}}^{-1}[\pmb{y_m}-(g(\pmb{x_m})\pmb{\theta}+{\mu_{\delta}}_m)]\nonumber\\
\hat{\Sigma}_\delta&=\Sigma_{\pmb{\delta},\pmb{\delta}}-\Sigma_{\pmb{\delta},\pmb{y_m}}\Sigma_{\pmb{y_m},\pmb{y_m}}^{-1}\Sigma_{\pmb{y_m},\pmb{\delta}}
\end{align}

For the parameters $\pmb{\theta}, \lambda$ and $\alpha$, the conditional posterior distributions are given as 

\begin{align}\label{G:2}
\pmb{\theta}|\pmb{y}, \bx, \pmb{\zeta}, \delta, \lambda, k, \alpha&\sim\mathcal{N}_d(\hat{\mu}_{\pmb{\theta}}, \hat{\Sigma}_{\pmb{\theta}})\nonumber\\
\hat{\mu}_{\pmb{\theta}}&=\left(g^T(\pmb{x_{n-m}})g(\pmb{x_{n-m}})+g^T(\pmb{x_{m}})g(\pmb{x_{m}}) \right)^{-1}\nonumber\\
&\times [g^T(\pmb{x_{n-m}})\pmb{y_{n-m}}+g^T(\pmb{x_{m}})\pmb{y_{m}}-g^T(\pmb{x_{m}})\delta(\pmb{x_{m}})]\nonumber\\
\hat{\Sigma}_{\pmb{\theta}}&=\lambda^2\left(g^T(\pmb{x_{n-m}})g(\pmb{x_{n-m}})+g^T(\pmb{x_{m}})g(\pmb{x_{m}}) \right)^{-1}\nonumber\\
\lambda|\pmb{y}, \bx, \pmb{\zeta}, \delta, \pmb{\theta}, k, \alpha&\sim \mathcal{IG}(\hat{a}_\lambda, \hat{b}_\lambda)\nonumber\\
\hat{a}_\lambda&=\nicefrac{(n+m-1)}{2}\nonumber\\
\hat{b}_\lambda&=\nicefrac{\left(\sum_{\zeta_i=0}(y_i-g(x_i)\pmb{\theta})^2+\sum_{\zeta_i=1}(y_i-g(x_i)\pmb{\theta}-\delta(x_i))^2+k(\delta(\pmb{x_{m}})-\mu_\delta^{(m)})^T\text{Corr}_{\pmb{x_{m}}}^{-1}(\delta(\pmb{x_{m}})-\mu_\delta^{(m)})\right)}{2}\nonumber\\
\alpha|\pmb{y}, \bx, \pmb{\zeta}, \delta, \pmb{\theta}, \lambda, k&\sim \mathcal{B}eta(n-m+a_0, m+a_0)
\end{align}
where $\pmb{y_{n-m}}=(y_{i_1}, \ldots, y_{i_{n-m}})$ and $\pmb{x_{n-m}}=(x_{i_1}, \ldots, x_{i_{n-m}})$ are $(n-m)-$components column vector of $y_i$ and $(n-m)\times d$ matrix of $x_i$ for which $\zeta_i=0$. 
Since $\delta(X)$ follows a Gaussian process centered on $n-$components vector of $\mu_\delta$ and covariance matrix $\Sigma_\delta=\nicefrac{\lambda^2}{k}\text{Corr}$, then the marginal distribution over the subset $\{ i: \zeta_i=1\}$ of $\delta(x_i)$ is a Gaussian process with the parameters $\mu_\delta^{(m)}$ and $\Sigma_{\pmb{x_{m}}}$. We denote by $\mu_\delta^{(m)}, \Sigma_{\pmb{x_{m}}}$, $m-$components column vector of $\mu_\delta^{(m)}$ and $m\times m$ covariance matrix $\Sigma_{\pmb{x_{m}}}=\nicefrac{\lambda^2}{k}\text{Corr}_{\pmb{x_{m}}}$ obtained by dropping the irrelevant variables, from the mean vector and the covariance matrix of $\delta(X)$. 

The posterior distributions of two parameters $k$ and $\gamma_\delta$ have not closed form and they are hence generated by two random walk Metropolis-Hastings steps. 
The random walk proposal scales $s_{k}, s_{\gamma_\delta}$, are determined according to a proposal scale calibration step which consists of automatically tuning the scales towards optimal acceptance rates \citep{OGRAGWRG1997,KKJELCPR2017}. Then each parameter is estimated using simulated samples from its posterior distribution by making use of the product of the full mixture density function to the prior distribution and independent proposal distribution.  

\section{Illustration}\label{sec:illu}

In this section, we examine the performance of the mixture estimation approach to differentiate between the competing models. In the following, we proceed with two examples starting with some simulated datasets. We therefore analyze some 
datasets simulated once from the models $\MF_0$ another time from $\MF_1$. Thereafter, a real-life dataset concerning an industrial case study is investigated. We will see that the Metropolis-within-Gibbs defined in the previous section produces convergent posterior simulations with acceptance rate close to the optimal one \citep{OGRAGWRG1997}.

\begin{exemp}\label{ex:1}
As a first example, we reassess some simulated datasets initially defined and analyzed by \cite{GDMKPBAPEP2016} who applied the intrinsic Bayes factor to make a decision about $\MF_0$ or $\MF_1$ and to provide the best approximation of the physical system. As mentioned before, \cite{GDMKPBAPEP2016} illustrated a strong confounding between the code calibration parameter and the model discrepancy that complicates the model selection. More precisely, the Bayes factor depends highly on the correlation length and for large $\gamma_\delta$ values, the intrinsic Bayes factor takes value close to 1 and remains in favor of $\MF_0$ even though the analyzed sample has been simulated from the discrepancy-corrected code. 

Here, we estimate the mixture model $\MF_\alpha$ defined in \eqref{eq:8} by analyzing two situations where the dependent variable $y$ is simulated according to the model $\MF_0$  and then to the model $\MF_1$. If $x=\{\nicefrac{i}{n}\}_{i=1}^n$, we define the function $g(x)$ with a degree 2 polynomial code in $x$ as $(\mathbf{1}, x, x^2)$. The true value of $\pmb{\theta}, \lambda$ and $k$ are taken equal to: $\pmb{\theta}^*=(4,1,2)^T, \lambda^*=0.1$ and  $k^*=0.1$. The correlation function of the Gaussian process prior modeling the discrepancy is constructed according to the exponential function. 

We simulate samples of size $n$ for different choices of the true value of the parameter $\gamma_\delta$, $\gamma_\delta^*$ between $0$ and $1$. We afterwards analyze the samples by supposing $\text{Corr}_{\gamma_\delta}$ to be unknown, and then assign a proper beta prior on $\gamma_\delta$, $\gamma_\delta\sim \mathcal{B}eta(b_1,b_2)$ to compute the posterior estimate of the correlation function. We then estimate the mixture distribution $\MF_\alpha$ and analyze the sensitivity of the results to the choice of $b_1$ and $b_2$. For the hyperparameter $k$, we adopt the prior distribution $\mathcal{B}eta(2,18)$ in the case where the data points are simulated from $\MF_1$. The hyperparameter $a_0$ of the mixture weight, $\alpha$, prior $\mathcal{B}eta(a_0, a_0)$ is equal to $0.5$. 

\subsection{Bayesian inference of $\MF_\alpha$ for samples simulated from $\MF_0$}
We simulate $50$ datasets of size $n=30$ from $\MF_0: y_i=g(x)\pmb{\theta}^*+\epsilon_i$. We then compute the posterior draws of the mixture model parameters by considering the following priors for the hyperparameters of $\MF_1$: $\delta\sim \mathcal{GP}(0_n, \Sigma_{\delta})$, $\gamma_\delta \sim \mathcal{B}eta(1,1)$ and  $k \sim \mathcal{B}eta(1,1)$.  
\begin{figure}[ht]
\includegraphics[width=0.45\textwidth,height=4cm]{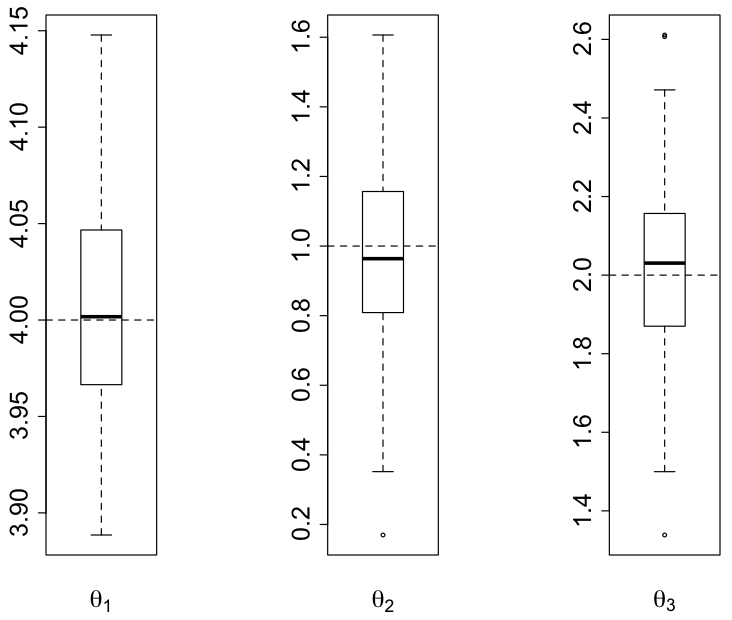}\includegraphics[width=0.45\textwidth,height=4cm]{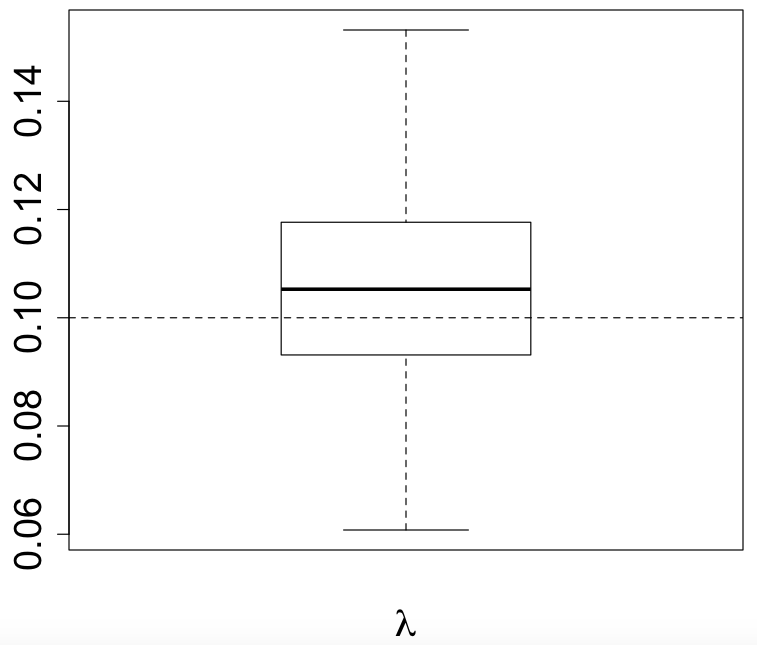}
\includegraphics[width=0.5\textwidth,height=4cm]{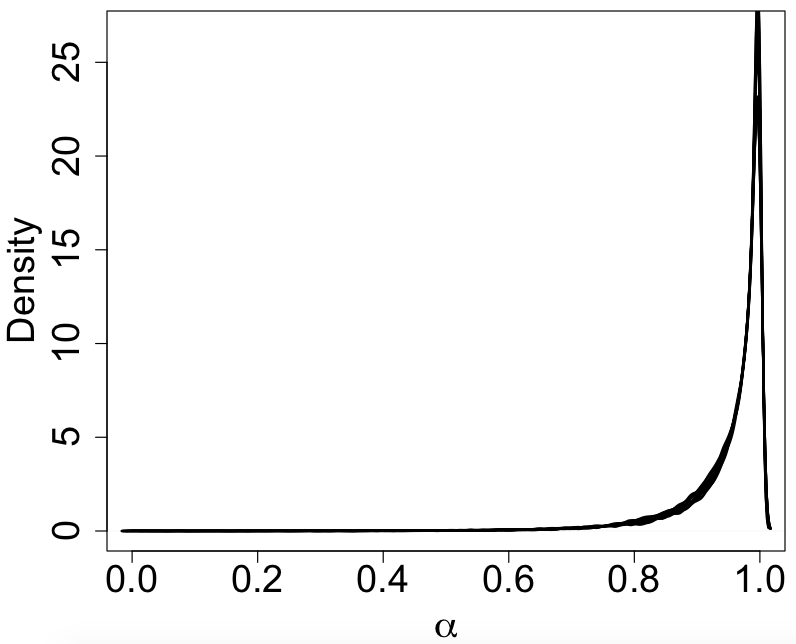}\includegraphics[width=0.5\textwidth,height=4cm]{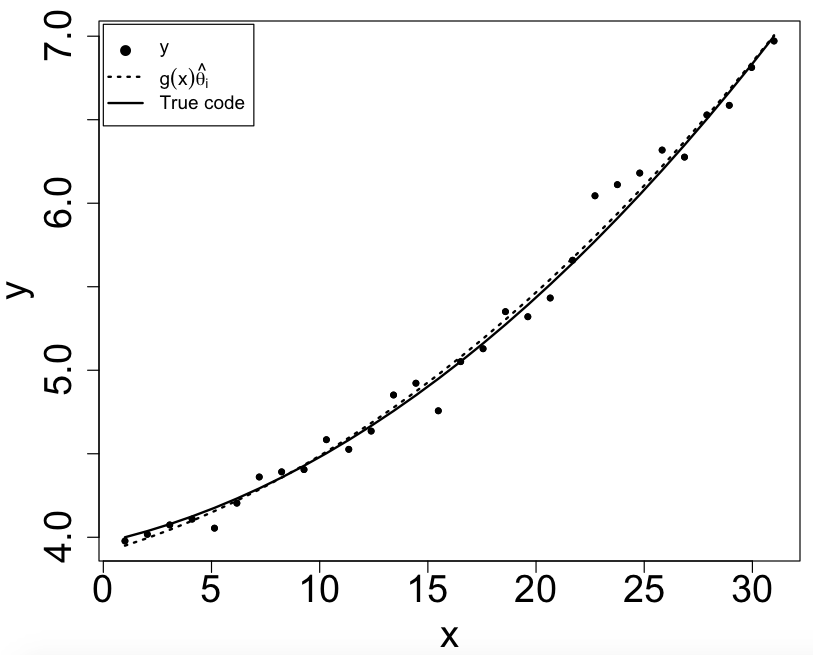}
\caption{{\bf Bayesian inference of $\MF_\alpha$ for samples simulated from $\MF_0$ (Example \ref{ex:1}):} For $50$ samples of size $30$ simulated from $\MF_0: y_i=g(x)\pmb{\theta}^*+\epsilon_i$, ({\em Top}) Boxplots of $50$ posterior point estimates of $\pmb{\theta}$ and $\lambda$ when each point estimate is computed by the sample average of the MCMC draws. The horizontal dotted lines indicate the true values considered for simulating datasets. ({\em Bottom}) On the left, $50$ empirical densities of the MCMC draws of $\alpha$ and on the right, comparison between data points $y_i$ plotted versus $x_i$ (black points), the posterior estimate of $\MF_0$ obtained by averaging over MCMC iterations and the true code (solide line).  
The number of MCMC iterations is $2\times10^4$ with a burn-in of $10^3$ iterations.}
\label{figure:m0}
\end{figure}

Figure \ref{figure:m0} displays the consistency and convergence of the posterior means of $\pmb{\theta}$ and $\lambda$ to the true parameter values for $50$ datasets.
As shown by the figure, the posterior mode of $\alpha$ (the weight of $\MF_0$ within the mixture model) is slightly close to one, indicating that the true model is supported by the data points. Furthermore, the calibrated code (the dotted line) looks very close to the true code (the solid line) from which the measurements have been simulated. Note that the calibrated code has been obtained based on the posterior mean of $\pmb{\theta}$. 

\subsection{Bayesian inference of $\MF_\alpha$ for samples simulated from $\MF_1$}
\begin{description}

\item[Sensitivity of $\lambda$ and $\alpha$ to the prior choice of the correlation length, $\gamma_\delta$ :] In order to compare the behavior of $\alpha$ and $\lambda$ for two prior choices : $\gamma_\delta\sim \mathcal{B}eta(1,1)$ and informative prior, we compute the posterior draws of the mixture model parameters by analyzing $50$ samples of size $50$ simulated from $\MF_1$ when $\gamma_\delta^*$ varies between $0.01$ and $0.9$, $\delta^*(x)\sim \mathcal{GP}(0_n, \Sigma_{\delta})$ and $ \lambda^*=0.1$. Figure \ref{figure:m1_1} illustrates  that the posterior means of $\alpha$ are close to zero when $\gamma_\delta^*\geq 0.1$. This means that the true model $\MF_1$ is supported by $\alpha$ for the data in this case. 
The increase of the means of $\alpha$ is however illustrated by the figure when $\gamma_\delta\in (0,0.1)$. On the other hand, $\lambda$ is overestimated when $\gamma_\delta<0.1$ and the more $\gamma_\delta$ increases, the more $\lambda$ tends to the true value $0.1$. For both correlation length priors, the posterior estimates of $\alpha$ and $\lambda$ have almost the same behaviors.  

\begin{figure}[ht]
\includegraphics[width=0.5\textwidth,height=4cm]{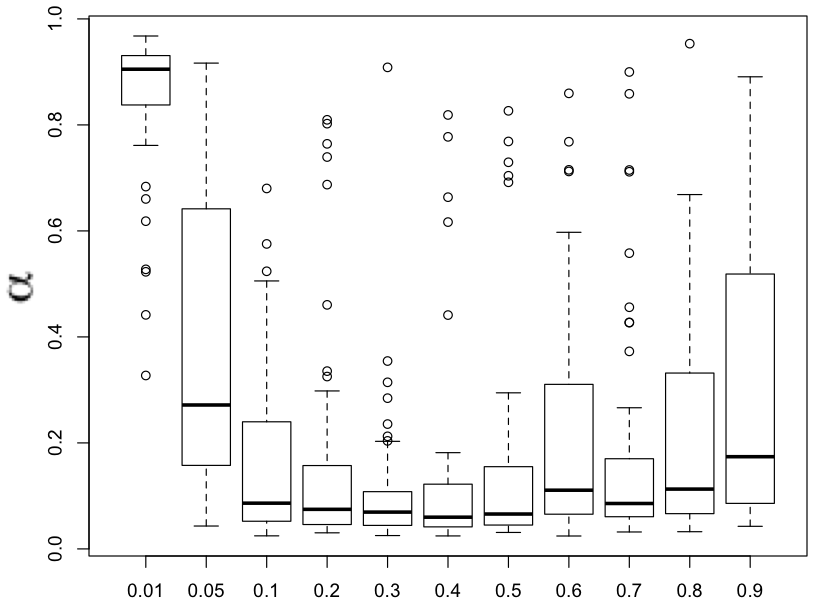}\includegraphics[width=0.5\textwidth,height=4cm]{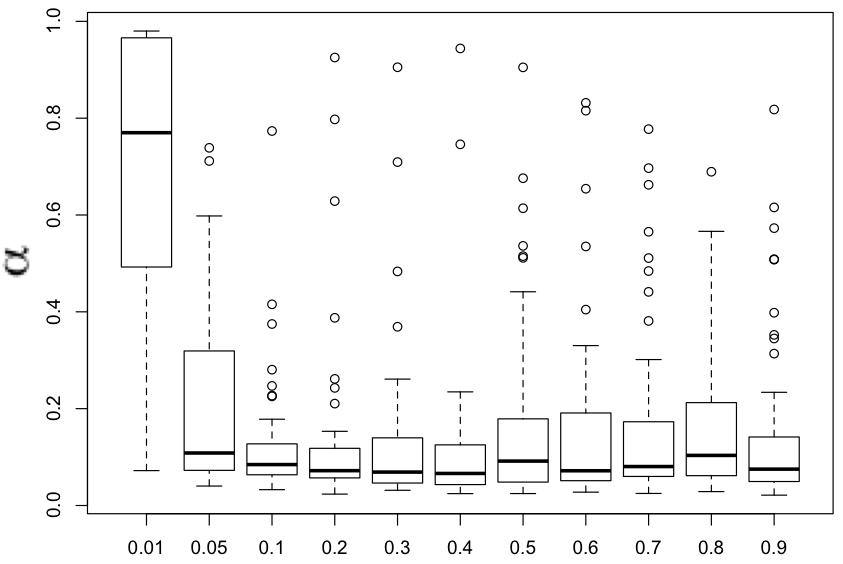}

\includegraphics[width=0.5\textwidth,height=4cm]{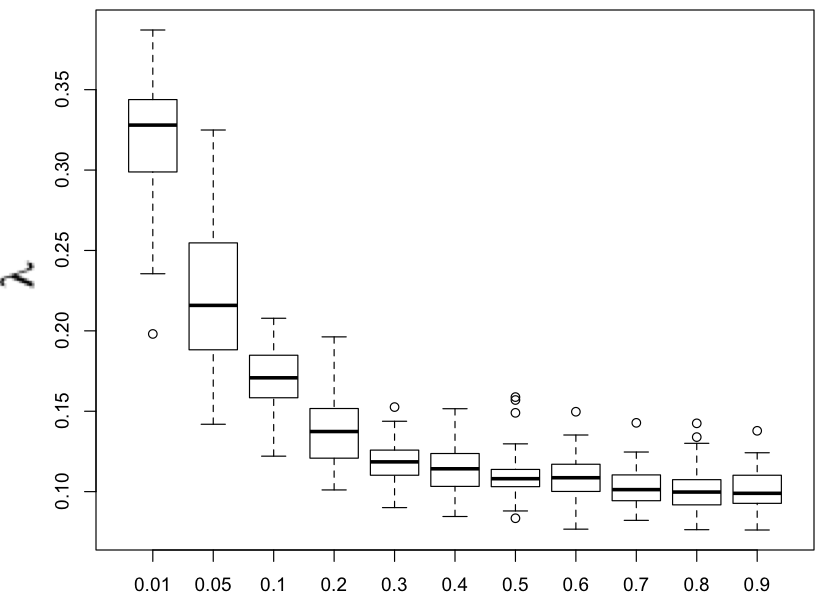}\includegraphics[width=0.5\textwidth,height=4cm]{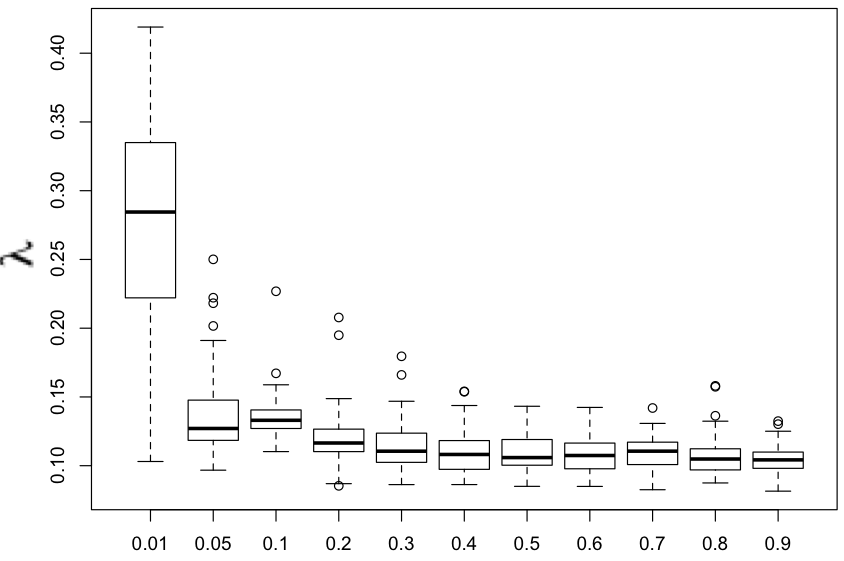}
\caption{{\bf Sensitivity of $\alpha$ and $\lambda$ to the prior choice of $\gamma_\delta$ (Example \ref{ex:1}):} For $50$ datasets of size $n=50$ simulated from $\MF_1$ when $\gamma_\delta^*\in {0.01, 0.05, 0.1, \ldots, 0.9}$, $\lambda^*=0.1$ and $\delta^*(x)\sim \mathcal{GP}(0_n, \Sigma_{\delta})$, Boxplots of $50$ posterior point estimates of ({\em Top}) $\alpha$; ({\em Bottom}) $\lambda$ when each point estimate is computed by the sample average of the MCMC draws. The prior distribution of $\gamma_\delta$ is ({\em left}) $\mathcal{U}(0,1)$; ({\em right}) informative $\mathcal{B}eta(b_1,b_2)$ with $b_1$ and $b_2$ chosen so as that the prior centers on the true value.  The number of MCMC iterations is $10^4$ with a burn-in of $10^3$ iterations.}
\label{figure:m1_1}
\end{figure}

\item[Posterior estimates of $\pmb{\theta}$ and $k$ :] For both parameters $\pmb{\theta}$ and $k$, the  distributions of the means of the $10^4$ MCMC draws obtained for $50$ datasets simulated from $\MF_1$, have been summarized on Figure \ref{figure:2}. While the posterior draws of $\pmb{\theta}$ and $k$  converge to the true value, a very slight increase is observed for $k$ when $\gamma_\delta<0.1$. 

\begin{figure}[ht]
\includegraphics[width=0.5\textwidth,height=4cm]{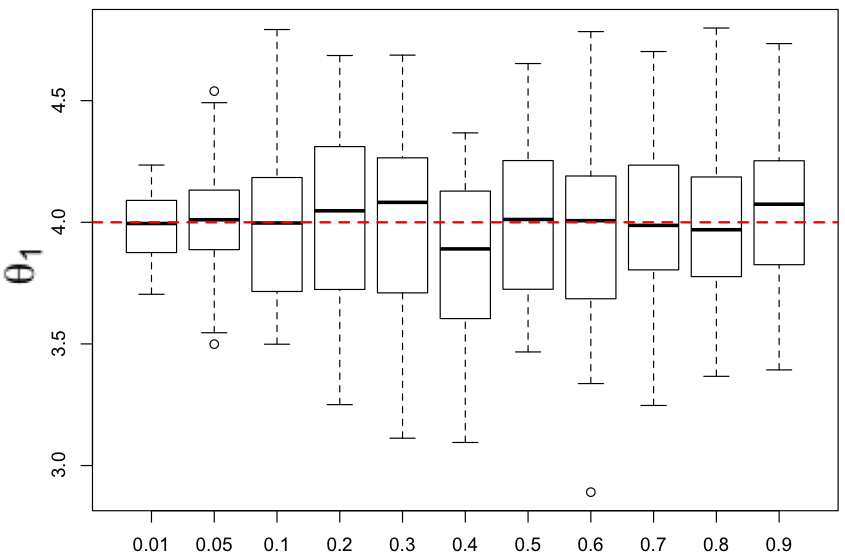}\includegraphics[width=0.5\textwidth,height=4cm]{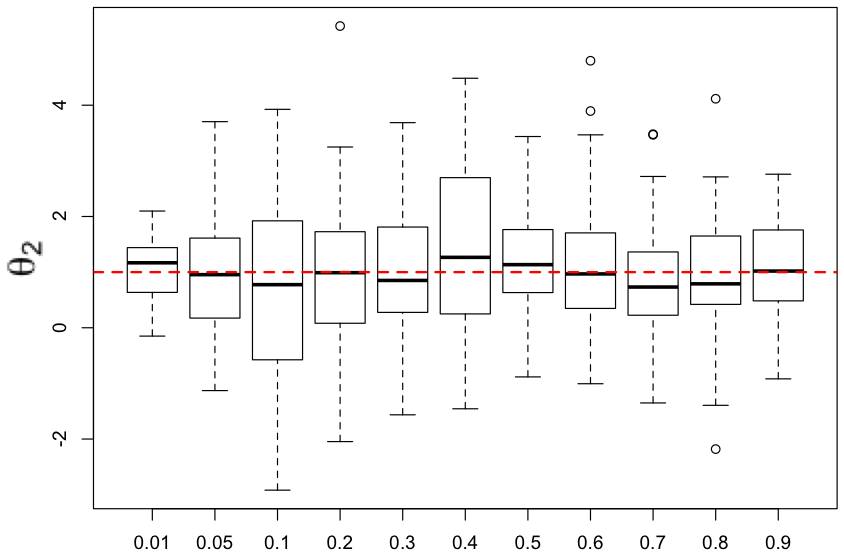}

\includegraphics[width=0.5\textwidth,height=4cm]{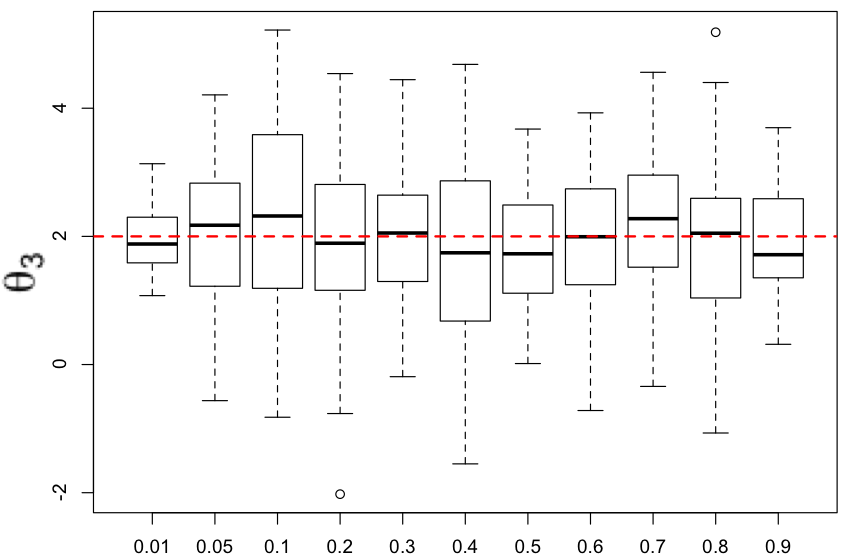}\includegraphics[width=0.5\textwidth,height=4cm]{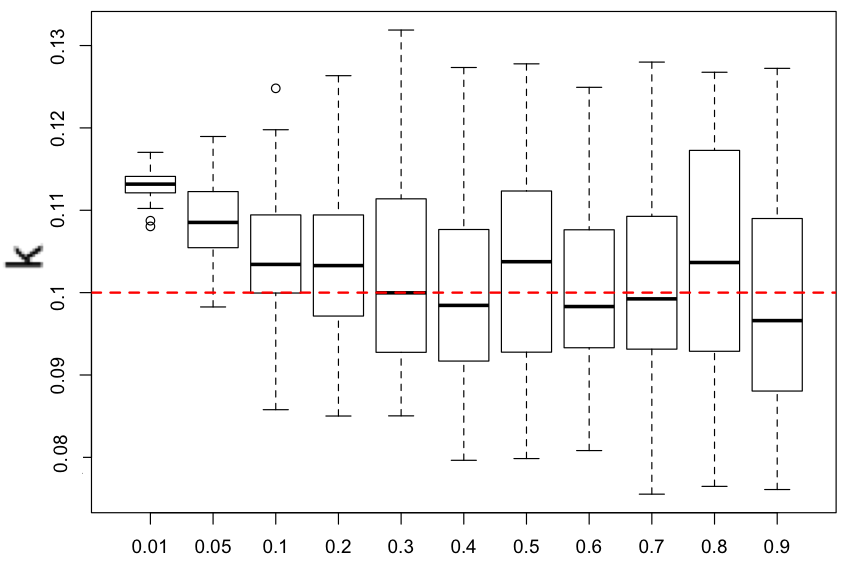}
\caption{{\bf Posterior estimates of $\pmb{\theta}$ and $k$ (Example \ref{ex:1}):} For $50$ datasets of size $n=50$ simulated from $\MF_1$ when $\gamma_\delta^*\in {0.01, 0.05, 0.1, \ldots, 0.9}$, $\lambda^*=0.1$ and $\delta^*(x)\sim \mathcal{GP}(0_n, \Sigma_{\delta})$, Boxplots of $50$ posterior point estimates of $\pmb{\theta}$ and $k$ when each point estimate is computed by the sample average of the MCMC draws.  The prior distribution of $\gamma_\delta$ is $\mathcal{U}(0,1)$. The dashed red lines indicate the true value of each parameter. The number of MCMC iterations is $10^4$ with a burn-in of $10^3$ iterations.}
\label{figure:2}
\end{figure}

\item[Evolution of $\alpha$ with the increase of the sample size :] For the samples simulated from $\MF_1$  when $\gamma_\delta^*=0.3$, $\lambda^*=0.1$ and $\delta^*(x)\sim \mathcal{GP}(0_n, \Sigma_{\delta})$, Figure \ref{figure:3} displays the range of the posterior distributions of $\alpha$ when either $n$ varies between $6$ and $100$. 
The figure demonstrates the convergence of the posterior draws of $\alpha$ to $0$, supporting the true model, as
the sample size increases. Note that the prior distribution of the correlation length is $\mathcal{U}(0,1)$. Clearly, $\alpha$ converges to zero  when the sample size is more than $n=20$.

\begin{figure}[ht]
\includegraphics[width=1\textwidth,height=4cm]{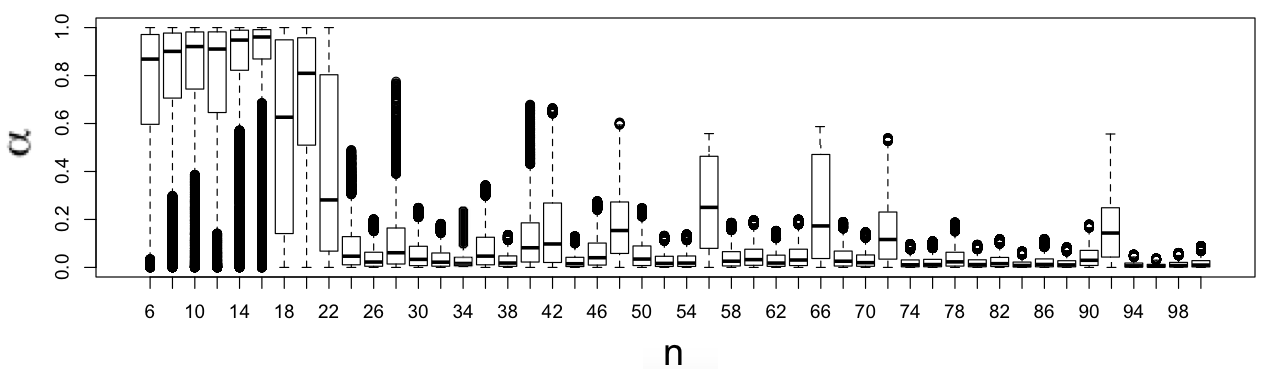}
\caption{{\bf Posterior distributions of $\alpha$ (Example \ref{ex:1}):} Boxplots of the posterior draws of $\alpha$ plotted for $50$ samples of size $n$ simulated from $\MF_1$ when $\gamma_\delta^*=0.3$, $\lambda^*=0.1$ and $\delta^*(x)\sim \mathcal{GP}(0_n, \Sigma_{\delta})$.  The prior distribution of $\gamma_\delta$ is $\mathcal{U}(0,1)$. The number of MCMC iterations is $10^4$ with a burn-in of $10^3$ iterations.}
\label{figure:3}
\end{figure}

\item[Model estimation :] For two samples of size $n=100$ simulated from $\MF_1$, Figure \ref{figure:4} displays the posterior estimates of the code $\MF_0: \hat{y}_i=g(x_i)\hat{\pmb{\theta}}$ and the code plus discrepancy $\MF_1: \hat{y}_i=g(x_i)\hat{\pmb{\theta}}+\hat{\delta}(x_i)$ when $\hat{\pmb{\theta}}$ and $\hat{\delta}(x_i)$ are the arithmetic averages of related MCMC draws. 
The figure illustrates that the calibrated code (skyblue line), model $\MF_0$, and the model with discrepancy (brown line), $\MF_1$, are very similar when the correlation length is very small $\gamma_\delta=0.01$. This leads $\alpha$ to approach both boundaries of the interval $[0,1]$ in favor of both models, $\MF_0$ and $\MF_1$. In the case where $\gamma_\delta=0.3$, $\alpha$ is very close to zero supporting the true model from which the sample has been simulated and the estimated model with discrepancy fits very well the data points.

\begin{figure}[h]
\includegraphics[width=0.6\textwidth,height=4cm]{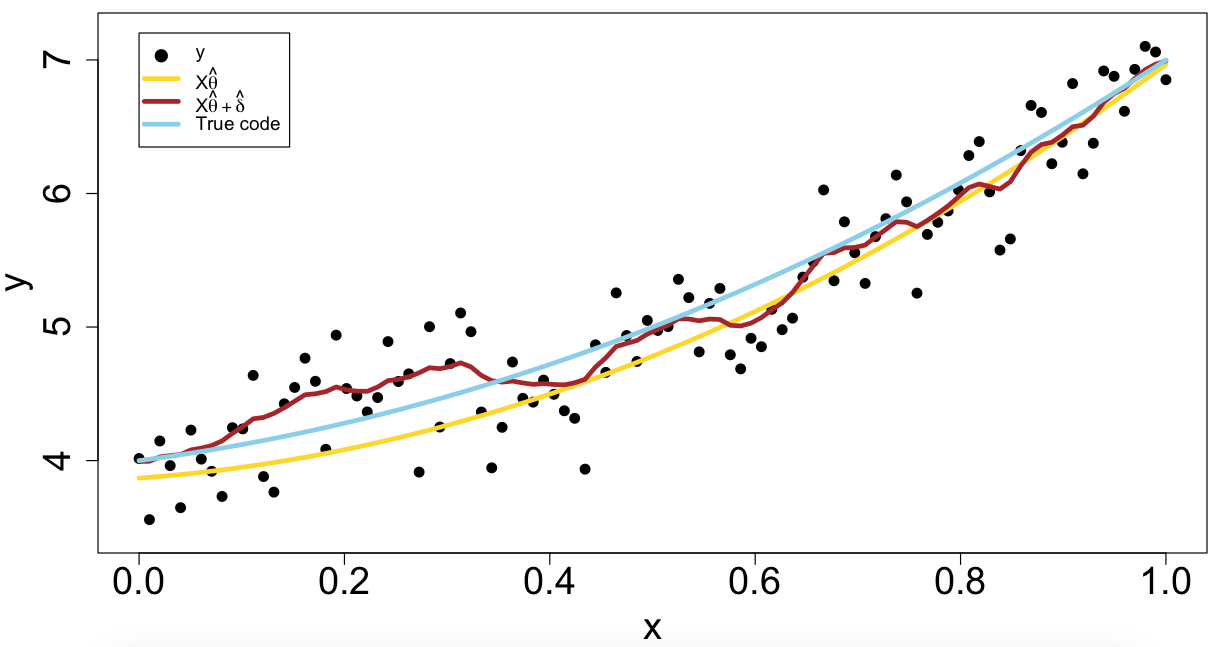}\includegraphics[width=0.4\textwidth,height=4cm]{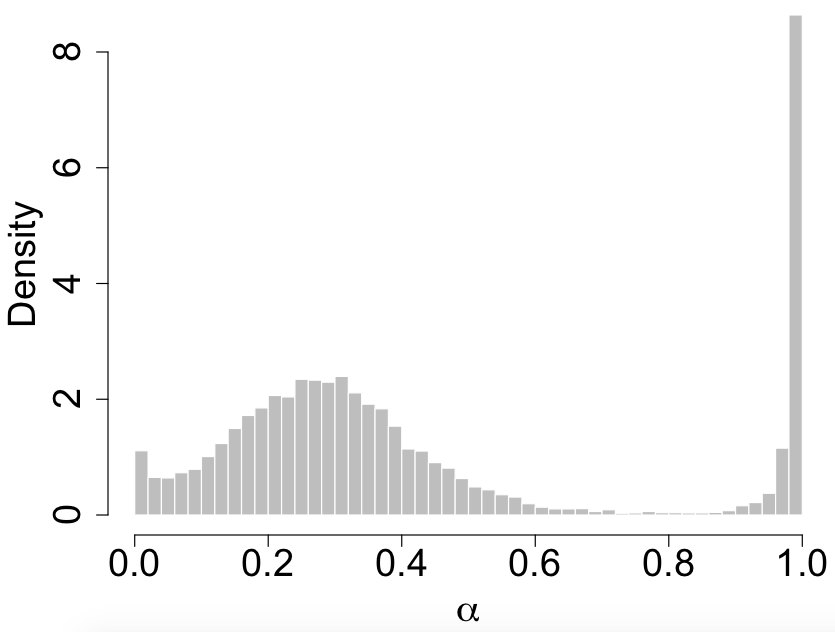}
\includegraphics[width=0.6\textwidth,height=4cm]{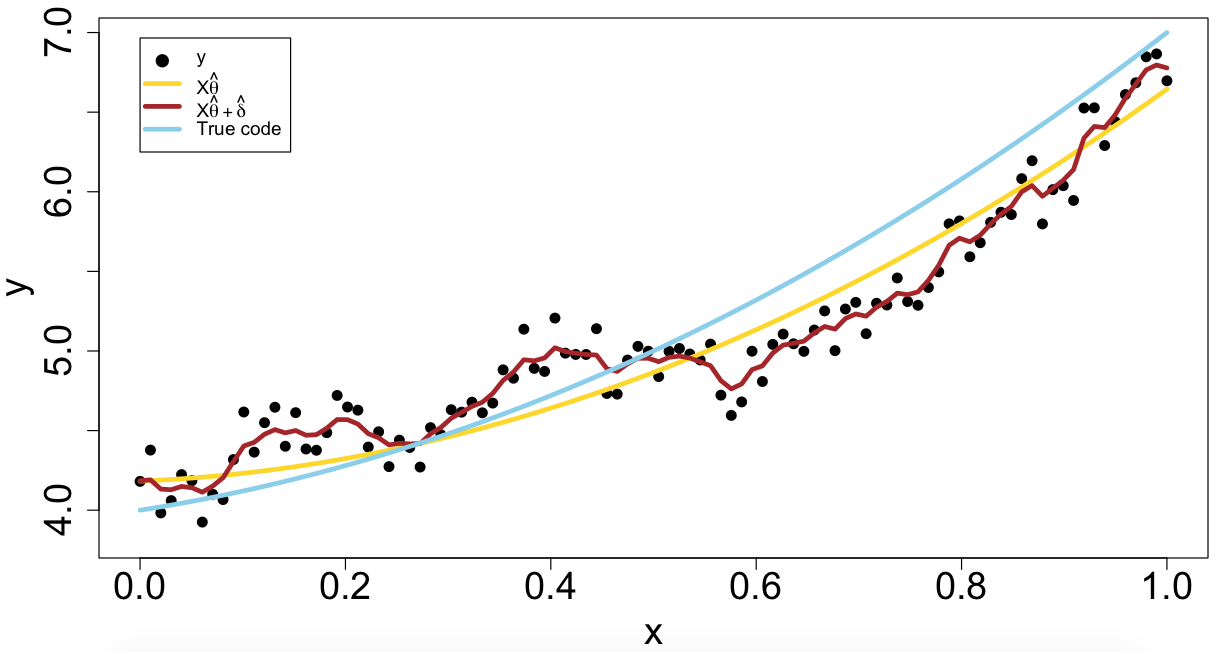}\includegraphics[width=0.4\textwidth,height=4cm]{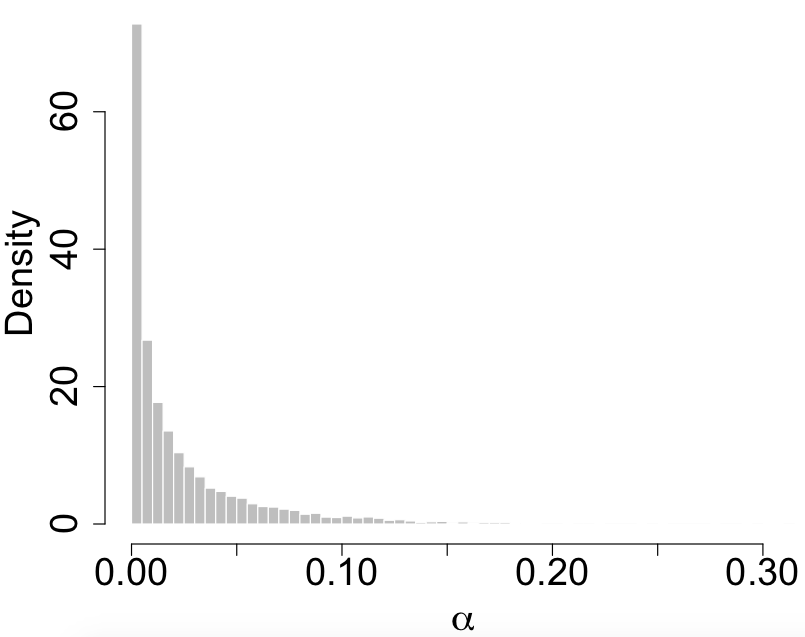}
\caption{{\bf Model estimation (Example \ref{ex:1}):} For two samples of size $100$ simulated from $\MF_1$ when ({\em Top}) $\gamma_\delta^*=0.01$; ({\em Bottom}) $\gamma_\delta^*=0.3$ and $\delta^*(x)$ is simulated from $\mathcal{GP}(0_n, \Sigma_{\delta})$: ({\em Left}) Comparison between data points $y_i$ plotted versus $x_i$ (black points), the posterior estimate of $\MF_1$ obtained by averaging over MCMC iterations (brown line), estimated code (gold line) and the true code (skyblue line). The prior distribution of $\delta(x)$ is $\mathcal{GP}(0_n, \Sigma_{\delta})$. ({\em Right}) Posterior distribution of $\alpha$, the weight of model $\MF_0$ in the mixture model. The number of MCMC iterations is $2\times10^4$ with a burn-in of $10^3$ iterations.}
\label{figure:4}
\end{figure}

\end{description}

\end{exemp}

\section{Case-study: Validation of a hydraulic model}\label{sec:case}

\subsection{Industrial context and case description}

\begin{figure}[h!]
\centering
\includegraphics[width=.4\textwidth]{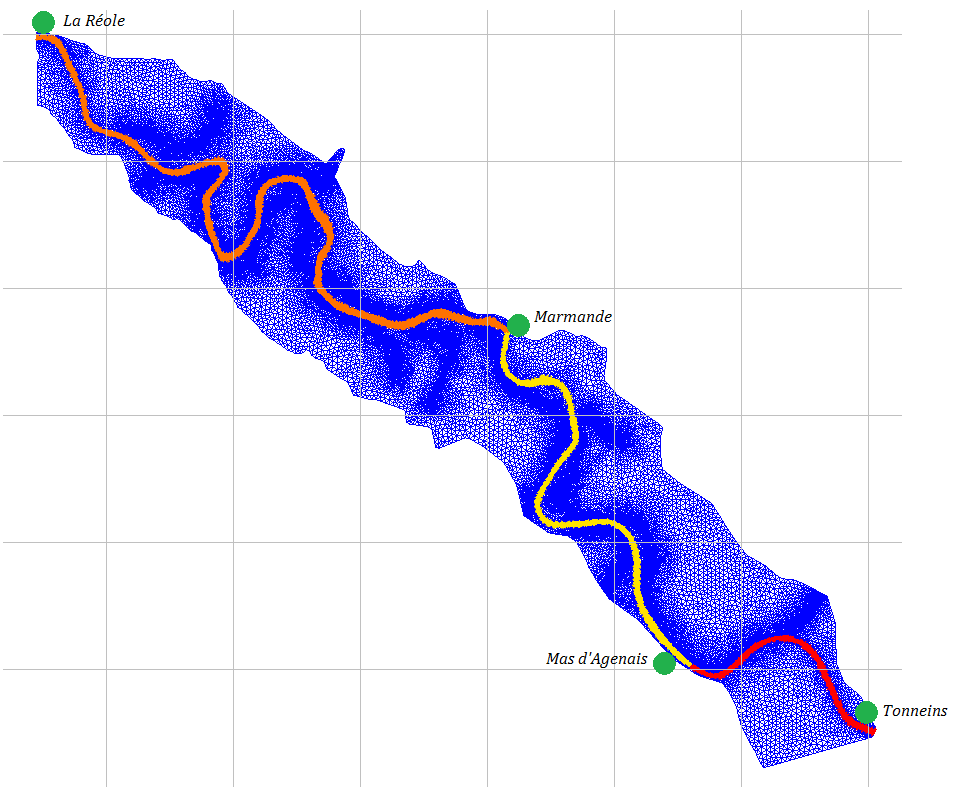}
\caption{\label{fig:Garonne}  Meshing of the major water bed for the Garonne river, between Tonneins and La R\'eole.}
\end{figure}

EDF possesses many electric power plants located nearby water sources, either for cooling purposes, or as a primary source of energy, for instance in the case of water dams. Hence, protection against floods has always been one of the main concerns driving joint work between the production, engineering and research \& developpment (R\&D) departments of EDF, toward ensuring the safety and reliability of its industrial park (see for instance the MADONE R\&D project, aimed at evaluating safety margins concerning all types of natural agressions). To this end, the TELEMAC-MASCARET hydraulic solver suite \url{http://opentelemac.org/} has been developped over several decades by the Hydraulic and Environment national Lab of EDF (LNHE), allowing to simulate complex free-surface water flows, based on a meshing of the watercourse's bed.

We study here a TELEMAC2D model of the Garonne river between Thonneins and La R\'eole. This is based on Saint-Venant's equations, solved through a finite-element approximations based on a meshing of the river bed (see Figure~\ref{fig:Garonne}). This model can be seen as a function:
\begin{eqnarray}
\bf h &=& f(q, \bf K_s),
\end{eqnarray}
where $\bf h$ is the simulated field of water heights throughout the mesh, over a certain amount of time (in practice, until it stablizes into a stable state). This depends on two different sets of inputs:
\begin{itemize}
\item a single controlled variable $q$, which is the water discharge at the entry of the river section, considered as a (fixed) boundary condition;
\item a vector $\bf {K_s}$ of five uncertain parameters, corresponding to the values taken by the Strickler coefficient, which is a measure of the river bed's smoothness, over five subdivisions of the considered river section, each subdivision being assumed homogeneous in terms of regularity.
\end{itemize}

\subsection{Data and priors}

In order to predict the future flood height distribution as accurately as possible throughout the water bassin under study, we wish to reduce the uncertainty on the Strickler coefficients through Bayesian calibration, using historical measures of discharge/water height values at Marmande and Mas-Agenais. In the following, we will note these observations $(q_{\rm s,i}, h^o_{\rm s,i})_{1\leq i\leq n_{\rm s}}$, where $\rm s$ is a site index, taking values $\rm s=MR$ for Marmande and $\rm s=MS$ for Mas Agenais, and $n_{\rm s}$ is the number of observed points at site ${\rm s}$. We will also use the vector notation $({\bf q}_{\rm s}, {\bf h}_{\rm s}^o)$ for simplicity.

Prior information about the Strickler coefficients values are available under the form of lower and upper bounds $(\bf a, \bf b)$, such that $\bf {K_s} \in [\bf a, \bf b]$, where inclusion is to be taken componentwise. We translate this into a uniform prior distribution, though a more refined robust inference approach could also be used in this context, see \cite{RiosInsua00}.

The hydraulic model's behavior throughout this rectangular prior domain was explored by means of a latin hypercube sample (LHS) comprising $N=787$ points. Calculations, amounting to a grand total of several thousand hours computation time, were performed on an EDF R\&D scientific computing cluster.

\subsection{Code emulation and linearization}

Based on the framework of \cite{Bachoc2014}, we sought a linear approximation to the code with respect to the uncertain parameters $\bf K_s$, in the neighborhood of a `well-chosen' (in a sense to be made clear shortly) point $\widehat {\bf K}_{\bf s}$, according to:
\begin{eqnarray}\label{eq:linearization}
{\bf h_{\rm s}} := f(\bf{ q_{\rm s}}; \bf K_s) &\approx& f(\bf{ q_{\rm s}}; \widehat {\bf K}_{\bf s}) + J_{\rm s} \cdot (\bf K_s - \widehat {\bf K}_{\bf s}),
\end{eqnarray}
where $J_{\rm s} := \nabla_{\bf K_s}f(\bf{ q_{\rm s}}; \widehat {\bf K}_{\bf s})$ is the $n_{\rm s}\times 5$ matrix of partial derivatives of $f$ with respect to each component of $\bf K_s$ at point $\widehat {\bf K}_{\bf s}$ and for all observed input values ${\bf q_{\rm s}} = (q_{\rm s,1},\ldots,q_{\rm s,n_{\rm s}})$.

In practice, to build this approximation, we first fitted a Gaussian process emulator to both outputs (Marmande and Mas Agenais) of the above LHS numerical design, using the scikit-learn Python library \url{http://scikit-learn.org/}. The Mat\'ern 1.5 covariance kernel gave the best results in both cases, with excellent leave-one out predictive scores (Q2) of $98.8\%$ in Mas Agenais and $99.9\%$ in Marmande. 
This justified substituting the Gaussian process mean prediction $m_{\rm s}(\bf K_s)$ to the actual code output $f(\bf{ q_{\rm s}}; \bf K_{\bf s})$ in each site $\rm s$. Hence, the covariance matrix $c_{\rm s}(\bf K_s, \bf K_s')$, quantifying the emulator's uncertainty, was not used here.

Next, we chose to define the reference point $\widehat {\bf K}_{\bf s}$ as the ordinary least-squares estimate~:
\begin{eqnarray*}
\widehat {\bf K}_{\bf s} &=& \arg\min_{\bf K_s} \sum_{\rm s \in \{{\rm MR},{\rm MS}\}} ||{\bf h}_{\rm s}^o - m_{\rm s}(\bf K_s)||^2.
\end{eqnarray*}
We then linearized the Gaussian process mean $m_{\rm s}(\bf K_s)$ around $\widehat {\bf K}_{\bf s}$ according to Equation~(\ref{eq:linearization}), and using finite differences to evaluate the partial derivatives. We observed that the hydraulic model emulator varies linearly around the OLS estimate in both sites with respect to the Strickler coefficients, and that, as shown in Figure~\ref{fig:linearization}:
\begin{itemize}
\item the water height at Marmande seems mainly influenced by ${K_s}_3$
\item the water height at Mas Agenais seems mainly influenced by ${K_s}_4$, though it is less clear than for Marmande.
\end{itemize}

In view of these observations, we decided to estimate ${K_s}_3$ using the data from Marmande, and ${K_s}_4$ using the data from Mas Agenais.

 \begin{figure}[h]
 \centering
 \includegraphics[height=0.2\textheight]{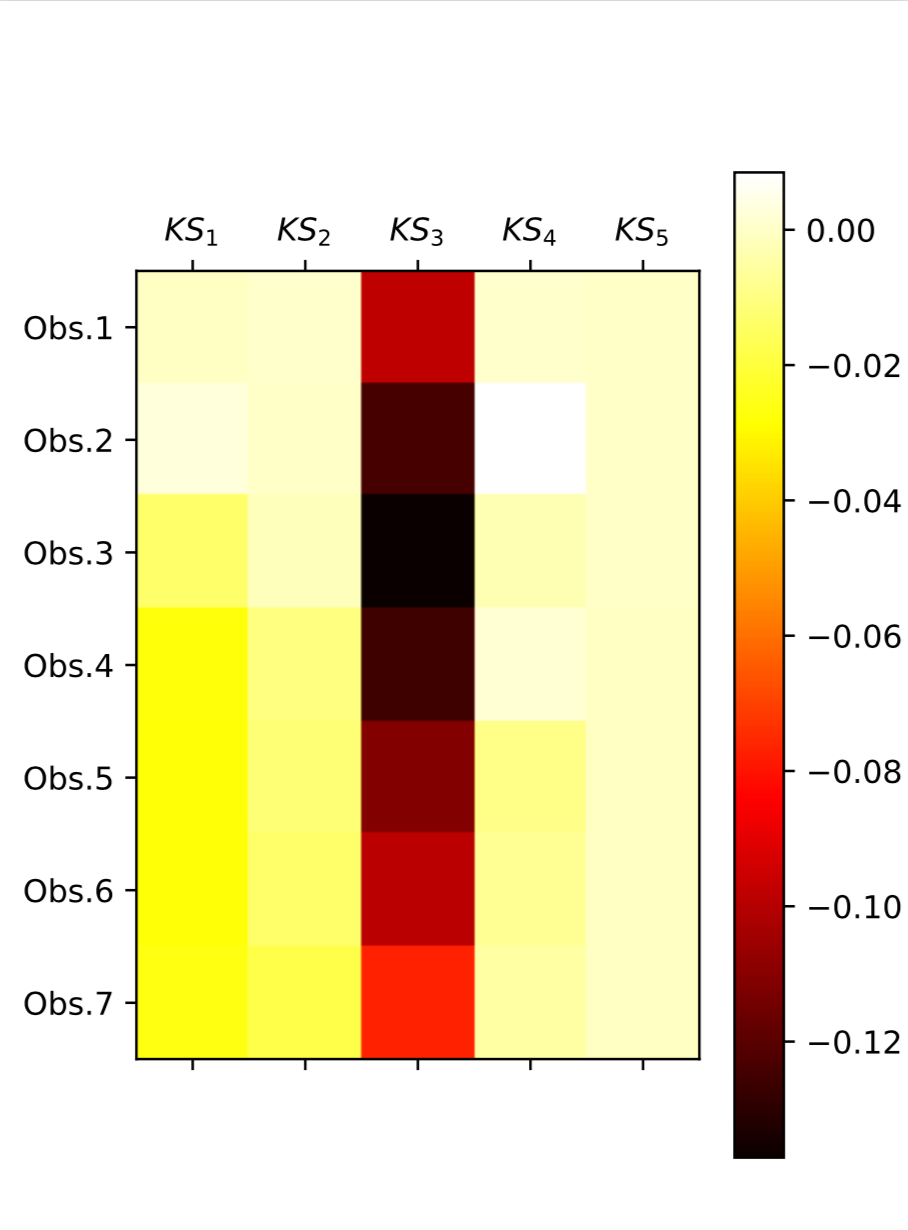}
 \includegraphics[height=0.2\textheight]{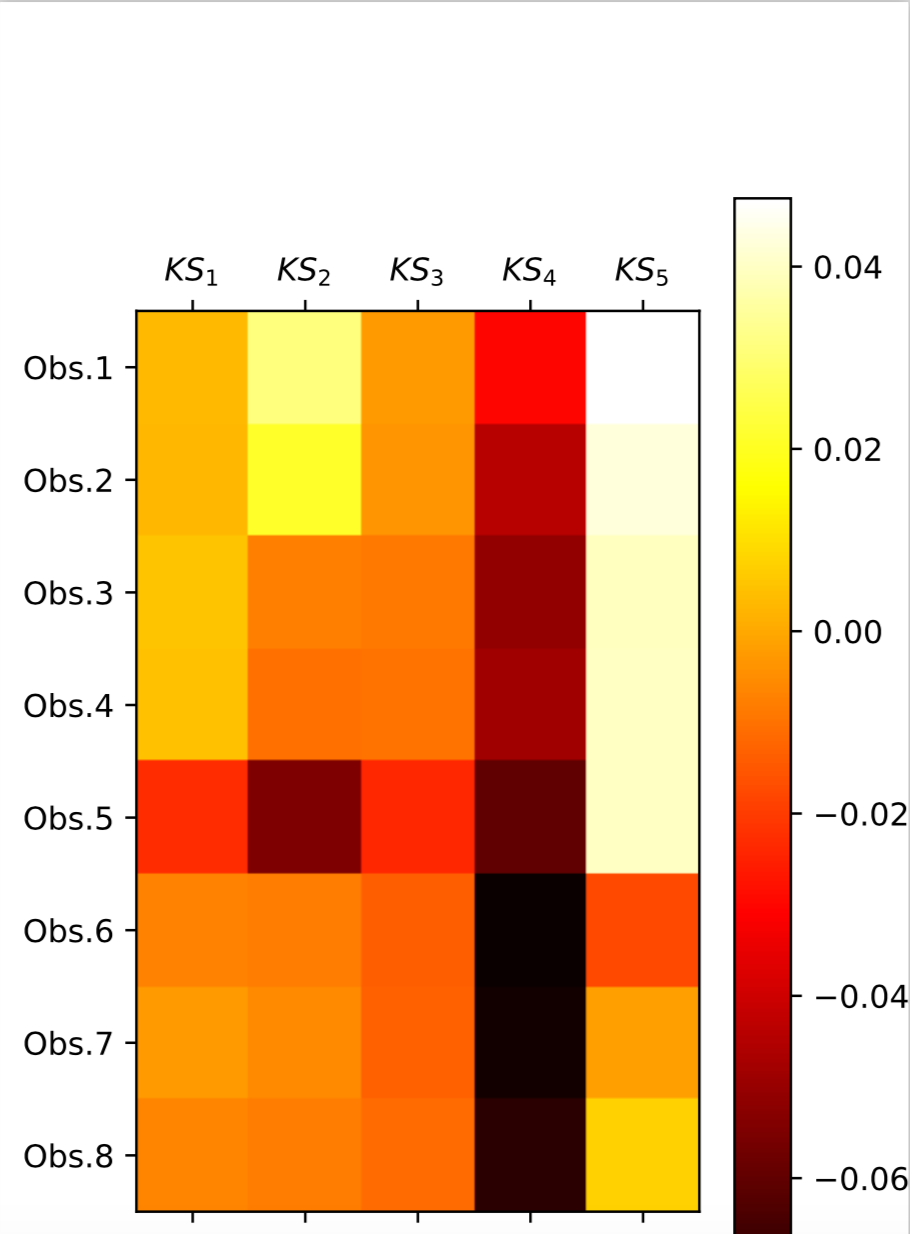}\\
 \caption{\label{fig:linearization} Jacobian matrices of the Gaussian process emulator with respect to the Strickler coefficients, in the neighborhood of the OLS estimate, at Marmande (top) and Mas Agenais (bottom).}
 \end{figure}

\subsection{Code calibration / validation}

Following the results of the sensitivity analysis described in the previous section, we used the mixture model approach to estimate:
\begin{itemize}
\item the ${K_s}_3$ STrickler coefficient by fitting the linearized code to the Marmande dataset, considering the other parameters as fixed;
\item the ${K_s}_4$ Strickler coefficient by fitting the linearized code to the Mas Agenais dataset, considering the other parameters as fixed.
\end{itemize}

More specifically, we used the MCMC algorithm~\ref{algo:mwg} to generate $N=10^4$ posterior draws from the mixture model parameters in each case, using the following priors for the hyperparameters of $\MF_1$: $\delta\sim \mathcal{GP}(0_n, \Sigma_{\delta})$, where $\Sigma_{\delta}$ is an exponential covariance function, $\gamma_\delta \sim \mathcal{B}eta(1,1)$,  $k \sim \mathcal{B}eta(1,1)$ and $\alpha\sim \mathcal{B}eta(1,1)$ for the proportion of biased code predictions. This last prior differs from the $\mathcal{B}eta(0.5,0.5)$ distribution advocated in \cite{KKKMCPRJR2014} and used in the numerical experiments. Indeed, this specific prior is well-suited for large datasets and when we suspect that all observations come from the same class of the mixture. In contrast, the present case study concerns a limited dataset, in which we have no reason to believe that code predictions are valid for either all observations or none simultaneously; hence a uniform prior seems to make more sense. 

\subsection{Results for the Marmande and Mas Agenais datasets}

\paragraph{Parameter estimation.}

The posterior densities of are shown in Figure~\ref{fig:posterior_densities} . As can be seen, they are all significantly more concentrated then the priors they derive from, showing that, even though the available data is limited, it does contain enough information to learn the mixture model. Furthermore, it note that the Strickler coefficient is better estimated for Marmance than Mas Agenais. This is consistent with the fact that the proportion $\alpha$ of observations well explained by the code is on the average higher in Marmande than in Mas Agenais (middle), which implies that the data provide more information on the Strickler coefficients in the first site. In contrast, the estimated standard deviation of measurement errors is very similar in both sites (second from left), with posterior mean below $1$~meter, which is consistent with our knowldege about water height measure procedures.

\begin{figure}[h]
\centering
\includegraphics[width=0.2\textwidth]{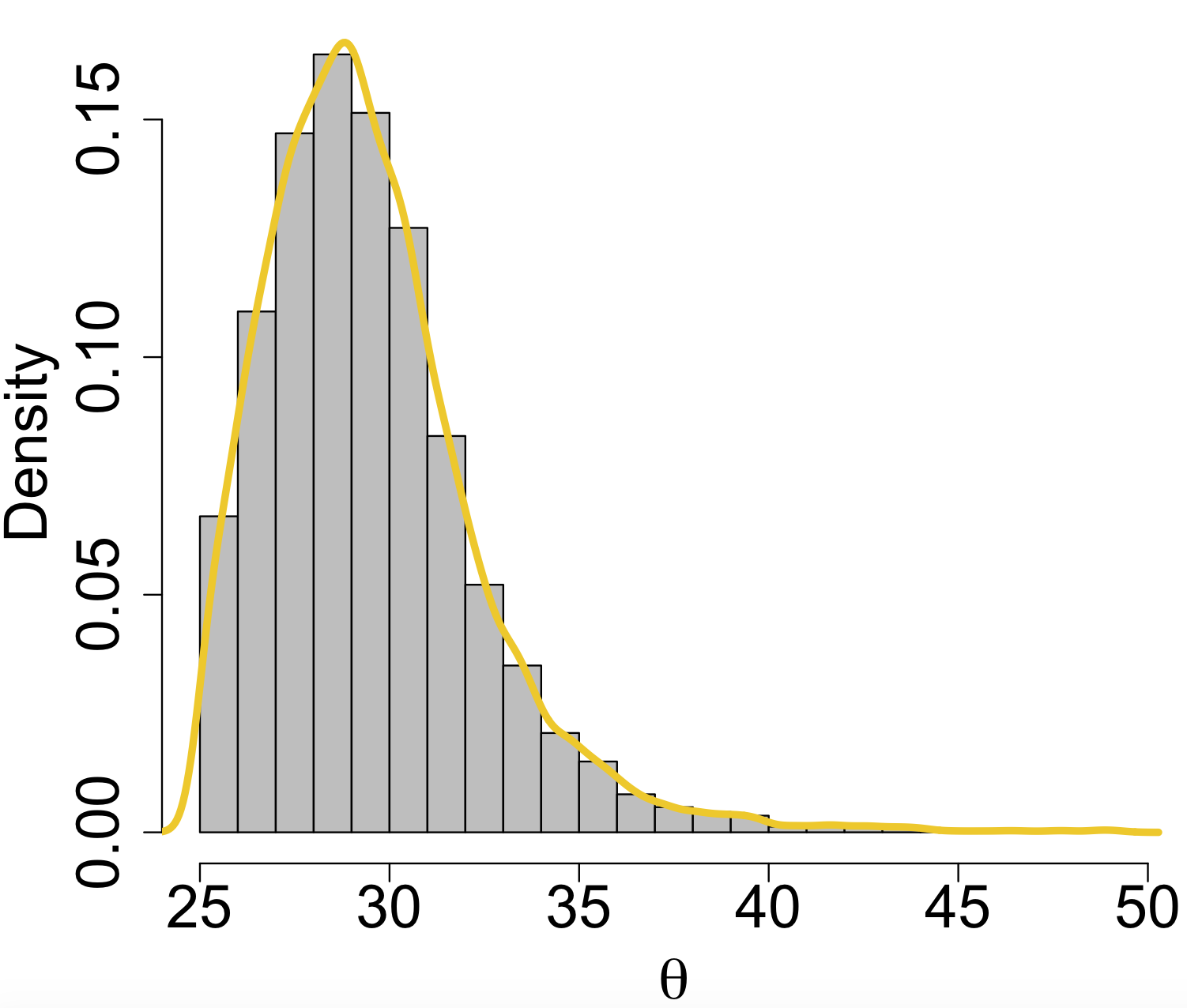}
\includegraphics[width=0.2\textwidth]{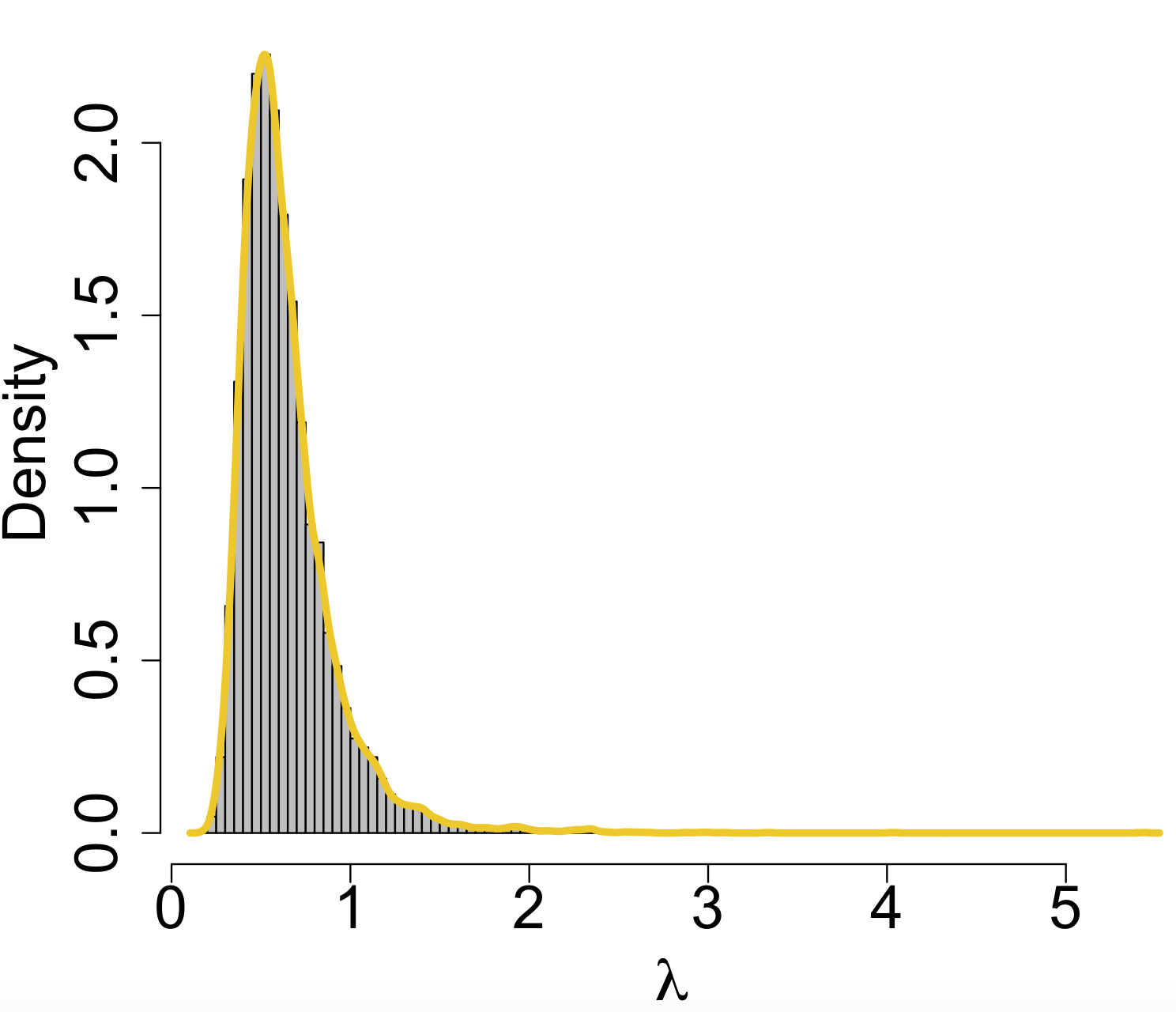}
\includegraphics[width=0.2\textwidth]{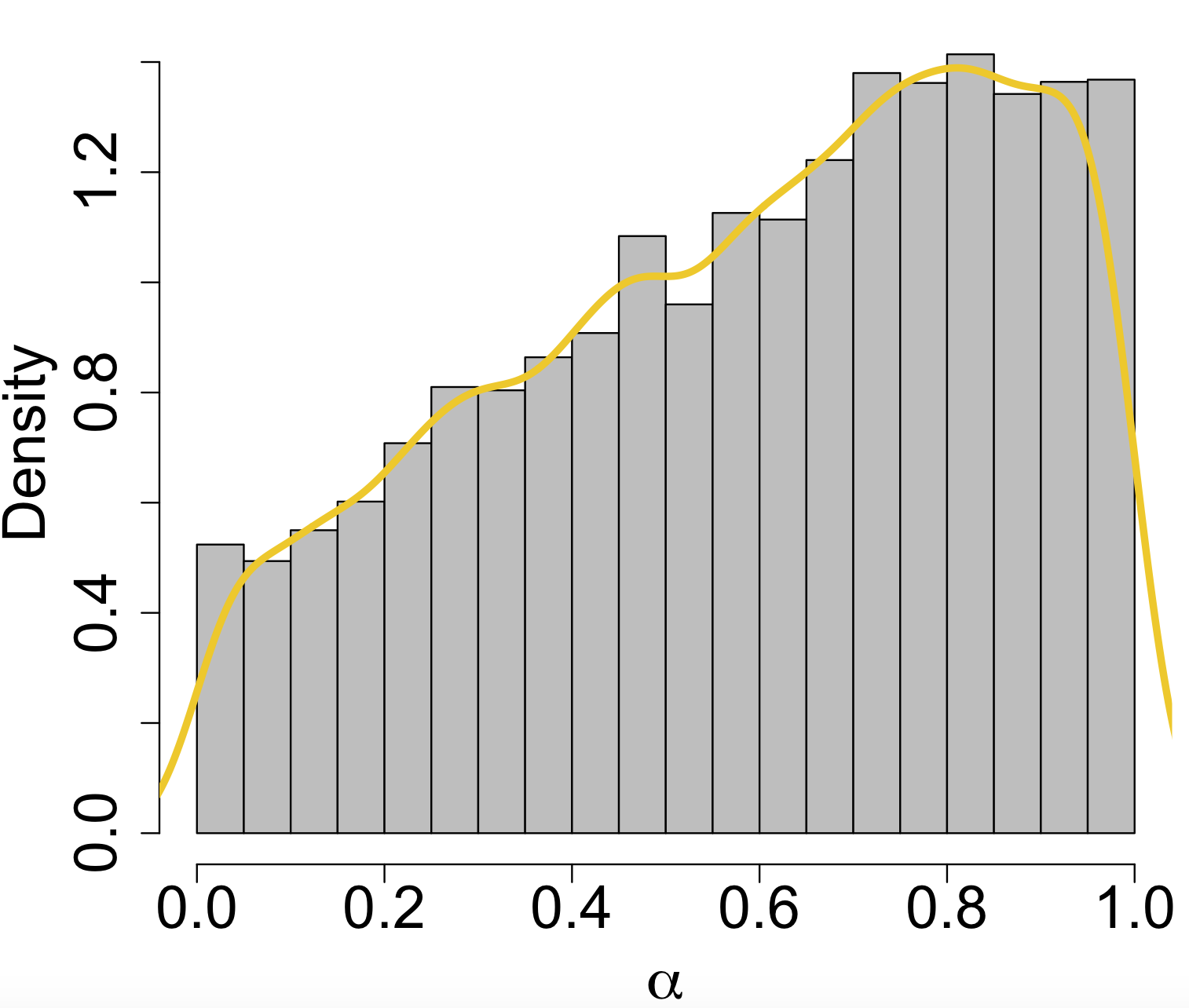}
\includegraphics[width=0.2\textwidth]{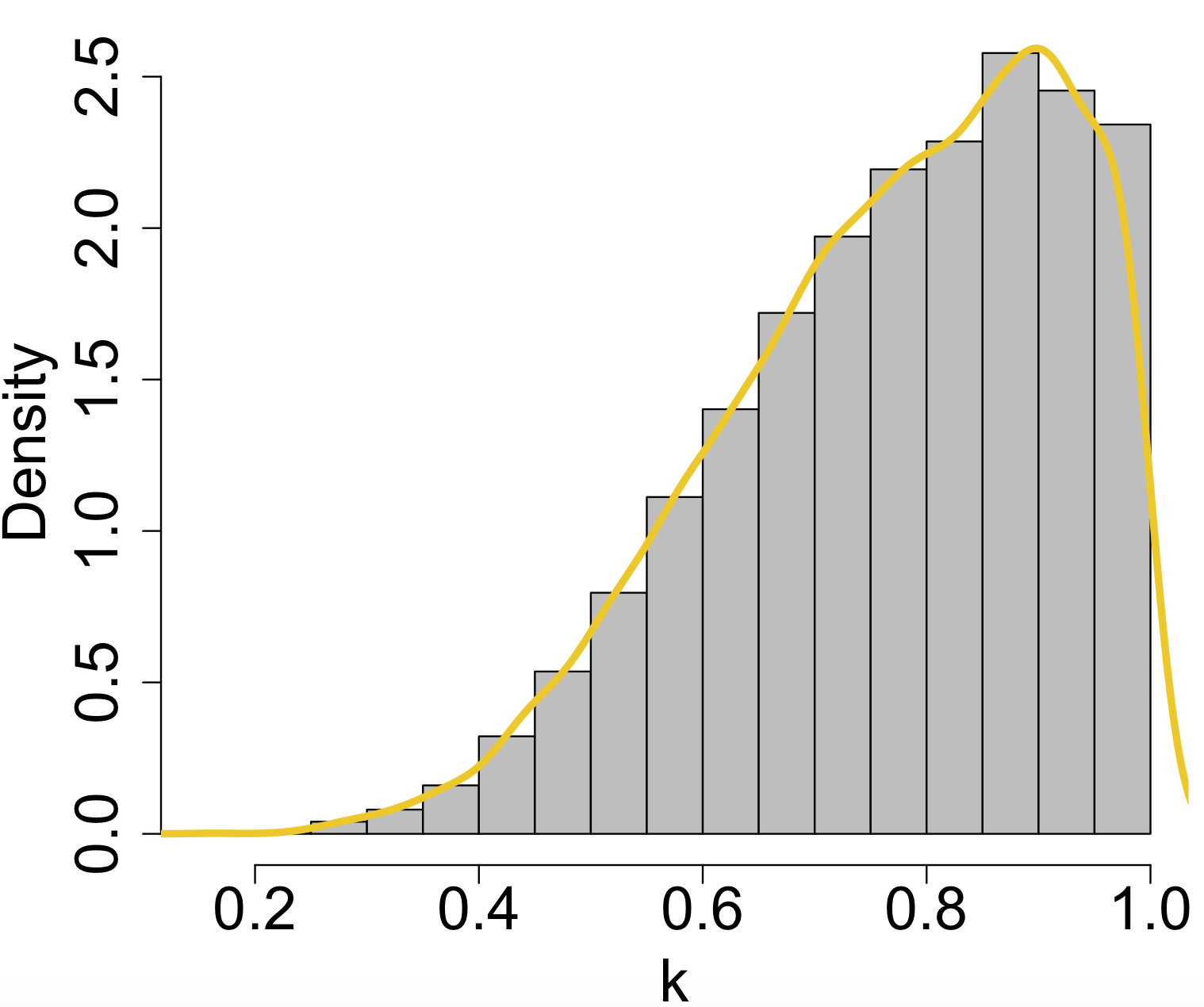}
\includegraphics[width=0.2\textwidth]{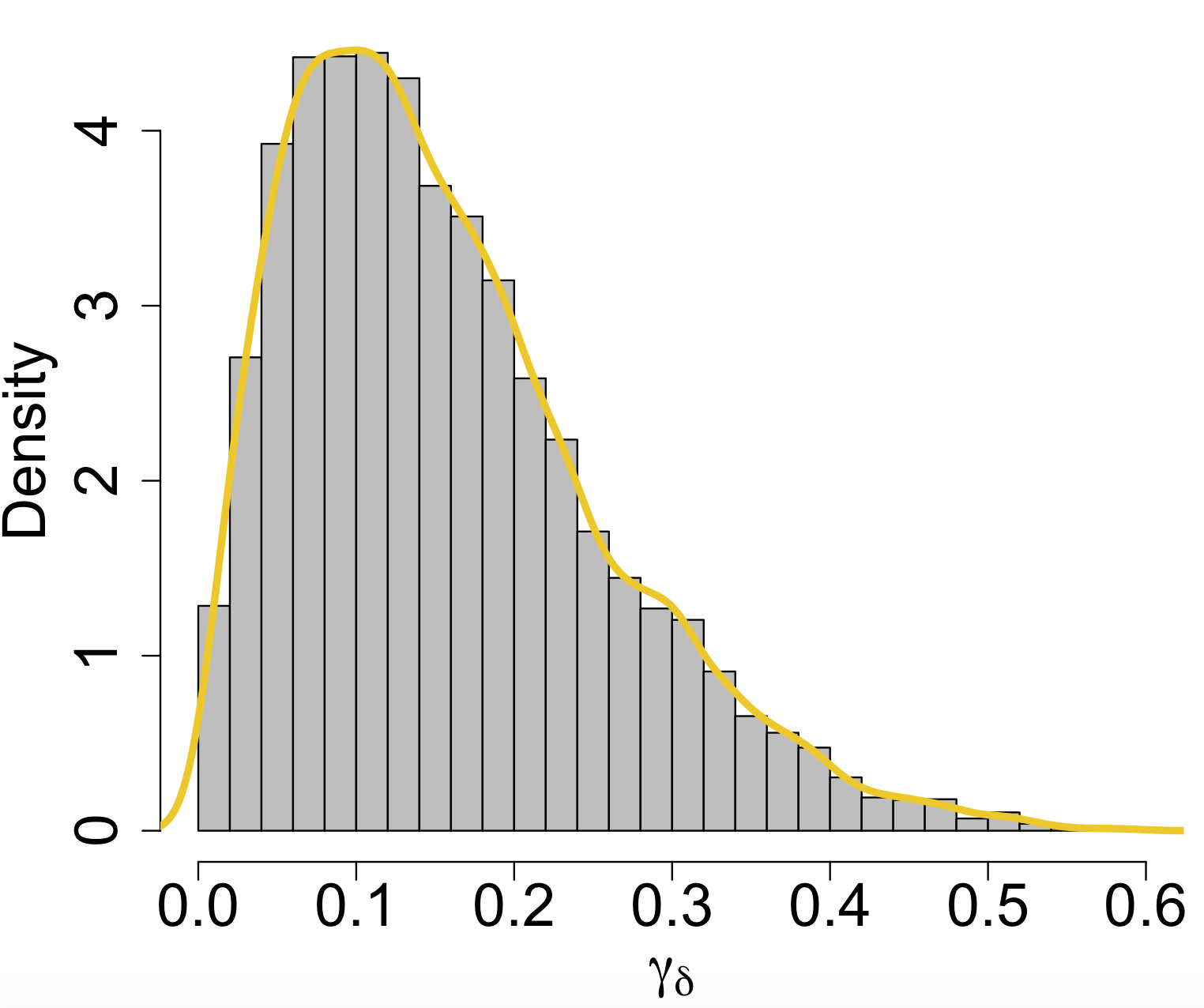}\\

\includegraphics[width=0.2\textwidth]{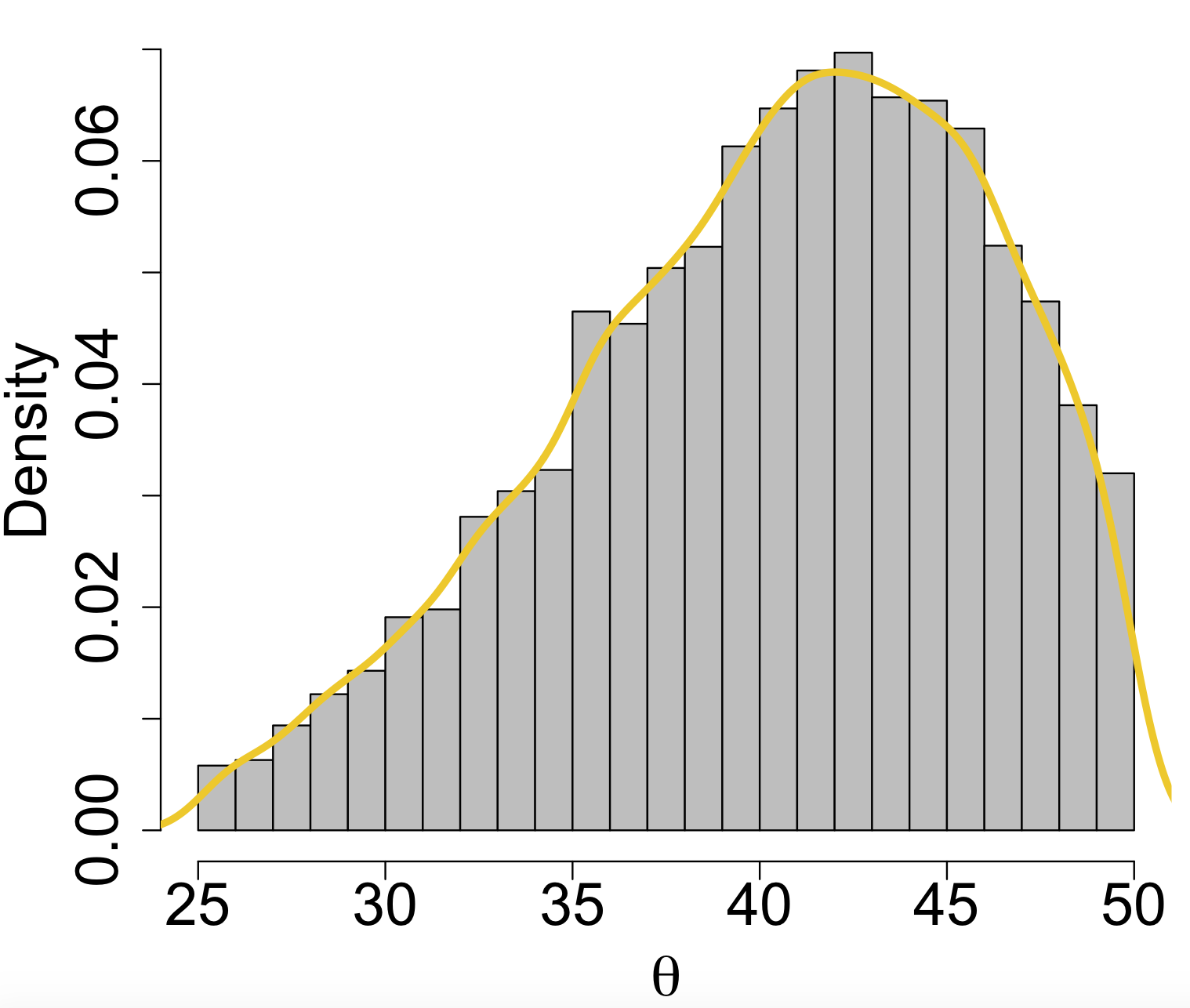}
\includegraphics[width=0.2\textwidth]{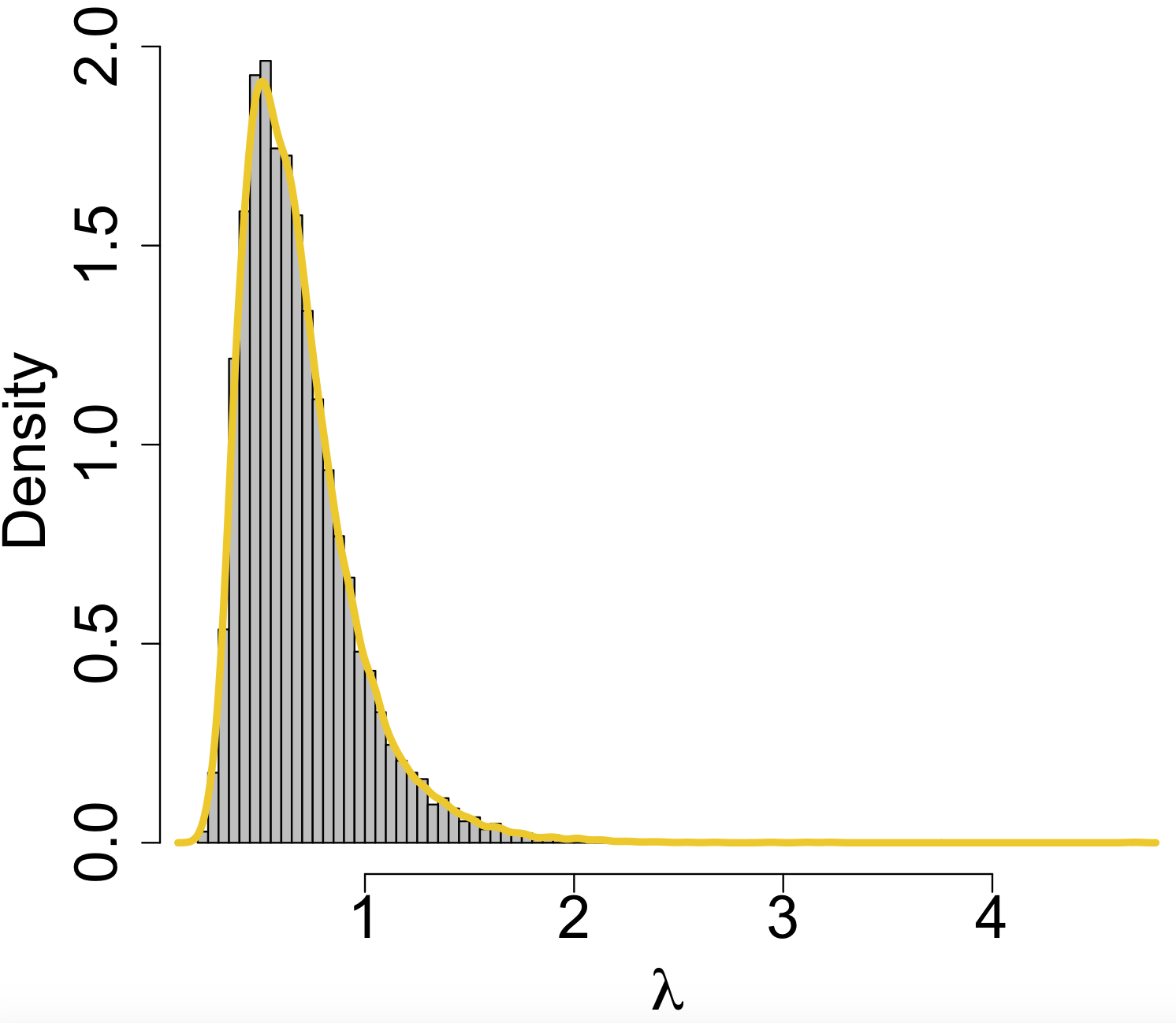}
\includegraphics[width=0.2\textwidth]{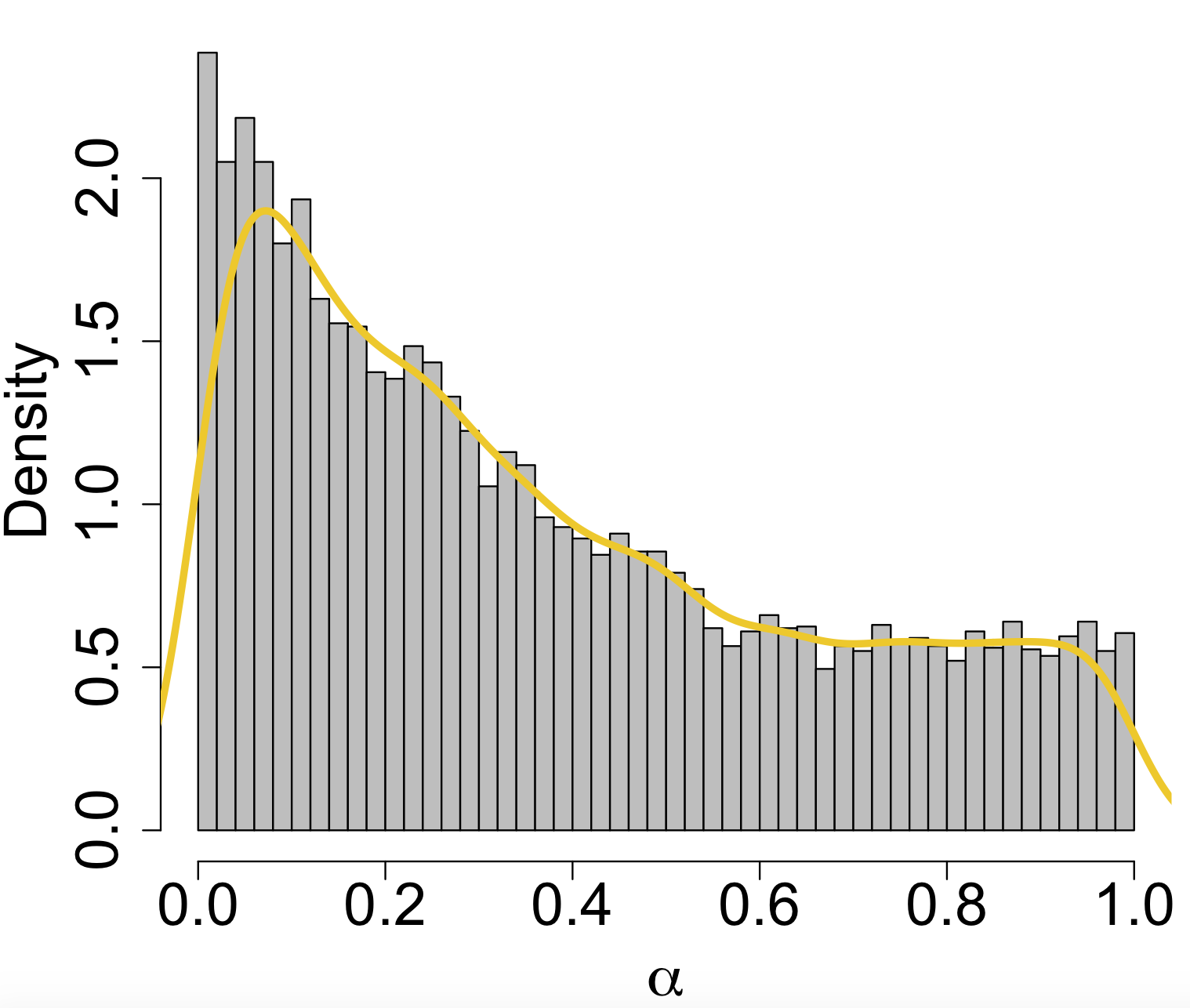}
\includegraphics[width=0.2\textwidth]{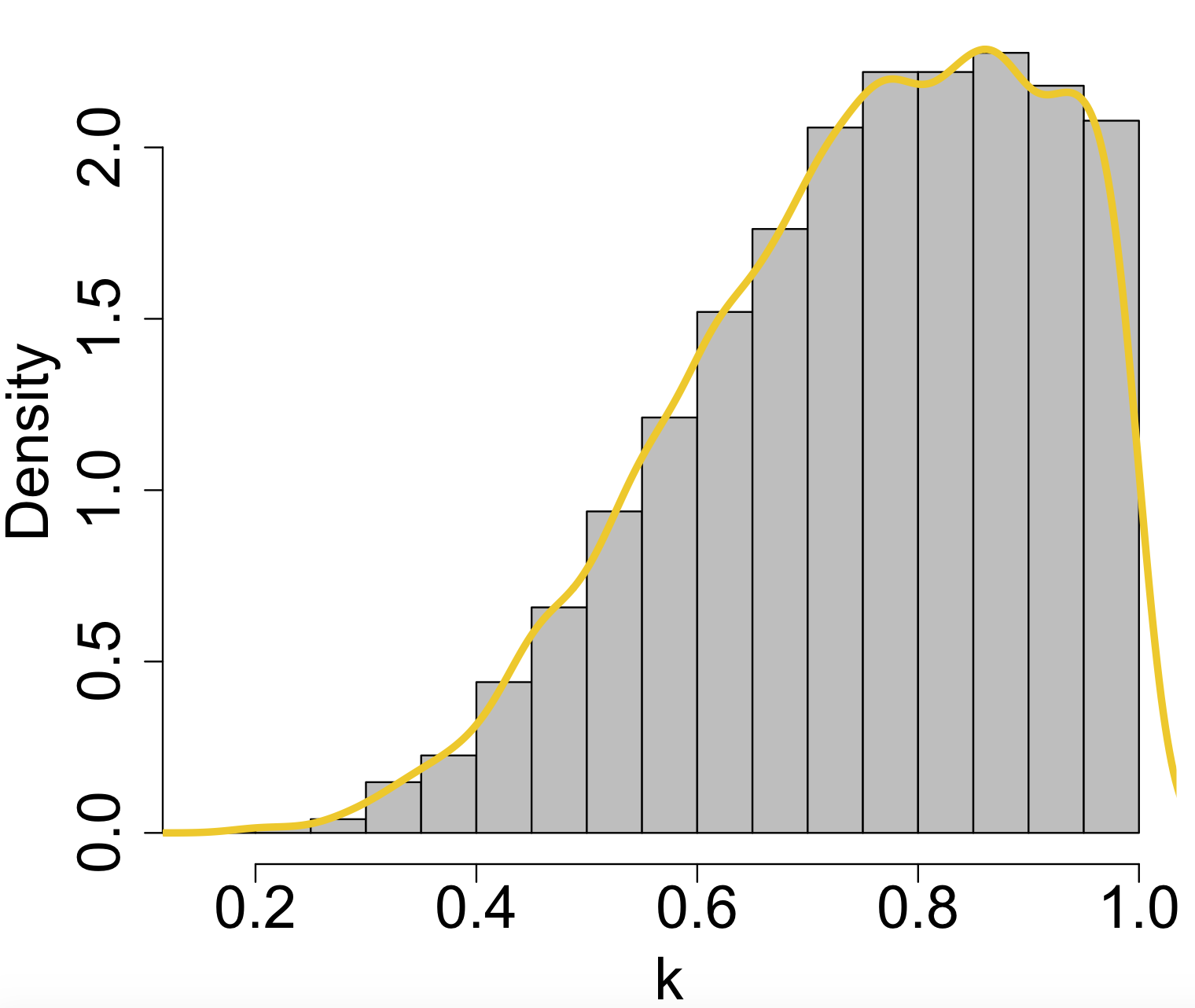}
\includegraphics[width=0.2\textwidth]{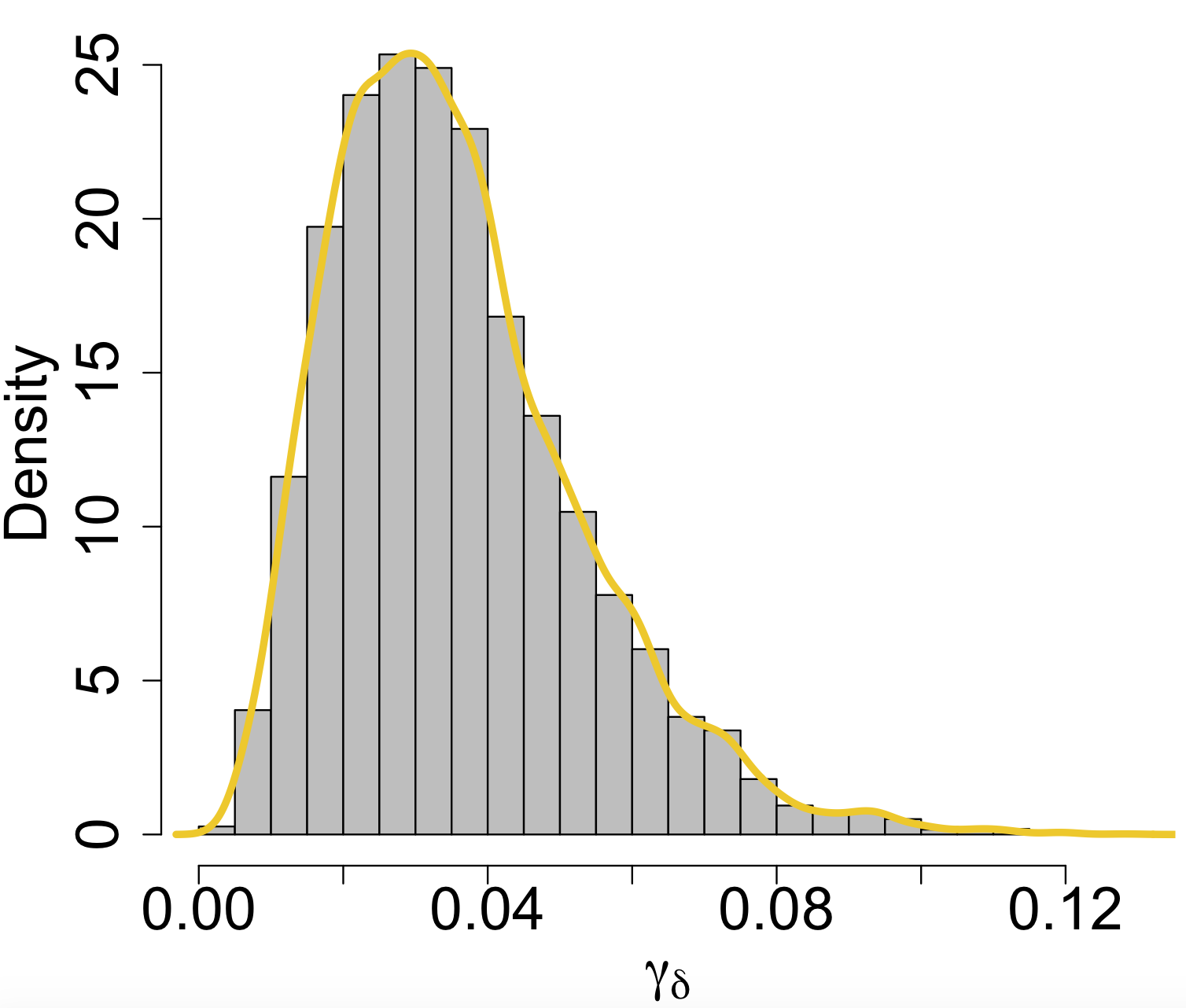}

\caption{\label{fig:posterior_densities} Histogram of marginal posterior density samples for the Marmande (top) and Mas Agenais (bottom) dataset. From left to right and top to bottom: third Strickler coefficient, standard deviation of measurement errors, proportion of observations for which the code is valid, variance ratio between measurement errors and model bias, and correlation length.}
\end{figure}

\paragraph{Bias detection.}

An exciting feature of the mixture model approach is the possiblity to detect a model bias separately for each observation.
Specifically, we can compute the probability of a code bias $\mathbb P[\zeta_i=1|z]$
through Rao-Blackwellization, using (\ref{zeta}), taking the mean of the conditional probabilities: 
\begin{eqnarray*}
\widehat{\mathbb P}_{RB}[\zeta_i=1|z] &=& \frac{1}{T}\sum_{t=1}^T P[\zeta_i=1|\alpha^{(t)},\theta^{(t)},\lambda^{(t)},\delta_i^{(t)},z]\\
&=& \frac{1}{T}\sum_{t=1}^T
\frac
{
(1-\alpha)^{(t)}f_{\MF_1}( y_i | x_i, \delta^{(t-1)}, \pmb{\theta}^{(t)}, \lambda^{(t)}, k^{(t)}, \gamma_\delta^{(t)} )
}
{
\alpha^{(t)}f_{\MF_0}(y_i|x_i,\pmb{\theta}^{(t)},\lambda^{(t)})
+
(1-\alpha)^{(t)}f_{\MF_1}( y_i | x_i, \delta^{(t-1)}, \pmb{\theta}^{(t)}, \lambda^{(t)}, k^{(t)}, \gamma_\delta^{(t)} )
}
\end{eqnarray*}
Results are given in Tables~\ref{tab:nobias_p_Marmande} and \ref{tab:nobias_p_Masagenais}. In both cases, a significant amount of uncertainty remains about the presence or absence of a bias for each observation, as could be expected given the limited amount of available data. However, two interesting features are that: the presence of a model bias is most probable at low flow values, and in Mas Agenais rather than in Marmande, consistent with the estimated proportion of valid predictions in Figure~\ref{fig:posterior_densities}.

\begin{table}[h!]
\begin{tabular}{c|ccccccc}
Observation nb.		& 1		 & 2	  & 3	   & 4		& 5		 & 6	  & 7    \\ 
\hline
Bias probability  &  0.464 & 0.477 & 0.361 & 0.346 & 0.359 & 0.336 & 0.351 \\
\end{tabular}
\caption{\label{tab:nobias_p_Marmande} Probability of a code bias for each observation in Marmande.}
\end{table}

\begin{table}[h!]
\begin{tabular}{c|cccccccc}
Observation nb.		& 1		 & 2	  & 3	   & 4		& 5		 & 6	  & 7&8    \\ 
\hline
Bias probability  &  0.738 & 0.695 & 0.677 & 0.701 & 0.596 & 0.686 & 0.629 & 0.61 \\
\end{tabular}
\caption{\label{tab:nobias_p_Masagenais} Probability of a code bias for each observation in Mas Agenais.}
\end{table}

\paragraph{Pure code vs bias-corrected predictions.}


 We are now interested in comparing predictions from the mixture model~(\ref{eq:6}) with the actual data points. Introducing the latent variables~$\zeta_i$, it boils down to:
 $$
 y_i | \theta, \delta, \zeta_i, \lambda \stackrel{ind}{\sim}\mathcal N\left(g(x_i)\theta+\delta(x_i)\boldsymbol 1_{\{\zeta_i=1\}}, \lambda^2\right)
 $$
 Here we focus on the pure-code predictions $g(x_i)\theta$ and effective biases $\delta(x_i)\boldsymbol 1_{\{\zeta_i=1\}}$; suming the two results in bias-corrected predictions. 
 
Posterior distribution of both pure code vs bias-corrected predictions are shown in Figure~\ref{fig:posterior_vs_data} (left and right, respectively). 
Note that a positive code bias seems apparent at low flow values in both sites (left), though it is swamped in the observation noise, especially in Marmande. Bias-corrected predictions help mitigate this problem, at the cost of a slightly higher dispersion. Also, note that the bias correction is most apparent in Mas Agenais, as could be expected from the posterior probabilities of model biases in Table~\label{tab:nobias_p_Masagenais} and the posterior distributions of valid code predictions in Figure~\ref{fig:posterior_densities}.

 \begin{figure}[h]
 \centering
 \includegraphics[width=0.5\textwidth]{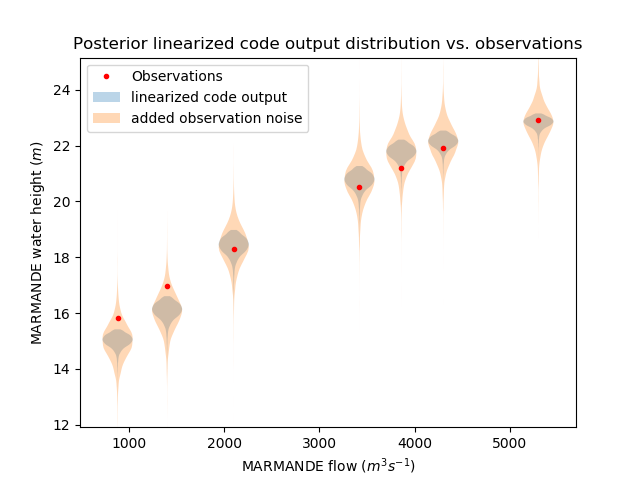}\hfill
 \includegraphics[width=0.5\textwidth]{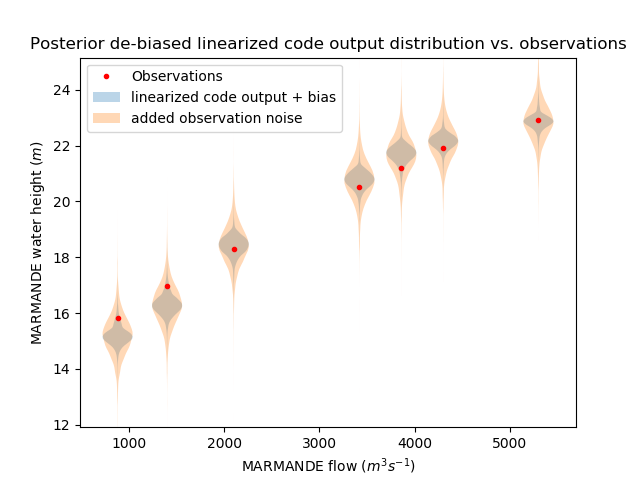}\\
 \includegraphics[width=0.5\textwidth]{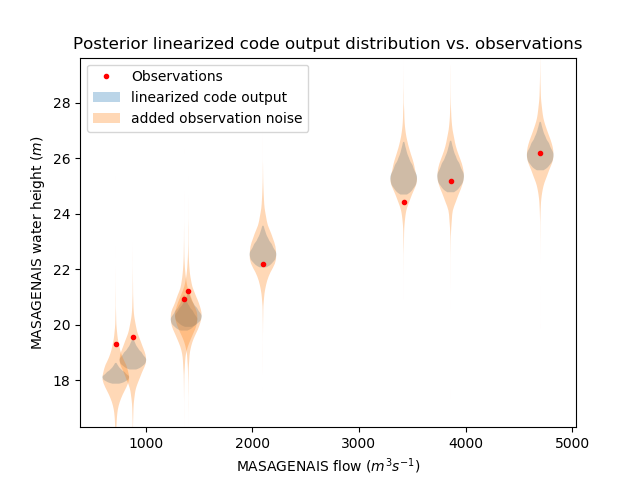}\hfill
 \includegraphics[width=0.5\textwidth]{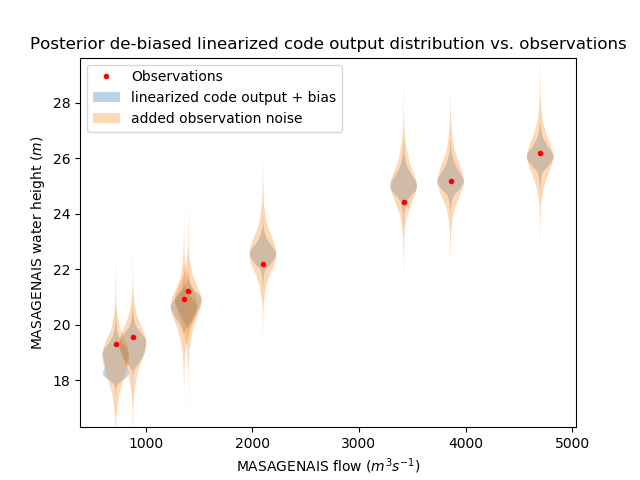}\\
 \caption{\label{fig:posterior_vs_data} Posterior linearized code predictions of water heights vs data, in Marmande (top) and Mas Agenais (bottom), before (left) and after (right) model bias correction.}
 \end{figure}

\section*{Conclusion}
Recently, substantial statistical research has focused on the development of methodology for using detailed numerical codes to carry out inference. The physical systems are modeled by the computer code that requires a set of inputs with some known and specified, others unknown. 
However, lack of enough available field data from the true physical system disables us to inform about the unknown inputs and also to inform about the uncertainty that is associated with a simulation based prediction. In other words, some uncertainties can arise from unknown calibration parameters or limited simulation runs.
Another may be provided by the discrepancy between the simulation code and the actual physical system. A lot of works has focused on overall verification and validation process of the numerical models so as to answer the question of whether or not the simulation model is useful for prediction.
Bayesian inference has been widely used in order to carry out both model calibration and prediction by relying on Gaussian process models to model unknown functions of the model inputs. The posterior distribution of the discrepancy term is therefore considered as an indicator of how well the code is matching reality.
Our focus here is on the validation of the numerical models in the case where the code is a simple linear function of the input and the calibration parameter. To do so, we first rely on \cite{GDMKPBAPEP2016}'s proposition who considered the code validation question as a Bayesian model selection problem. While the pure code and the discrepancy-corrected prediction are the two models under comparison, \cite{KKKMCPRJR2014}'s method has then been used for making decision. By embedding the competing models as the components of an encompassing mixture model, \cite{KKKMCPRJR2014}'s method draws on
the posterior distribution of the weights of both components for deciding about which model is most favoured by the data.   
The Bayesian analysis of the mixture model has been done based on an improper prior for the code parameter, the measurement error that are shared between the competing models and proper priors for other mixture parameters. The estimation is carried out using Markov chain Monte Carlo method. 
We study the sensitivity of the resulting posterior distribution to the prior choice of the correlation length of the model discrepancy Gaussian process prior for various simulated samples. Furthermore the behaviour of the parameter point estimates has been evaluated in the case where the code error is small and simulated based on a Gaussian process centered on the vector $0_n$. 
Because the model selection via mixture estimation model necessitates estimating all parameters of the mixture model, this method allows for the code calibration and validation at once. Our experiences have been that: (a) when the sample is simulated from the pure code, the posterior estimate of $\alpha$ is always close to one in favour of the true model and the posterior distribution of the component-wise parameters converges to the true value. 
(b) when the sample is simulated from discrepancy-corrected prediction, the simulation results illustrated  dependence of the posterior distribution of the mixture weight and the measurement noise on the value of the correlation length priors. 
To be precise, the simulation studies illustrate that the posterior distribution of $\alpha$ strongly supports the true model when the true correlation length $\gamma_\delta$ is equal to or greater than $0.1$ (Figure \ref{figure:m1_1}). For the values of the correlation length very close to zero ($\gamma_\delta=0.01$), the similarity between the calibrated code and the discrepancy-corrected prediction leads $\alpha$ to approach the boundaries of the unit interval meaning that $\alpha$ supports both models for the data in that case. This leads to obtain a posterior estimate shifted to greater value that the true one for the measurement noise (as displayed by Figure \ref{figure:m1_1}).    
On the other hand, the Bayesian estimate of the code parameter $\pmb{\theta}$ is always close to the true value whatever the value of the correlation length is.
A last point about the results is that by increasing the sample size, the posterior estimate of $\alpha$ very quickly converges to one for the true model when the correlation length is $0.3$ ($n>22$ as shown by Figure \ref{figure:3}).

\section*{ACKNOWLEDGMENTS}

We are grateful to Christian Robert for fruitful discussions that we had about this work. This work was supported 
by the research contract $\text{n}^\circ$ 8610-5920026033 between \' Electricit\'e de
France and AgroParisTech. 

\newpage
\bibliographystyle{ims}  
\bibliography{references}

\newpage
\appendix
\label{app}

\section{Proof of Theorem 1}\label{appendix}
\begin{proof}
For $Y=(y_1, y_2, \ldots, y_n)$ and $X=(x_1 \quad x_2 \quad x_n)^T$, the marginal likelihood is defined as 
\scriptsize \begin{equation}\label{eq:18}
m_{\MF_{\alpha}}(Y|X)=\int\prod_{i=1}^n\alpha \nicefrac{\exp\left(-\frac{(y_i-g(x_i)\pmb{\theta})^2}{2\lambda^2} \right)}{(2\pi\lambda^2)^{1/2}}+ (1-\alpha)\nicefrac{\exp\left(-\frac{(y_i-g(x_i)\pmb{\theta}-\delta(x_i))^2}{2\lambda^2}\right)}{(2\pi\lambda^2)^{1/2}}\pi(\pmb{\theta},\lambda,\alpha,\delta,k,\gamma_\delta)\mathrm{d}(\pmb{\theta},\lambda,\alpha,\delta,k,\gamma_\delta)
\end{equation}
\normalsize

The $n-$fold binomial product \scriptsize $\prod_{i=1}^n\alpha \exp\left(-\frac{1}{2\lambda^2}(y_i-g(x_i)\pmb{\theta})^2 \right)+ (1-\alpha)\exp\left(-\frac{1}{2\lambda^2}(y_i-g(x_i)\pmb{\theta}-\delta(x_i))^2\right)$ \normalsize can be expanded into a sum involving terms of the form 
$$\alpha^a (1-\alpha)^b \exp\left(-\frac{1}{2\lambda^2}\left[\sum_{i\in \mathcal{A}\subset \{1, \ldots, n\}}(y_i-g(x_i)\pmb{\theta})^2+\sum_{j\neq i}(y_i-g(x_i)\pmb{\theta}-\delta(x_i))^2\right] \right)$$
where the exponents $a$ and $b$ are nonnegative integers with $a+b = n$. We show here that for any sample size $n>d$, it is possible to expand the $n-$fold binomial product in \eqref{eq:18} into a summation of $2^n$ terms. To do so, let $\mathcal{S}=\{1, 2, \ldots, n \}$ and for $\nu=1, \ldots, n-1$, $\mathcal{A}^{(\nu)}$ be a set of all possible $\nu-$combination of $\mathcal{S}$'s elements. 
In other words, $\mathcal{A}^{(\nu)}$ is supposed to be a set of all possible subsets of $\nu$ distinct elements of $\mathcal{S}$. 
$\mathcal{A}^{(\nu)}$ has a cardinality of $\binom n\nu$ denoted by $|\mathcal{A}^{(\nu)}|$ since the number of all $\nu-$combination of $S$ equals $\binom n\nu$. If $\Omega=\{\mathcal{A}^{(1)}, \mathcal{A}^{(2)}, \ldots, \mathcal{A}^{(n-1)} \}$, then it is easy to show that $|\Omega|=\sum_{\nu=1}^{n-1}|\mathcal{A}^{(\nu)}|=2^n-2$.
We also denote by $\mathcal{A}^{(\nu)}_j$, the $j-$th element of $\mathcal{A}^{(\nu)}$ which is a subset of $S$ with $\nu$ elements. The binomial product in \eqref{eq:18} is then given by

\small \begin{align}\label{eq:19}
\prod_{i=1}^n\alpha& \nicefrac{\exp\left(-\frac{1}{2\lambda^2}(y_i-g(x_i)\pmb{\theta})^2 \right)}{(2\pi\lambda^2)^{1/2}}+ (1-\alpha)\nicefrac{\exp\left(-\frac{1}{2\lambda^2}(y_i-g(x_i)\pmb{\theta}-\delta(x_i))^2\right)}{(2\pi\lambda^2)^{1/2}}\nonumber\\
&=\alpha^n \nicefrac{\exp\left(-\frac{1}{2\lambda^2}\left[\sum_{i=1}^n(y_i-g(x_i)\pmb{\theta})^2\right] \right)}{(2\pi\lambda^2)^{n/2}}\nonumber\\
&+\alpha^{n-1}(1-\alpha)\sum_{j=1}^{\binom n1}\nicefrac{\exp\left(-\frac{1}{2\lambda^2}\left[\sum_{i\in \mathcal{S}-\mathcal{A}^{(1)}_j}(y_i-g(x_i)\pmb{\theta})^2+\sum_{i\in \mathcal{A}^{(1)}_j}(y_i-g(x_i)\pmb{\theta}-\delta(x_i))^2\right] \right)}{(2\pi\lambda^2)^{n/2}}\nonumber\\
&+\alpha^{n-2}(1-\alpha)^2\sum_{j=1}^{\binom n2}\nicefrac{\exp\left(-\frac{1}{2\lambda^2}\left[\sum_{i\in \mathcal{S}-\mathcal{A}^{(2)}_j}(y_i-g(x_i)\pmb{\theta})^2+\sum_{i\in \mathcal{A}^{(2)}_j}(y_i-g(x_i)\pmb{\theta}-\delta(x_i))^2\right] \right)}{(2\pi\lambda^2)^{n/2}}\nonumber\\
&\vdots\nonumber\\
&+\alpha(1-\alpha)^{n-1}\sum_{j=1}^{\binom n{n-1}}\nicefrac{\exp\left(-\frac{1}{2\lambda^2}\left[\sum_{i\in \mathcal{S}-\mathcal{A}^{(n-1)}_j}(y_i-g(x_i)\pmb{\theta})^2+\sum_{i\in \mathcal{A}^{(n-1)}_j}(y_i-g(x_i)\pmb{\theta}-\delta(x_i))^2\right] \right)}{(2\pi\lambda^2)^{n/2}}\nonumber\\
&+(1-\alpha)^n \nicefrac{\exp\left(-\frac{1}{2\lambda^2}\left[\sum_{i=1}^n(y_i-g(x_i)\pmb{\theta}-\delta(x_i))^2\right] \right)}{(2\pi\lambda^2)^{n/2}}
\end{align}
\normalsize
which can be rewritten as
\small \begin{align}\label{eq:20}
\prod_{i=1}^n\alpha& \nicefrac{\exp\left(-\frac{1}{2\lambda^2}(y_i-g(x_i)\pmb{\theta})^2 \right)}{(2\pi\lambda^2)^{1/2}}+ (1-\alpha)\nicefrac{\exp\left(-\frac{1}{2\lambda^2}(y_i-g(x_i)\pmb{\theta}-\delta(x_i))^2\right)}{(2\pi\lambda^2)^{1/2}}\nonumber\\
&=\alpha^n \nicefrac{\exp\left(-\frac{1}{2\lambda^2}\left[\sum_{i=1}^n(y_i-g(x_i)\pmb{\theta})^2\right] \right)}{(2\pi\lambda^2)^{n/2}}\nonumber\\
&+\sum_{l=1}^n\alpha^{n-l}(1-\alpha)^l\sum_{j=1}^{\binom nl}\nicefrac{\exp\left(-\frac{1}{2\lambda^2}\left[\sum_{i\in \mathcal{S}-\mathcal{A}^{(l)}_j}(y_i-g(x_i)\pmb{\theta})^2+\sum_{i\in \mathcal{A}^{(l)}_j}(y_i-g(x_i)\pmb{\theta}-\delta(x_i))^2\right] \right)}{(2\pi\lambda^2)^{n/2}}\nonumber\\
&+(1-\alpha)^n \nicefrac{\exp\left(-\frac{1}{2\lambda^2}\left[\sum_{i=1}^n(y_i-g(x_i)\pmb{\theta}-\delta(x_i))^2\right] \right)}{(2\pi\lambda^2)^{n/2}}.
\end{align}
\normalsize
The marginal likelihood is then computed by marginalizing out the parameters $\alpha, \pmb{\theta}, \lambda, \delta, \gamma_\delta$ and $k$ from \eqref{eq:20}. To do so, we first integrate out the parameter $\delta$, 
\small \begin{align}\label{eq:21}
m_{\MF_{\alpha}}(Y|X)&=\int\int_{\delta}\alpha^n \nicefrac{\exp\left(-\frac{1}{2\lambda^2}\left[\sum_{i=1}^n(y_i-g(x_i)\pmb{\theta})^2\right] \right)}{(2\pi\lambda^2)^{n/2}}\pi(\delta)\mathrm{d}(\delta)\pi(\pmb{\theta}, \lambda, \alpha, k,\gamma_\delta)\mathrm{d}(\pmb{\theta}, \lambda, \alpha, k,\gamma_\delta)\nonumber\\
+\int \int_{\delta}\sum_{l=1}^n&\alpha^{n-l}(1-\alpha)^l\sum_{j=1}^{\binom nl}\nicefrac{\exp\left(-\frac{1}{2\lambda^2}\left[\sum_{i\in \mathcal{S}-\mathcal{A}^{(l)}_j}(y_i-g(x_i)\pmb{\theta})^2+\sum_{i\in \mathcal{A}^{(l)}_j}(y_i-g(x_i)\pmb{\theta}-\delta(x_i))^2\right] \right)}{(2\pi\lambda^2)^{n/2}}\pi(\delta)\mathrm{d}(\delta)\nonumber\\
&\pi(\pmb{\theta}, \lambda, \alpha, k,\gamma_\delta)\mathrm{d}(\pmb{\theta}, \lambda, \alpha, k,\gamma_\delta)\nonumber\\
+\int\int_{\delta}&(1-\alpha)^n \nicefrac{\exp\left(-\frac{1}{2\lambda^2}\left[\sum_{i=1}^n(y_i-g(x_i)\pmb{\theta}-\delta(x_i))^2\right] \right)}{(2\pi\lambda^2)^{n/2}}\pi(\delta)\mathrm{d}(\delta)\pi(\pmb{\theta}, \lambda, \alpha, k,\gamma_\delta)\mathrm{d}(\pmb{\theta}, \lambda, \alpha, k,\gamma_\delta).
\end{align}
\normalsize
For the first term on the right side of equation \eqref{eq:21}, the integral with respect to $\delta$ is removed since $\int_{\delta}\pi(\delta)\mathrm{d}(\delta)=1$. For the second term, let  $n_j^l=|\mathcal{A}^{(l)}_j|$ be the cardinality of $\mathcal{A}^{(l)}_j$. 
We then denote by $Y_{\mathcal{A}^{(l)}_j}$: an $n_j^l-$components column vector of $y_i$ for all $i \in \mathcal{A}^{(l)}_j$. We also suppose that $g(x_i)_{i \in \mathcal{A}^{(l)}_j}$ and $\delta(x_i)_{i \in \mathcal{A}^{(l)}_j}$ indicate an $n_j^l\times d$ matrix of $g(x_i)$s and $n_j^l-$components column vector of $\delta(x_i)$s obtained for all $i \in \mathcal{A}^{(l)}_j$.  Since $\delta(X)$ follows a Gaussian process centered on $n-$components vector of $0$ and covariance matrix $\Sigma_\delta=\nicefrac{1}{k}\lambda^2\text{Corr}_{\gamma_\delta}$, then the marginal distribution over the subset $\{ i \in \mathcal{A}^{(l)}_j\}$ of $\delta(x_i)$ is a Gaussian process with the mean vector: $n_j^l-$components column vector of $0$ and $n_j^l\times n_j^l$ covariance matrix $\Sigma_{\mathcal{A}^{(l)}_j}=\nicefrac{1}{k}\lambda^2\text{Corr}_{\gamma_\delta;\mathcal{A}^{(l)}_j}$ obtained by dropping the irrelevant variables, $\delta(x_i); i \in \{\mathcal{S}-\mathcal{A}^{(l)}_j\}$, from the mean vector and the covariance matrix of $\delta(X)$. Using these notations, we can write 
\small $$\sum_{i\in \mathcal{A}^{(l)}_j}(y_i-g(x_i)\pmb{\theta}-\delta(x_i))^2=(Y_{\mathcal{A}^{(l)}_j}-g(x_i)_{i \in \mathcal{A}^{(l)}_j}\pmb{\theta}-\delta(x_i)_{i \in \mathcal{A}^{(l)}_j})^T(Y_{\mathcal{A}^{(l)}_j}-g(x_i)_{i \in \mathcal{A}^{(l)}_j}\pmb{\theta}-\delta(x_i)_{i \in \mathcal{A}^{(l)}_j})$$
\normalsize
when for any vector or matrix $A$, $A^T$ indicates the transpose of $A$. By replacing $\pi(\delta)$ by the Gaussian process prior density in \eqref{eq:21}, we obtain
\scriptsize \begin{align}\label{eq:22}
m_{\MF_{\alpha}}(Y|X)&=\int\int_{\pmb{\theta}}\alpha^n \nicefrac{\exp\left(-\frac{1}{2\lambda^2}\left[\sum_{i=1}^n(y_i-g(x_i)\pmb{\theta})^2\right] \right)}{(2\pi\lambda^2)^{n/2}}\pi(\pmb{\theta}, \lambda, \alpha, k,\gamma_\delta)\mathrm{d}(\pmb{\theta}, \lambda, \alpha, k,\gamma_\delta)\nonumber\\
+\int \sum_{l=1}^n&\alpha^{n-l}(1-\alpha)^l\sum_{j=1}^{\binom nl}\nicefrac{\exp\left(-\frac{1}{2\lambda^2}\left[\sum_{i\in \mathcal{S}-\mathcal{A}^{(l)}_j}(y_i-g(x_i)\pmb{\theta})^2\right] \right)}{(2\pi\lambda^2)^{n/2}}|\text{Corr}_{\gamma_\delta;\mathcal{A}^{(l)}_j}|^{-1/2}\nonumber\\
\int_{\delta(x_i)_{i \in \mathcal{A}^{(l)}_j}}&\nicefrac{\exp\left(-\frac{1}{2\lambda^2}\left[(Y_{\mathcal{A}^{(l)}_j}-g(x_i)_{i \in \mathcal{A}^{(l)}_j}\pmb{\theta}-\delta(x_i)_{i \in \mathcal{A}^{(l)}_j})^T(Y_{\mathcal{A}^{(l)}_j}-g(x_i)_{i \in \mathcal{A}^{(l)}_j}\pmb{\theta}-\delta(x_i)_{i \in \mathcal{A}^{(l)}_j})+\delta^T(x_i)_{i \in \mathcal{A}^{(l)}_j}k\text{Corr}_{\gamma_\delta;\mathcal{A}^{(l)}_j}^{-1}\delta(x_i)_{i \in \mathcal{A}^{(l)}_j}\right] \right)}{( \nicefrac{2\pi\lambda^2}{k})^{\nicefrac{n_j^l}{2}}}\nonumber\\
&\mathrm{d}(\delta(x_i)_{i \in \mathcal{A}^{(l)}_j})\pi(\pmb{\theta}, \lambda, \alpha, k,\gamma_\delta)\mathrm{d}(\pmb{\theta}, \lambda, \alpha, k,\gamma_\delta)\nonumber\\
+\int\int_{\delta}&(1-\alpha)^n|\text{Corr}_{\gamma_\delta}|^{-1/2} (2\pi\lambda^2)^{-n/2}\nicefrac{\exp\left(-\frac{1}{2\lambda^2}\left[\sum_{i=1}^n(y_i-g(x_i)\pmb{\theta}-\delta(x_i))^2+\delta^T(X)(k\text{Corr}_{\gamma_\delta})^{-1}\delta(X)\right] \right)}{(\nicefrac{2\pi \lambda^2}{k})^{n/2}}\mathrm{d}(\delta)\nonumber\\
&\pi(\pmb{\theta}, \lambda, \alpha, k,\gamma_\delta)\mathrm{d}(\pmb{\theta}, \lambda, \alpha, k,\gamma_\delta).
\end{align}
\normalsize
If we denote \small $\sum_{i=1}^n(y_i-g(x_i)\pmb{\theta}-\delta(x_i))^2=(Y-g(X)\pmb{\theta}-\delta(X))^T(Y-g(X)\pmb{\theta}-\delta(X))$ \normalsize, then the integrals with respect to $\delta$ can be written as follows
\scriptsize \begin{align}\label{eq:23}
\int_{\delta(x_i)_{i \in \mathcal{A}^{(l)}_j}}&\nicefrac{\exp\left(-\frac{1}{2\lambda^2}\left[(Y_{\mathcal{A}^{(l)}_j}-g(x_i)_{i \in \mathcal{A}^{(l)}_j}\pmb{\theta}-\delta(x_i)_{i \in \mathcal{A}^{(l)}_j})^T(Y_{\mathcal{A}^{(l)}_j}-g(x_i)_{i \in \mathcal{A}^{(l)}_j}\pmb{\theta}-\delta(x_i)_{i \in \mathcal{A}^{(l)}_j})+\delta^T(x_i)_{i \in \mathcal{A}^{(l)}_j}k\text{Corr}_{\gamma_\delta;\mathcal{A}^{(l)}_j}^{-1}\delta(x_i)_{i \in \mathcal{A}^{(l)}_j}\right] \right)}{(\nicefrac{2\pi \lambda^2}{k})^{\nicefrac{n_j^l}{2}}}\nonumber\\
&\mathrm{d}(\delta(x_i)_{i \in \mathcal{A}^{(l)}_j})\nonumber\\
&=\int_{\delta(x_i)_{i \in \mathcal{A}^{(l)}_j}}\nicefrac{\exp\left(-\frac{1}{2\lambda^2}\left[
(\delta(x_i)_{i \in \mathcal{A}^{(l)}_j}-\pmb{\mu}_{j,l}^{\delta})^T\pmb{\Sigma}_{j,l}(\delta(x_i)_{i \in \mathcal{A}^{(l)}_j}-\pmb{\mu}_{j,l}^{\delta})+\text{Const}_j^l
\right] \right)}{(\nicefrac{2\pi \lambda^2}{k})^{\nicefrac{n_j^l}{2}}}\mathrm{d}(\delta(x_i)_{i \in \mathcal{A}^{(l)}_j})\nonumber\\
\end{align}
\normalsize
and
\small \begin{align}\label{eq:24}
\int_{\delta}&\nicefrac{\exp\left(-\frac{1}{2\lambda^2}\left[\sum_{i=1}^n(y_i-g(x_i)\pmb{\theta}-\delta(x_i))^2+\delta^T(X)k\text{Corr}_{\gamma_\delta}^{-1}\delta(X)\right] \right)}{(\nicefrac{2\pi \lambda^2}{k})^{n/2}}\mathrm{d}(\delta)\nonumber\\
&=\int_{\delta}\nicefrac{\exp\left(-\frac{1}{2\lambda^2}\left[
(\delta(X)-\pmb{\mu}_{n}^{\delta})^T\pmb{\Sigma}_n(\delta(X)-\pmb{\mu}_{n}^{\delta})+\text{Const}_n
\right] \right)}{(\nicefrac{2\pi \lambda^2}{k})^{n/2}}\mathrm{d}(\delta)
\end{align}
\normalsize
where
\small \begin{align}\label{const:1}
\pmb{\mu}_{j,l}^{\delta}&=\left(\pmb{\Sigma}_{j,l} \right)^{-1}(Y_{\mathcal{A}^{(l)}_j}-g(x_i)_{i \in \mathcal{A}^{(l)}_j}\pmb{\theta})\nonumber\\
\pmb{\Sigma}_{j,l}&=\mathbb{I}_{n_j^l}+k\text{Corr}_{\gamma_\delta;\mathcal{A}^{(l)}_j}^{-1}\nonumber\\
\text{Const}_j^l&=(Y_{\mathcal{A}^{(l)}_j}-g(x_i)_{i \in \mathcal{A}^{(l)}_j}\pmb{\theta})^T
(Y_{\mathcal{A}^{(l)}_j}-g(x_i)_{i \in \mathcal{A}^{(l)}_j}\pmb{\theta})
-(Y_{\mathcal{A}^{(l)}_j}-g(x_i)_{i \in \mathcal{A}^{(l)}_j}\pmb{\theta})^T
\left(\pmb{\Sigma}_{j,l} \right)^{-1}
(Y_{\mathcal{A}^{(l)}_j}-g(x_i)_{i \in \mathcal{A}^{(l)}_j}\pmb{\theta})\nonumber\\
\pmb{\mu}_n^{\delta}&=\left(\pmb{\Sigma}_n \right)^{-1}(Y-g(X)\pmb{\theta})\nonumber\\
\pmb{\Sigma}_n&=\mathbb{I}_n+k\text{Corr}_{\gamma_\delta}^{-1}\nonumber\\
\text{Const}_n&=(Y-g(X)\pmb{\theta})^T(Y-g(X)\pmb{\theta})-(Y-g(X)\pmb{\theta})^T\left(\pmb{\Sigma}_n \right)^{-1}(Y-g(X)\pmb{\theta})
\end{align}
\normalsize
and $\mathbb{I}_{n_j^l}$ indicates an identity matrix of size $n_j^l$. 
In order to remove the integral in \eqref{eq:23} and \eqref{eq:24}, we must prove that $\pmb{\Sigma}_{j,l}$ is positive definite. For every non-zero column vector $\nu$ of $n$ real numbers, if we define $w=\text{Corr}_{\gamma_\delta;\mathcal{A}^{(l)}_j} \nu$, then $$w^T\text{Corr}_{\gamma_\delta;\mathcal{A}^{(l)}_j}^{-1}w=\nu^T\text{Corr}_{\gamma_\delta;\mathcal{A}^{(l)}_j}^T\text{Corr}_{\gamma_\delta;\mathcal{A}^{(l)}_j}^{-1}\text{Corr}_{\gamma_\delta;\mathcal{A}^{(l)}_j}\nu=\nu^T\text{Corr}_{\gamma_\delta;\mathcal{A}^{(l)}_j}^T\nu.$$
Since the correlation matrix 
$\text{Corr}_{\gamma_\delta;\mathcal{A}^{(l)}_j}$ is symmetric and positive definite, so 
$\text{Corr}_{\gamma_\delta;\mathcal{A}^{(l)}_j}^T=\text{Corr}_{\gamma_\delta;\mathcal{A}^{(l)}_j}$ and $w^T\text{Corr}_{\gamma_\delta;\mathcal{A}^{(l)}_j}^{-1}w=\nu^T\text{Corr}_{\gamma_\delta;\mathcal{A}^{(l)}_j}\nu>0$. 
Beside that, for the diagonal matrix 
$\mathbb{I}_{n_j^l}$, we have $\nu^T\mathbb{I}_{n_j^l}\nu=\sum_{i=1}^{n_j^l}\nu_i^2>0$. Because both $\mathbb{I}_{n_j^l}$, $\text{Corr}_{\gamma_\delta;\mathcal{A}^{(l)}_j}^{-1}$ 
are symmetric, and from $k\in (0, 1)$ and 
$\nu^T(\mathbb{I}_{n_j^l}+k \text{Corr}_{\gamma_\delta;\mathcal{A}^{(l)}_j}^{-1})\nu=\nu^T\mathbb{I}_{n_j^l}\nu+\nu^Tk\text{Corr}_{\gamma_\delta;\mathcal{A}^{(l)}_j}^{-1}\nu>0 $, 
we conclude that the matrix $\pmb{\Sigma}_{j,l}$ is therefore symmetric and positive definite. By using the same method, we can easily demonstrate that $\pmb{\Sigma}_n$ is also symmetric and positive definite.
Two integrals \eqref{eq:23} and \eqref{eq:24} are then given by
\scriptsize \begin{align}\label{eq:25}
\int_{\delta(x_i)_{i \in \mathcal{A}^{(l)}_j}}&\nicefrac{\exp\left(-\frac{1}{2\lambda^2}\left[(Y_{\mathcal{A}^{(l)}_j}-g(x_i)_{i \in \mathcal{A}^{(l)}_j}\pmb{\theta}-\delta(x_i)_{i \in \mathcal{A}^{(l)}_j})^T(Y_{\mathcal{A}^{(l)}_j}-g(x_i)_{i \in \mathcal{A}^{(l)}_j}\pmb{\theta}-\delta(x_i)_{i \in \mathcal{A}^{(l)}_j})+\delta^T(x_i)_{i \in \mathcal{A}^{(l)}_j}k\text{Corr}_{\gamma_\delta;\mathcal{A}^{(l)}_j}^{-1}\delta(x_i)_{i \in \mathcal{A}^{(l)}_j}\right] \right)}{(\nicefrac{2\pi \lambda^2}{k})^{n_j^l/2}}\nonumber\\
&\mathrm{d}(\delta(x_i)_{i \in \mathcal{A}^{(l)}_j})\nonumber\\
&=\exp\left(-\frac{1}{2\lambda^2}\left[\text{Const}_j^l\right] \right)|\Sigma_{j,l}|^{-\nicefrac{1}{2}}k^{\nicefrac{n_j^l}{2}}\nonumber\\
\end{align}
\normalsize
and
\small \begin{align}\label{eq:26}
\int_{\delta}&\nicefrac{\exp\left(-\frac{1}{2\lambda^2}\left[\sum_{i=1}^n(y_i-g(x_i)\pmb{\theta}-\delta(x_i))^2+\delta^T(X)k\text{Corr}_{\gamma_\delta}^{-1}\delta(X)\right] \right)}{(\nicefrac{2\pi \lambda^2}{k})^{n/2}}\mathrm{d}(\delta)\nonumber\\
&=\exp\left(-\frac{1}{2\lambda^2}\left[\text{Const}_n\right] \right)|\Sigma_{n}|^{-\nicefrac{1}{2}}k^{\nicefrac{n}{2}}
\end{align}
\normalsize
By replacing two results \eqref{eq:25} and \eqref{eq:26} in \eqref{eq:22}, we obtain
\small \begin{align}\label{eq:27}
m_{\MF_{\alpha}}(Y|X)&=\int\int_{\pmb{\theta}}\alpha^n \nicefrac{\exp\left(-\frac{1}{2\lambda^2}\left[\sum_{i=1}^n(y_i-g(x_i)\pmb{\theta})^2\right] \right)}{(2\pi\lambda^2)^{n/2}}\pi(\pmb{\theta}, \lambda, \alpha, k)\mathrm{d}(\pmb{\theta}, \lambda, \alpha, k)\nonumber\\
&+\int \int_{\pmb{\theta}}\sum_{l=1}^n\alpha^{n-l}(1-\alpha)^l\sum_{j=1}^{\binom nl}\nicefrac{\exp\left(-\frac{1}{2\lambda^2}\left[\sum_{i\in \mathcal{S}-\mathcal{A}^{(l)}_j}(y_i-g(x_i)\pmb{\theta})^2\right] \right)}{(2\pi\lambda^2)^{n/2}}|\text{Corr}_{\gamma_\delta;\mathcal{A}^{(l)}_j}|^{-1/2}\nonumber\\
&\exp\left(-\frac{1}{2\lambda^2}\left[\text{Const}_j^l\right] \right)|\Sigma_{j,l}|^{-\nicefrac{1}{2}}k^{\nicefrac{n_j^l}{2}}\pi(\pmb{\theta}, \lambda, \alpha, k,\gamma_\delta)\mathrm{d}(\pmb{\theta}, \lambda, \alpha, k,\gamma_\delta)\nonumber\\
+\int\int_{\pmb{\theta}}(1-\alpha)^n&|\text{Corr}_{\gamma_\delta}|^{-1/2} (2\pi\lambda^2)^{-n/2}\exp\left(-\frac{1}{2\lambda^2}\left[\text{Const}_n\right] \right)|\Sigma_{n}|^{-\nicefrac{1}{2}}k^{\nicefrac{n}{2}}\pi(\pmb{\theta}, \lambda, \alpha, k,\gamma_\delta)\mathrm{d}(\pmb{\theta}, \lambda, \alpha, k,\gamma_\delta)
\end{align}
\normalsize
which is rewritten as follows
\small \begin{align}\label{eq:28}
m_{\MF_{\alpha}}(Y|X)&=\int\int_{\pmb{\theta}}\alpha^n \nicefrac{\exp\left(-\frac{1}{2\lambda^2}\left[\sum_{i=1}^n(y_i-g(x_i)\pmb{\theta})^2\right] \right)}{(2\pi\lambda^2)^{n/2}}\pi(\pmb{\theta}, \lambda, \alpha, k)\mathrm{d}(\pmb{\theta}, \lambda, \alpha, k)\nonumber\\
&+\int \sum_{l=1}^n\alpha^{n-l}(1-\alpha)^l\sum_{j=1}^{\binom nl}\nicefrac{|\text{Corr}_{\gamma_\delta;\mathcal{A}^{(l)}_j}|^{-1/2}|\Sigma_{j,l}|^{-\nicefrac{1}{2}}k^{\nicefrac{n_j^l}{2}}}{(2\pi\lambda^2)^{n/2}}\nonumber\\
&\int_{\pmb{\theta}}\exp\left(-\frac{1}{2\lambda^2}\left[\sum_{i\in \mathcal{S}-\mathcal{A}^{(l)}_j}(y_i-g(x_i)\pmb{\theta})^2+\text{Const}_j^l\right] \right)\pi(\pmb{\theta}, \lambda, \alpha, k,\gamma_\delta)\mathrm{d}(\pmb{\theta}, \lambda, \alpha, k,\gamma_\delta)\nonumber\\
+\int(1-\alpha)^n&\nicefrac{|\text{Corr}_{\gamma_\delta}|^{-1/2}|\Sigma_{n}|^{-\nicefrac{1}{2}}k^{\nicefrac{n}{2}}}{(2\pi\lambda^2)^{n/2}}\int_{\pmb{\theta}}\exp\left(-\frac{1}{2\lambda^2}\left[\text{Const}_n\right] \right)\pi(\pmb{\theta}, \lambda, \alpha, k,\gamma_\delta)\mathrm{d}(\pmb{\theta}, \lambda, \alpha, k,\gamma_\delta)
\end{align}
\normalsize

Since the prior distribution of $\pmb{\theta}, \lambda$ is $\pi(\pmb{\theta}, \lambda)=\nicefrac{1}{\lambda}$, then the first integral with respect to $\pmb{\theta}$ in \eqref{eq:28} can be written as 
\begin{align}\label{eq:add}
\int_{\pmb{\theta}}& \nicefrac{\exp\left(-\frac{1}{2\lambda^2}\left[\sum_{i=1}^n(y_i-g(x_i)\pmb{\theta})^2\right] \right)}{(2\pi\lambda^2)^{n/2}}\mathrm{d}(\pmb{\theta})\nonumber\\
&=\int_{\pmb{\theta}} \nicefrac{\exp\left(-\frac{1}{2\lambda^2}(\pmb{\theta}-\pmb{\mu}_{(1,n)}^{\pmb{\theta}})^T(g^T(X)g(X))(\pmb{\theta}-\pmb{\mu}_{(1,n)}^{\pmb{\theta}})-\frac{\text{Const}_{(1,n)}}{2\lambda^2} \right)}{(2\pi\lambda^2)^{n/2}}\mathrm{d}(\pmb{\theta})\nonumber\\
&=(2\pi \lambda^2)^{\nicefrac{(d-n)}{2}}|g^T(X)g(X)|^{-\nicefrac{1}{2}}\exp(-\frac{\text{Const}_{(1,n)}}{2\lambda^2}). 
\end{align}
where
\begin{align}\label{cons1n}
\pmb{\mu}_{(1,n)}^{\pmb{\theta}}&=(g^T(X)g(X))^{-1}g^T(X)Y\nonumber\\
\text{Const}_{(1,n)}&=Y^TY-Y^Tg(X)(g^T(X)g(X))^{-1}g^T(X)Y
\end{align}
For the second integral with respect to $\pmb{\theta}$ in \eqref{eq:28}, we first rewrite \small $\sum_{i\in \mathcal{S}-\mathcal{A}^{(l)}_j}(y_i-g(x_i)\pmb{\theta})^2$ \normalsize as \small $(Y_{\mathcal{S}-\mathcal{A}^{(l)}_j}-g(x_i)_{i \in \mathcal{S}-\mathcal{A}^{(l)}_j}\pmb{\theta})^T(Y_{\mathcal{S}-\mathcal{A}^{(l)}_j}-g(x_i)_{i \in \mathcal{S}-\mathcal{A}^{(l)}_j}\pmb{\theta})$ \normalsize or equivalently 
\small \begin{align}
\sum_{i\in \mathcal{S}-\mathcal{A}^{(l)}_j}(y_i-g(x_i)\pmb{\theta})^2&=(\pmb{\theta}-\pmb{\mu}_{(1,j)}^{\pmb{\theta}})^T\left(g^T(x_i)_{i \in \mathcal{S}-\mathcal{A}^{(l)}_j}g(x_i)_{i \in \mathcal{S}-\mathcal{A}^{(l)}_j}\right)(\pmb{\theta}-\pmb{\mu}_{(1,j)}^{\pmb{\theta}})+\text{Const}_{(1,j)}
\end{align} \normalsize
where
\small \begin{align}
\pmb{\mu}_{(1,j)}^{\pmb{\theta}}&=\left(g^T(x_i)_{i \in \mathcal{S}-\mathcal{A}^{(l)}_j}g(x_i)_{i \in \mathcal{S}-\mathcal{A}^{(l)}_j} \right)^{-1}g^T(x_i)_{i \in \mathcal{S}-\mathcal{A}^{(l)}_j}Y_{\mathcal{S}-\mathcal{A}^{(l)}_j}\nonumber\\
\text{Const}_{(1,j)}=Y^T_{\mathcal{S}-\mathcal{A}^{(l)}_j}Y_{\mathcal{S}-\mathcal{A}^{(l)}_j}&-Y^T_{\mathcal{S}-\mathcal{A}^{(l)}_j}g(x_i)_{i \in \mathcal{S}-\mathcal{A}^{(l)}_j}\left(g^T(x_i)_{i \in \mathcal{S}-\mathcal{A}^{(l)}_j}g(x_i)_{i \in \mathcal{S}-\mathcal{A}^{(l)}_j} \right)^{-1}g^T(x_i)_{i \in \mathcal{S}-\mathcal{A}^{(l)}_j}Y_{\mathcal{S}-\mathcal{A}^{(l)}_j}
\end{align} \normalsize
$Y_{\mathcal{S}-\mathcal{A}^{(l)}_j}$ is an $(n-n_j^l)-$components column vector of $y_i$ for all $i \notin \mathcal{A}^{(l)}_j$ and $g(x_i)_{i \in \mathcal{S}-\mathcal{A}^{(l)}_j}$ indicates an $(n-n_j^l)\times d$ matrix of $g(x_i)$s obtained for all $i \notin \mathcal{A}^{(l)}_j$. On the other hand, $\text{Const}_j^l$ in \eqref{eq:28} can be written as follows
\small \begin{align}\label{consjll}
\text{Const}_j^l&=(Y_{\mathcal{A}^{(l)}_j}-g(x_i)_{i \in \mathcal{A}^{(l)}_j}\pmb{\theta})^T\left(\mathbb{I}_{n_j^l}-\left(\pmb{\Sigma}_{j,l} \right)^{-1}\right)
(Y_{\mathcal{A}^{(l)}_j}-g(x_i)_{i \in \mathcal{A}^{(l)}_j}\pmb{\theta})\nonumber\\
&=(\pmb{\theta}-\mu_{(2,j)}^{\pmb{\theta}})^Tg^T(x_i)_{i \in \mathcal{A}^{(l)}_j}\left(\mathbb{I}_{n_j^l}-\left(\pmb{\Sigma}_{j,l} \right)^{-1}\right)g(x_i)_{i \in \mathcal{A}^{(l)}_j}
(\pmb{\theta}-\mu_{(2,j)}^{\pmb{\theta}})+\text{Const}_{(2,j)}
\end{align} \normalsize
where
\scriptsize \begin{align}
\mu_{(2,j)}^{\pmb{\theta}}&=\left[g^T(x_i)_{i \in \mathcal{A}^{(l)}_j}\left(\mathbb{I}_{n_j^l}-\left(\pmb{\Sigma}_{j,l} \right)^{-1}\right)g(x_i)_{i \in \mathcal{A}^{(l)}_j}\right]^{-1}g^T(x_i)_{i \in \mathcal{A}^{(l)}_j}\left(\mathbb{I}_{n_j^l}-\left(\pmb{\Sigma}_{j,l} \right)^{-1}\right)Y_{\mathcal{A}^{(l)}_j}\nonumber\\
\text{Const}_{(2,j)}&=Y^T_{\mathcal{A}^{(l)}_j}\left(\mathbb{I}_{n_j^l}-\left(\pmb{\Sigma}_{j,l} \right)^{-1}\right)Y_{\mathcal{A}^{(l)}_j}\nonumber\\
-Y^T_{\mathcal{A}^{(l)}_j}
&\left(\mathbb{I}_{n_j^l}-\left(\pmb{\Sigma}_{j,l} \right)^{-1}\right)g(x_i)_{i \in \mathcal{A}^{(l)}_j}
\left[g^T(x_i)_{i \in \mathcal{A}^{(l)}_j}\left(\mathbb{I}_{n_j^l}-\left(\pmb{\Sigma}_{j,l} \right)^{-1}\right)g(x_i)_{i \in \mathcal{A}^{(l)}_j}\right]^{-1}
g^T(x_i)_{i \in \mathcal{A}^{(l)}_j}\left(\mathbb{I}_{n_j^l}-\left(\pmb{\Sigma}_{j,l} \right)^{-1}\right)Y_{\mathcal{A}^{(l)}_j}
\end{align}\normalsize
This yields 
\small \begin{align}\label{eq:29}
\int_{\pmb{\theta}}&\exp\left(-\frac{1}{2\lambda^2}\left[\sum_{i\in \mathcal{S}-\mathcal{A}^{(l)}_j}(y_i-g(x_i)\pmb{\theta})^2+\text{Const}_j^l\right] \right)\mathrm{d}(\pmb{\theta})\nonumber\\
&=\int_{\pmb{\theta}}\exp\left(-\nicefrac{\left[\sum_{i\in \mathcal{S}-\mathcal{A}^{(l)}_j}(y_i-g(x_i)\pmb{\theta})^2+(Y_{\mathcal{A}^{(l)}_j}-g(x_i)_{i \in \mathcal{A}^{(l)}_j}\pmb{\theta})^T\left(\mathbb{I}_{n_j^l}-\left(\pmb{\Sigma}_{j,l} \right)^{-1}\right)
(Y_{\mathcal{A}^{(l)}_j}-g(x_i)_{i \in \mathcal{A}^{(l)}_j}\pmb{\theta})
\right]}{2\lambda^2} \right)\nonumber\\%
&=\int_{\pmb{\theta}}\exp\left(-\nicefrac{\left[(\pmb{\theta}-\pmb{\mu}_{(1,j)}^{\pmb{\theta}})^T\left(g^T(x_i)_{i \in \mathcal{S}-\mathcal{A}^{(l)}_j}g(x_i)_{i \in \mathcal{S}-\mathcal{A}^{(l)}_j}\right)(\pmb{\theta}-\pmb{\mu}_{(1,j)}^{\pmb{\theta}})+\text{Const}_{(1,j)}
\right]}{2\lambda^2} \right)\nonumber\\
&\exp\left(-\nicefrac{\left[(\pmb{\theta}-\mu_{(2,j)}^{\pmb{\theta}})^Tg^T(x_i)_{i \in \mathcal{A}^{(l)}_j}\left(\mathbb{I}_{n_j^l}-\left(\pmb{\Sigma}_{j,l} \right)^{-1}\right)g(x_i)_{i \in \mathcal{A}^{(l)}_j}
(\pmb{\theta}-\mu_{(2,j)}^{\pmb{\theta}})+\text{Const}_{(2,j)}
\right]}{2\lambda^2} \right)\mathrm{d}(\pmb{\theta})\nonumber\\
&\leq \{\int_{\pmb{\theta}}\left[\exp\left(-\nicefrac{\left[(\pmb{\theta}-\pmb{\mu}_{(1,j)}^{\pmb{\theta}})^T\left(g^T(x_i)_{i \in \mathcal{S}-\mathcal{A}^{(l)}_j}g(x_i)_{i \in \mathcal{S}-\mathcal{A}^{(l)}_j}\right)(\pmb{\theta}-\pmb{\mu}_{(1,j)}^{\pmb{\theta}})+\text{Const}_{(1,j)}
\right]}{2\lambda^2} \right)\right]^2\mathrm{d}(\pmb{\theta})\}^{\nicefrac{1}{2}}\nonumber\\
&\{
\int_{\pmb{\theta}}\left[\exp\left(-\nicefrac{\left[(\pmb{\theta}-\mu_{(2,j)}^{\pmb{\theta}})^Tg^T(x_i)_{i \in \mathcal{A}^{(l)}_j}\left(\mathbb{I}_{n_j^l}-\left(\pmb{\Sigma}_{j,l} \right)^{-1}\right)g(x_i)_{i \in \mathcal{A}^{(l)}_j}
(\pmb{\theta}-\mu_{(2,j)}^{\pmb{\theta}})+\text{Const}_{(2,j)}
\right]}{2\lambda^2} \right)\right]^2\mathrm{d}(\pmb{\theta})
\}^{\nicefrac{1}{2}}\nonumber\\
&=\{\int_{\pmb{\theta}}\exp\left(-2\nicefrac{\left[(\pmb{\theta}-\pmb{\mu}_{(1,j)}^{\pmb{\theta}})^T\left(g^T(x_i)_{i \in \mathcal{S}-\mathcal{A}^{(l)}_j}g(x_i)_{i \in \mathcal{S}-\mathcal{A}^{(l)}_j}\right)(\pmb{\theta}-\pmb{\mu}_{(1,j)}^{\pmb{\theta}})+\text{Const}_{(1,j)}
\right]}{2\lambda^2} \right)
\mathrm{d}(\pmb{\theta})\}^{\nicefrac{1}{2}}\nonumber\\
&\{
\int_{\pmb{\theta}}\exp\left(-2\nicefrac{\left[(\pmb{\theta}-\mu_{(2,j)}^{\pmb{\theta}})^Tg^T(x_i)_{i \in \mathcal{A}^{(l)}_j}\left(\mathbb{I}_{n_j^l}-\left(\pmb{\Sigma}_{j,l} \right)^{-1}\right)g(x_i)_{i \in \mathcal{A}^{(l)}_j}
(\pmb{\theta}-\mu_{(2,j)}^{\pmb{\theta}})+\text{Const}_{(2,j)}
\right]}{2\lambda^2} \right)
\mathrm{d}(\pmb{\theta})
\}^{\nicefrac{1}{2}}\nonumber\\
\end{align}
\normalsize
when $\text{Const}_j^l$ in \eqref{eq:28} has been replaced by the expression defined in equation \eqref{consjll} and the inequalities hold based on {\em H\"older's inequality}. Since $g(x_i)\neq 0$, then for any non zero column vector $\nu$, 
$$
\nu^Tg^T(x_i)_{\mathcal{S}-\mathcal{A}^{(l)}_j}g(x_i)_{\mathcal{S}-\mathcal{A}^{(l)}_j}\nu=\left(g(x_i)_{\mathcal{S}-\mathcal{A}^{(l)}_j}\nu\right)^Tg(x_i)_{\mathcal{S}-\mathcal{A}^{(l)}_j}\nu
$$
is always positive which means that $g^T(x_i)_{\mathcal{S}-\mathcal{A}^{(l)}_j}g(x_i)_{\mathcal{S}-\mathcal{A}^{(l)}_j}$ is a $d \times d$ positive definite matrix. So the integral 
$$\int_{\pmb{\theta}}\exp\left(-2\nicefrac{\left[(\pmb{\theta}-\pmb{\mu}_{(1,j)}^{\pmb{\theta}})^T\left(g^T(x_i)_{i \in \mathcal{S}-\mathcal{A}^{(l)}_j}g(x_i)_{i \in \mathcal{S}-\mathcal{A}^{(l)}_j}\right)(\pmb{\theta}-\pmb{\mu}_{(1,j)}^{\pmb{\theta}})+\text{Const}_{(1,j)}
\right]}{2\lambda^2} \right)
\mathrm{d}(\pmb{\theta})$$
 is proper and equals $(\pi \lambda^2)^{\nicefrac{d}{2}}|g^T(x_i)_{\mathcal{S}-\mathcal{A}^{(l)}_j}g(x_i)_{\mathcal{S}-\mathcal{A}^{(l)}_j}|^{-\nicefrac{1}{2}}\exp(\nicefrac{-\text{Const}_{(1,j)}}{\lambda^2})$. %

In order to demonstrate the propriety of the second integral in \eqref{eq:29}, we must first check whether the term $g^T(x_i)_{i\in \mathcal{A}^{(l)}_j}\left(\mathbb{I}_{n_j^l}-\left(\pmb{\Sigma}_{j,l} \right)^{-1}\right)g(x_i)_{i\in \mathcal{A}^{(l)}_j}$ is 1) symmetric, 2) positive definite. 
So as to show the matrix symmetry, we compute the matrix transpose and compare it with our matrix. 
As we have demonstrated before, the matrix $\left(\pmb{\Sigma}_{j,l} \right)$ is symmetric positive definite and so is $\left(\pmb{\Sigma}_{j,l} \right)^{-1}$. Since the transpose does not change the diagonal matrix $\mathbb{I}_{n_j^l}$, the following equation  
\begin{align*}
\left(g^T(x_i)_{\mathcal{S}-\mathcal{A}^{(l)}_j}\left(\mathbb{I}_{n_j^l}-\left(\pmb{\Sigma}_{j,l} \right)^{-1}\right)g(x_i)_{\mathcal{S}-\mathcal{A}^{(l)}_j}\right)^T&=g^T(x_i)_{i\in \mathcal{A}^{(l)}_j}\left[\mathbb{I}^T_{n_j^l}-\left(\pmb{\Sigma}^T_{j,l}\right)^{-1}\right]g(x_i)_{i\in \mathcal{A}^{(l)}_j}\\
&=g^T(x_i)_{i\in \mathcal{A}^{(l)}_j}\left[\mathbb{I}_{n_j^l}-
\left(\pmb{\Sigma}_{j,l} \right)^{-1}\right]g(x_i)_{i\in \mathcal{A}^{(l)}_j}
 \end{align*}
leads us to conclude that $g^T(x_i)_{i\in \mathcal{A}^{(l)}_j}\left(\mathbb{I}_{n_j^l}-\left(\pmb{\Sigma}_{j,l} \right)^{-1}\right)g(x_i)_{i\in \mathcal{A}^{(l)}_j}$ is symmetric. For demonstrating that $g^T(x_i)_{i\in \mathcal{A}^{(l)}_j}\left(\mathbb{I}_{n_j^l}-\left(\pmb{\Sigma}_{j,l} \right)^{-1}\right)g(x_i)_{i\in \mathcal{A}^{(l)}_j}$ or equivalently 
$$g^T(x_i)_{i\in \mathcal{A}^{(l)}_j}\left[\mathbb{I}_{n_j^l}-\left(\mathbb{I}_{n_j^l}+k\text{Corr}_{\gamma_\delta;\mathcal{A}^{(l)}_j}^{-1}\right)^{-1}\right]g(x_i)_{i\in \mathcal{A}^{(l)}_j}$$
is positive definite, we must verify whether the related eigenvalues are all positive. Because the correlation matrix $\text{Corr}_{\mathcal{A}^{(l)}_j}$ is positive definite, then there exists an orthogonal matrix $Q$ of size $n_j^l\times n_j^l$ (i.e. $QQ^T=\mathbb{I}_{n_j^l}$) such that $\text{Corr}_{\gamma_\delta;\mathcal{A}^{(l)}_j}=Q\Psi Q^T$ when $\Psi$ is a diagonal matrix whose diagonal elements are the eigenvalues of the correlation matrix. This definition leads to rewrite the term $\text{Corr}_{\gamma_\delta;\mathcal{A}^{(l)}_j}^{-1}$ as follows
\begin{align}\label{ccccr}
\text{Corr}_{\gamma_\delta;\mathcal{A}^{(l)}_j}^{-1}&=(Q\Psi Q^T)^{-1}\nonumber\\
&=Q\Psi^{-1} Q^T
 \end{align}
where the inverse of the diagonal matrix $\Psi=\left[ \psi_{ii}\right]_{i=1}^{n_j^l}$ is obtained by replacing each element in the diagonal with its reciprocal as $\Psi^{-1}=\left[ \nicefrac{1}{\psi_{ii}}\right]_{i=1}^{n_j^l}$. Since $\mathbb{I}_{n_j^l}=Q\mathbb{I}_{n_j^l} Q^T$, from \eqref{ccccr}, we then obtain
\begin{align*}
\mathbb{I}_{n_j^l}-\left(\pmb{\Sigma}_{j,l} \right)^{-1}&=\mathbb{I}_{n_j^l}-\left(\mathbb{I}_{n_j^l}+k\text{Corr}_{\gamma_\delta;\mathcal{A}^{(l)}_j}^{-1}\right)^{-1}\\
&=Q\mathbb{I}_{n_j^l} Q^T-\left(Q\mathbb{I}_{n_j^l} Q^T+Q\left[ \nicefrac{k}{ \psi_{ii}}\right]_{i=1}^{n_j^l}Q^T\right)^{-1}\\
&=Q\mathbb{I}_{n_j^l} Q^T-(Q^T)^{-1}\left(\mathbb{I}_{n_j^l} +\left[ \nicefrac{k}{ \psi_{ii}}\right]_{i=1}^{n_j^l} \right)^{-1}Q^{-1}\\
&=Q\mathbb{I}_{n_j^l} Q^T-Q \left(\left[ 1+\nicefrac{k}{ \psi_{ii}}\right]_{i=1}^{n_j^l} \right)^{-1}Q^T\\
&=Q \left[ 1-\nicefrac{1}{1+\nicefrac{k}{ \psi_{ii}}}\right]_{i=1}^{n_j^l}Q^T\\
&=Q \left[ \nicefrac{k}{(\psi_{ii}+k) }\right]_{i=1}^{n_j^l}Q^T.
 \end{align*}
where $k\in (0, 1)$ and $\psi_{ii}>0$ and consequently $\mathbb{I}_{n_j^l}-\left(\pmb{\Sigma}_{j,l} \right)^{-1}$ has the positive eigenvalues. If we then suppose that the non zero vector $\omega=Q^Tg(x_i)_{i\in \mathcal{A}^{(l)}_j}\nu$, then  
$$\nu^Tg^T(x_i)_{i\in \mathcal{A}^{(l)}_j}\left[\mathbb{I}_{n_j^l}-\left(\mathbb{I}_{n_j^l}+k\text{Corr}_{\gamma_\delta;\mathcal{A}^{(l)}_j}^{-1}\right)^{-1}\right]g(x_i)_{i\in \mathcal{A}^{(l)}_j}\nu=\omega^T\left[ \nicefrac{k}{(\psi_{ii}+k) }\right]_{i=1}^{n_j^l}\omega$$
is always positive and $g^T(x_i)_{i\in \mathcal{A}^{(l)}_j}\left[\mathbb{I}_{n_j^l}-\left(\mathbb{I}_{n_j^l}+k\text{Corr}_{\gamma_\delta;\mathcal{A}^{(l)}_j}^{-1}\right)^{-1}\right]g(x_i)_{i\in \mathcal{A}^{(l)}_j}$ is therefore positive definite. 
The integral 
$$\int_{\pmb{\theta}}\exp\left(-2\nicefrac{\left[(\pmb{\theta}-\mu_{(2,j)}^{\pmb{\theta}})^Tg^T(x_i)_{i \in \mathcal{A}^{(l)}_j}\left(\mathbb{I}_{n_j^l}-\left(\pmb{\Sigma}_{j,l} \right)^{-1}\right)g(x_i)_{i \in \mathcal{A}^{(l)}_j}
(\pmb{\theta}-\mu_{(2,j)}^{\pmb{\theta}})+\text{Const}_{(2,j)}
\right]}{2\lambda^2} \right)
\mathrm{d}(\pmb{\theta})$$
 in inequality \eqref{eq:29} is then proper and equal to \scriptsize $(\pi \lambda^2)^{\nicefrac{d}{2}}|g^T(x_i)_{i \in \mathcal{A}^{(l)}_j}\left(\mathbb{I}_{n_j^l}-\left(\pmb{\Sigma}_{j,l} \right)^{-1}\right)g(x_i)_{i \in \mathcal{A}^{(l)}_j}|^{\nicefrac{-1}{2}}\exp(-\nicefrac{\text{Const}_{(2,j)}}{\lambda^2})$. \normalsize    
\eqref{eq:29} can then be rewritten as
 \scriptsize \begin{align}\label{eq:30}
\int_{\pmb{\theta}}&\exp\left(-\frac{1}{2\lambda^2}\left[\sum_{i\in \mathcal{S}-\mathcal{A}^{(l)}_j}(y_i-g(x_i)\pmb{\theta})^2+\text{Const}_j^l\right] \right)\mathrm{d}(\pmb{\theta})\nonumber\\
&\leq (\pi \lambda^2)^{\nicefrac{d}{2}}|g^T(x_i)_{\mathcal{S}-\mathcal{A}^{(l)}_j}g(x_i)_{\mathcal{S}-\mathcal{A}^{(l)}_j}|^{-\nicefrac{1}{4}}|g^T(x_i)_{i \in \mathcal{A}^{(l)}_j}\left(\mathbb{I}_{n_j^l}-\left(\pmb{\Sigma}_{j,l} \right)^{-1}\right)g(x_i)_{i \in \mathcal{A}^{(l)}_j}|^{\nicefrac{-1}{4}}\exp(-\nicefrac{(\text{Const}_{(1,j)}+\text{Const}_{(2,j)})}{2\lambda^2})
\end{align}
\normalsize
The last integral with respect to $\pmb{\theta}$ in \eqref{eq:28} is given by
 \small \begin{align}\label{eq:31}
\int_{\pmb{\theta}}&\exp\left(-\frac{1}{2\lambda^2}\left[\text{Const}_n\right] \right)\mathrm{d}(\pmb{\theta})\nonumber\\
&=\int_{\pmb{\theta}}\exp\left(-\frac{1}{2\lambda^2}\left[(Y-g(X)\pmb{\theta})^T(Y-g(X)\pmb{\theta})-(Y-g(X)\pmb{\theta})^T\left(\pmb{\Sigma}_n \right)^{-1}(Y-g(X)\pmb{\theta})\right] \right)\mathrm{d}(\pmb{\theta})\nonumber\\
&=\int_{\pmb{\theta}}\exp\left(-\frac{1}{2\lambda^2}\left[(\pmb{\theta}-\pmb{\mu}_{(2,n)}^{\pmb{\theta}})^Tg^T(X)\left(\mathbb{I}_n-\left(\pmb{\Sigma}_n \right)^{-1} \right)g(X)(\pmb{\theta}-\pmb{\mu}_{(2,n)}^{\pmb{\theta}})+\text{Const}_{(2,n)}\right] \right)\mathrm{d}(\pmb{\theta})\nonumber\\
&=(2\pi \lambda^2)^{\nicefrac{d}{2}}|g^T(X)\left(\mathbb{I}_n-\left(\pmb{\Sigma}_n \right)^{-1} \right)g(X)|^{-\nicefrac{1}{2}}\exp(-\nicefrac{\text{Const}_{(2,n)}}{2\lambda^2})
\end{align}
\normalsize
where
\scriptsize \begin{align}
\pmb{\mu}_{(2,n)}^{\pmb{\theta}}&=\left[g^T(X)\left(\mathbb{I}_n-\left(\pmb{\Sigma}_n \right)^{-1} \right)g(X)\right]^{-1}g^T(X)\left(\mathbb{I}_n-\left(\pmb{\Sigma}_n \right)^{-1} \right)Y\nonumber\\
\text{Const}_{(2,n)}&=Y^T\left(\mathbb{I}_n-\left(\pmb{\Sigma}_n \right)^{-1} \right)Y-Y^T\left(\mathbb{I}_n-\left(\pmb{\Sigma}_n \right)^{-1} \right)g(X)\left[g^T(X)\left(\mathbb{I}_n-\left(\pmb{\Sigma}_n \right)^{-1} \right)g(X)\right]^{-1}g^T(X)\left(\mathbb{I}_n-\left(\pmb{\Sigma}_n \right)^{-1} \right)Y
\end{align}\normalsize
Because $\text{Corr}_{\gamma_\delta}^{-1}$ is symmetric positive definite, we can easily prove that 
 $$g^T(X)\left(\mathbb{I}_n-\left(\pmb{\Sigma}_n \right)^{-1} \right)g(X)=g^T(X)\left(\mathbb{I}_n-\left(\mathbb{I}_n+k\text{Corr}_{\gamma_\delta}^{-1} \right)^{-1} \right)g(X)$$
 is also a symmetric positive definite matrix (the proof in the same manner as that for $g^T(x_i)_{i \in \mathcal{A}^{(l)}_j}(\mathbb{I}_{n_j^l}-\left(\pmb{\Sigma}_{j,l} \right)^{-1})g(x_i)_{i \in \mathcal{A}^{(l)}_j}$).
By replacing \eqref{eq:add}, \eqref{eq:30} and \eqref{eq:31} in the marginal likelihood \eqref{eq:28}, we obtain

\small \begin{align}\label{eq:32}
m_{\MF_{\alpha}}(Y|X)\leq \int \alpha^n &(2\pi )^{\nicefrac{(d-n)}{2}}|g^T(X)g(X)|^{-\nicefrac{1}{2}}\int_{\lambda}(\lambda^2)^{\nicefrac{(d-n-1)}{2}}\exp(-\nicefrac{\text{Const}_{(1,n)}}{2\lambda^2})\mathrm{d}(\lambda)\pi(\alpha, k)\mathrm{d}(\alpha, k)\nonumber\\
+\int \sum_{l=1}^n&\alpha^{n-l}(1-\alpha)^l\sum_{j=1}^{\binom nl}\nicefrac{|\text{Corr}_{\gamma_\delta;\mathcal{A}^{(l)}_j}|^{-1/2}|\Sigma_{j,l}|^{-\nicefrac{1}{2}}k^{\nicefrac{n_j^l}{2}}}{(2\pi)^{n/2}}(\pi)^{\nicefrac{d}{2}}|g^T(x_i)_{\mathcal{S}-\mathcal{A}^{(l)}_j}g(x_i)_{\mathcal{S}-\mathcal{A}^{(l)}_j}|^{-\nicefrac{1}{4}}\nonumber\\
|g^T(x_i)_{i \in \mathcal{A}^{(l)}_j}&\left(\mathbb{I}_{n_j^l}-\left(\pmb{\Sigma}_{j,l} \right)^{-1}\right)g(x_i)_{i \in \mathcal{A}^{(l)}_j}|^{\nicefrac{-1}{4}}\nonumber\\
&\int_{\lambda}(\lambda^2)^{\nicefrac{(d-n-1)}{2}}\exp(-\nicefrac{(\text{Const}_{(1,j)}+\text{Const}_{(2,j)})}{2\lambda^2})\mathrm{d}(\lambda)\pi(\alpha, k,\gamma_\delta)\mathrm{d}(\alpha, k,\gamma_\delta)\nonumber\\
+\int(1-\alpha)^n&\nicefrac{|\text{Corr}_{\gamma_\delta}|^{-1/2}|\Sigma_{n}|^{-\nicefrac{1}{2}}k^{\nicefrac{n}{2}}|g^T(X)\left(\mathbb{I}_n-\left(\pmb{\Sigma}_n \right)^{-1} \right)g(X)|^{-\nicefrac{1}{2}}}{(2\pi)^{\nicefrac{(n-d)}{2}}}\nonumber\\
&\int_{\lambda}(\lambda^2)^{\nicefrac{(d-n-1)}{2}}\exp(-\nicefrac{\text{Const}_{(2,n)}}{2\lambda^2})\mathrm{d}(\lambda)\pi(\alpha, k,\gamma_\delta)\mathrm{d}(\alpha, k,\gamma_\delta)
\end{align}
\normalsize
\eqref{eq:32} is integrable with respect to $\lambda$ if the the terms $\text{Const}_{(1,n)}, \text{Const}_{(1,j)}, \text{Const}_{(2,j)}$ and $\text{Const}_{(2,n)}$ are all positive. To do so, we first consider $\text{Const}_{(1,n)}$ defined in \eqref{cons1n} which can be rewritten as 
\begin{align}
\text{Const}_{(1,n)}&=Y^T\left[\mathbb{I}_n-\mathbf{H}_g(X)\right]Y
\end{align}
in which $\mathbf{H}_g(X)=g(X)(g^T(X)g(X))^{-1}g^T(X)$ is a hat matrix or orthogonal projection matrix. From a classical point of view, this is basically because  $\mathbf{H}_g(X)$ projects the vector of observations onto the vector of predictions for linear regression model, $\hat{Y}=\mathbf{H}_g(X)Y$, \citep{DCHREW1978}. 
Among a large number of useful algebraic properties, three facts of $\mathbf{H}_g(X)$ are summarized as follows: First, $\mathbf{H}_g(X)$ is symmetric (as shown before) and idempotent, leading that the diagonal elements $\left[\mathbf{H}_g(X)\right]_{ii}; i=1, \ldots, n$ satisfy $0\leq \left[\mathbf{H}_g(X)\right]_{ii}\leq 1$; Second, since $g(X)$ is a full rank matrix, it is easy to show that $\mathbf{H}_g(X)$ is positive definite with eigendecomposition $\mathbf{H}_g(X)=Q_{\mathbf{H}_g(X)}\Psi_{\mathbf{H}_g(X)}Q^T_{\mathbf{H}_g(X)}$; And third, the eigenvalues of $\mathbf{H}_g(X)$ denoted by $\psi^{\mathbf{H}_g(X)}_{ii}=\left[\Psi_{\mathbf{H}_g(X)}\right]_{ii}$ are either $0$ or $1$ and so $\psi^{\mathbf{H}_g(X)}_{ii}=0$ or $\psi^{\mathbf{H}_g(X)}_{ii}=1$. From the second and third properties, we deduce that       
\begin{align}\label{cons1n2}
\text{Const}_{(1,n)}&=Y^TQ_{\mathbf{H}_g(X)}\left[\mathbb{I}_n-
\Psi_{\mathbf{H}_g(X)}
\right]Q^T_{\mathbf{H}_g(X)}Y\nonumber\\
&=Y^TQ_{\mathbf{H}_g(X)}\left[ 1-\psi^{\mathbf{H}_g(X)}_{ii}\right]_{i=1}^nQ^T_{\mathbf{H}_g(X)}Y\nonumber\\
&=\sum_{i=1}^n \left((1-\psi^{\mathbf{H}_g(X)}_{ii})(\sum_{j=1}^n y_j q_{ji})^2\right)
\end{align}
in which $1-\psi^{\mathbf{H}_g(X)}_{ii}$ are also either $0$ or $1$ and $q_{ji}$ is the $ji-$th element of the matrix $Q_{\mathbf{H}_g(X)}$. Furthermore, from $\text{Const}_{(1,n)}=Y^T\left[Y-\hat{Y}\right]$, the $\text{Const}_{(1,n)}$ is almost never zero since $\mathbb{P}(Y=\hat{Y})=0$, a.s. This leads to conclude that the last equation in \eqref{cons1n2} is obviously positive. In the same manner, we can prove that $\text{Const}_{(1,j)}$ is always positive. 
The $\text{Const}_{(2,n)}$ and $\text{Const}_{(2,j)}$ have the same form and so we just prove that $\text{Const}_{(2,n)}$ is always positive. In order to simplify the formulas, we denote by $\Upsilon(\mathbb{I}_n, \pmb{\Sigma}_n)$, the term $\left(\mathbb{I}_n-\left(\pmb{\Sigma}_n \right)^{-1} \right)$. 
Since the matrix $\Upsilon(\mathbb{I}_n, \pmb{\Sigma}_n)$ is positive definite, we first show that there exists a matrix $\Upsilon(\mathbb{I}_n, \pmb{\Sigma}_n)^{\nicefrac{1}{2}}>0$ such that $\Upsilon(\mathbb{I}_n, \pmb{\Sigma}_n)=\Upsilon(\mathbb{I}_n, \pmb{\Sigma}_n)^{\nicefrac{1}{2}}\Upsilon(\mathbb{I}_n, \pmb{\Sigma}_n)^{\nicefrac{1}{2}}$. If the diagonal matrix of its eigenvalues is denoted by $\Psi_{\Upsilon(\mathbb{I}_n, \pmb{\Sigma}_n)}=\left[\psi_{ii}^{\Upsilon(\mathbb{I}_n, \pmb{\Sigma}_n)} \right]_{i=1}^n$, then the eigendecomposition is given by 
\small \begin{align}\label{cons2n}
\Upsilon(\mathbb{I}_n, \pmb{\Sigma}_n)&=Q_{\Upsilon}\Psi_{\Upsilon(\mathbb{I}_n, \pmb{\Sigma}_n)}Q^T_{\Upsilon};\quad Q_{\Upsilon}Q^T_{\Upsilon}=\mathbb{I}_n; \quad Q^T_{\Upsilon}Q_{\Upsilon}=\mathbb{I}_n\nonumber\\
&=Q_{\Upsilon}\Psi_{\Upsilon(\mathbb{I}_n, \pmb{\Sigma}_n)}^{\nicefrac{1}{2}}Q^T_{\Upsilon}Q_{\Upsilon}\Psi_{\Upsilon(\mathbb{I}_n, \pmb{\Sigma}_n)}^{\nicefrac{1}{2}}Q^T_{\Upsilon}; \quad \left[\Psi_{\Upsilon(\mathbb{I}_n, \pmb{\Sigma}_n)}^{\nicefrac{1}{2}} \right]_{ii}=\sqrt{\psi_{ii}^{\Upsilon(\mathbb{I}_n, \pmb{\Sigma}_n)}}\nonumber\\
&=\Upsilon(\mathbb{I}_n, \pmb{\Sigma}_n)^{\nicefrac{1}{2}}\Upsilon(\mathbb{I}_n, \pmb{\Sigma}_n)^{\nicefrac{1}{2}}
\end{align}\normalsize
From \eqref{cons2n}, we can rewritten $\text{Const}_{(2,n)}$ as follows 
\scriptsize \begin{align}
\text{Const}_{(2,n)}&=Y^T\Upsilon(\mathbb{I}_n, \pmb{\Sigma}_n)^{\nicefrac{1}{2}}\Upsilon(\mathbb{I}_n, \pmb{\Sigma}_n)^{\nicefrac{1}{2}}Y\nonumber\\
-&Y^T\Upsilon(\mathbb{I}_n, \pmb{\Sigma}_n)^{\nicefrac{1}{2}}\Upsilon(\mathbb{I}_n, \pmb{\Sigma}_n)^{\nicefrac{1}{2}}g(X)\left[g^T(X)\Upsilon(\mathbb{I}_n, \pmb{\Sigma}_n)^{\nicefrac{1}{2}}\Upsilon(\mathbb{I}_n, \pmb{\Sigma}_n)^{\nicefrac{1}{2}}g(X)\right]^{-1}g^T(X)\Upsilon(\mathbb{I}_n, \pmb{\Sigma}_n)^{\nicefrac{1}{2}}\Upsilon(\mathbb{I}_n, \pmb{\Sigma}_n)^{\nicefrac{1}{2}}Y\nonumber\\
=&Y^T\Upsilon(\mathbb{I}_n, \pmb{\Sigma}_n)^{\nicefrac{1}{2}}\left[ \mathbb{I}_n-\Upsilon(\mathbb{I}_n, \pmb{\Sigma}_n)^{\nicefrac{1}{2}}g(X)\left[g^T(X)\Upsilon(\mathbb{I}_n, \pmb{\Sigma}_n)^{\nicefrac{1}{2}}\Upsilon(\mathbb{I}_n, \pmb{\Sigma}_n)^{\nicefrac{1}{2}}g(X)\right]^{-1}g^T(X)\Upsilon(\mathbb{I}_n, \pmb{\Sigma}_n)^{\nicefrac{1}{2}}
\right]\Upsilon(\mathbb{I}_n, \pmb{\Sigma}_n)^{\nicefrac{1}{2}}Y.
\end{align}\normalsize
Because \scriptsize $\Upsilon(\mathbb{I}_n, \pmb{\Sigma}_n)^{\nicefrac{1}{2}}g(X)\left[g^T(X)\Upsilon(\mathbb{I}_n, \pmb{\Sigma}_n)^{\nicefrac{1}{2}}\Upsilon(\mathbb{I}_n, \pmb{\Sigma}_n)^{\nicefrac{1}{2}}g(X)\right]^{-1}g^T(X)\Upsilon(\mathbb{I}_n, \pmb{\Sigma}_n)^{\nicefrac{1}{2}}$ \normalsize is an orthogonal projection matrix, so the positiveness of $\text{Const}_{(2,n)}$ is trivially true. The integrals with respect to $\lambda$ are then equal to \small $\nicefrac{2^{(\nicefrac{n-d}{2})-1}\Gamma(\nicefrac{n-d}{2})}{\text{Const}_{(1,n)}^{\nicefrac{n-d}{2}}}$, 
$\nicefrac{2^{(\nicefrac{n-d}{2})-1}\Gamma(\nicefrac{n-d}{2})}{(\text{Const}_{(1,j)}+\text{Const}_{(2,j)})^{\nicefrac{n-d}{2}}}$ and $\nicefrac{2^{(\nicefrac{n-d}{2})-1}\Gamma(\nicefrac{n-d}{2})}{\text{Const}_{(2,n)}^{\nicefrac{n-d}{2}}}$, \normalsize respectively. 

After replacing the prior distribution of $\alpha$ in \eqref{eq:32}, the integrals with respect to $\alpha$ are then equal to $\nicefrac{\Gamma(n+a_0)\Gamma(2a_0)}{\Gamma(a_0)\Gamma(n+2a_0)}$, 
$\nicefrac{\Gamma(n-l+a_0)\Gamma(l+a_0)\Gamma(2a_0)}{\Gamma(n+2a_0)\Gamma(a_0)^2}$ and $\nicefrac{\Gamma(n+a_0)\Gamma(2a_0)}{\Gamma(a_0)\Gamma(n+2a_0)}$, respectively. We then obtain
\small \begin{align}\label{eq:33}
m_{\MF_{\alpha}}&(Y|X)\leq\int_{\gamma_\delta} \nicefrac{\left(\Gamma(n+a_0)\Gamma(2a_0)|g^T(X)g(X)|^{-\nicefrac{1}{2}}\Gamma(\nicefrac{n-d}{2})\right)}{2\left(\Gamma(a_0)\Gamma(n+2a_0)\text{Const}_{(1,n)}^{\nicefrac{n-d}{2}}(\pi )^{\nicefrac{(n-d)}{2}}\right)} \mathrm{d}(\gamma_\delta)
\nonumber\\
+\int_{\gamma_\delta}\sum_{l=1}^n&\sum_{j=1}^{\binom nl}\nicefrac{\left(\Gamma(n-l+a_0)\Gamma(l+a_0)\Gamma(2a_0)|\text{Corr}_{\gamma_\delta;\mathcal{A}^{(l)}_j}|^{-1/2}|\Sigma_{j,l}|^{-\nicefrac{1}{2}}|g^T(x_i)_{\mathcal{S}-\mathcal{A}^{(l)}_j}g(x_i)_{\mathcal{S}-\mathcal{A}^{(l)}_j}|^{-\nicefrac{1}{4}}\right)}{\left(\Gamma(n+2a_0)\Gamma(a_0)^2 2^{d/2+1}\pi^{\nicefrac{(n-d)}{2}}\right)}
\nonumber\\
&\nicefrac{\left(|g^T(x_i)_{i \in \mathcal{A}^{(l)}_j}\left(\mathbb{I}_{n_j^l}-\left(\pmb{\Sigma}_{j,l} \right)^{-1}\right)g(x_i)_{i \in \mathcal{A}^{(l)}_j}|^{\nicefrac{-1}{4}} \Gamma(\nicefrac{n-d}{2})\right)}{\left((\text{Const}_{(1,j)}+\text{Const}_{(2,j)})^{\nicefrac{n-d}{2}}\right)}\int k^{\nicefrac{n_j^l}{2}}\pi(k)\mathrm{d}(k)\mathrm{d}(\gamma_\delta)\nonumber\\
+&\int_{\gamma_\delta}\nicefrac{\left(|\text{Corr}_{\gamma_\delta}|^{-1/2}|\Sigma_{n}|^{-\nicefrac{1}{2}}\Gamma(n+a_0)\Gamma(2a_0)|g^T(X)\left(\mathbb{I}_n-\left(\pmb{\Sigma}_n \right)^{-1} \right)g(X)|^{-\nicefrac{1}{2}}\Gamma(\nicefrac{n-d}{2})\right)}{\left(2(\pi)^{\nicefrac{(n-d)}{2}}\Gamma(a_0)\Gamma(n+2a_0)\text{Const}_{(2,n)}^{\nicefrac{n-d}{2}}\right)}\mathrm{d}(\gamma_\delta)\nonumber\\
&\int k^{\nicefrac{n}{2}}\pi(k)\mathrm{d}(k)
\end{align}
\normalsize

Since the prior distribution of $k$, $\gamma_\delta$ are supposed to be proper and $\gamma_\delta\in [0,1)$, the integrals with respect to $k$ and then to $\gamma_\delta$ are proper. This means that the marginal likelihood $m_{\MF_{\alpha}}(Y|X)$ is therefore finite and the posterior distribution of the mixture model parameters is consequently proper. 
\end{proof}   

\end{document}